\documentclass[submission,copyright,creativecommons]{eptcs}
\providecommand{\event}{QPL 2024} %

\usepackage{iftex}
\usepackage{amsmath}
\usepackage{tikz/tikzit}
\usepackage{cite}
\usepackage{comment}
\usepackage{amsfonts}
\usepackage{mathtools}
\usepackage{cleveref}
\usepackage{booktabs}
\usepackage{makecell}
\usepackage{siunitx}
\usepackage{enumitem}
\usepackage{threeparttable}
\usepackage[makeroom]{cancel}
\usepackage[linesnumbered,ruled,vlined]{algorithm2e}
\ifpdf
  \usepackage{underscore}         %
  \usepackage[T1]{fontenc}        %
\else
  \usepackage{breakurl}           %
\fi

\usepackage{amsthm}
\theoremstyle{plain}
\newtheorem{corollary}{Corollary}
\newtheorem{lemma}{Lemma}
\newtheorem{theorem}{Theorem}

\tikzstyle{none-small}=[fill=none, draw=none, shape=circle, tikzit category=misc, tikzit shape=circle, tikzit fill=none, font={\footnotesize}]
\tikzstyle{none-small-gray}=[fill=none, draw=none, shape=circle, text=gray, tikzit category=misc, tikzit shape=circle, tikzit fill=none, font={\footnotesize}]
\tikzstyle{gate}=[shape=rectangle, text height=1ex, text depth=0.25ex, yshift=0.5mm, fill=white, draw=black, minimum height=3mm, yshift=-0.5mm, minimum width=3mm, font={\footnotesize}, tikzit category=circuit]
\tikzstyle{meter}=[shape=rectangle, text height=1ex, text depth=0.25ex, yshift=0.5mm, fill=white, draw=black, minimum height=3mm, yshift=-0.5mm, minimum width=3mm, font={\footnotesize}, tikzit category=circuit, text width=4.5mm, label={{[shift={(0,-1.15)}]\metersymb}}]
\tikzstyle{big gate}=[shape=rectangle, text height=1.5ex, text depth=0.25ex, yshift=0.5mm, fill=white, draw=black, minimum height=10mm, yshift=-0.5mm, minimum width=5mm, font={\normalsize}, tikzit category=circuit]
\tikzstyle{long gate}=[shape=rectangle, text height=1ex, text depth=0.25ex, yshift=0.5mm, fill=white, draw=black, minimum height=3mm, yshift=-0.5mm, minimum width=5mm, font={\footnotesize}, tikzit category=circuit]
\tikzstyle{Z dot}=[inner sep=0mm, minimum size=2mm, shape=circle, draw=black, fill={rgb,255: red,221; green,255; blue,221}, tikzit category=zx]
\tikzstyle{Z phase dot}=[minimum size=5mm, font={\footnotesize\boldmath}, shape=rectangle, rounded corners=2mm, inner sep=1mm, outer sep=-2mm, scale=0.8, tikzit shape=circle, draw=black, fill={rgb,255: red,221; green,255; blue,221}, tikzit draw=blue, tikzit category=zx]
\tikzstyle{X dot}=[Z dot, shape=circle, draw=black, fill={rgb,255: red,255; green,136; blue,136}, tikzit category=zx]
\tikzstyle{X phase dot}=[Z phase dot, tikzit shape=circle, tikzit draw=blue, fill={rgb,255: red,255; green,136; blue,136}, font={\footnotesize\boldmath}, tikzit category=zx]
\tikzstyle{hadamard}=[fill=yellow, draw=black, shape=rectangle, inner sep=0.6mm, minimum height=1.5mm, minimum width=1.5mm, tikzit category=zx]
\tikzstyle{paulibox}=[fill={rgb,255: red,221; green,221; blue,255}, draw=black, shape=rectangle, inner sep=0.6mm, minimum height=5mm, minimum width=5mm, font={\footnotesize}, text height=1.5ex, text depth=0.25ex, tikzit category=zx]
\tikzstyle{vertex}=[inner sep=0mm, minimum size=1mm, shape=circle, draw=black, fill=black, tikzit category=misc]
\tikzstyle{vertex set}=[inner sep=0mm, minimum size=1mm, shape=circle, draw=black, fill=white, font={\footnotesize\boldmath}, tikzit category=misc]
\tikzstyle{small black dot}=[fill=black, draw=black, shape=circle, inner sep=0pt, minimum width=1.2mm, tikzit category=circuit]
\tikzstyle{cnot ctrl}=[fill=black, draw=black, shape=circle, inner sep=0pt, minimum width=1.2mm, tikzit category=circuit]
\tikzstyle{cnot targ}=[fill=white, draw=white, shape=circle, tikzit category=circuit, label={center:$\oplus$}, inner sep=0pt, minimum width=2.1mm, tikzit fill={rgb,255: red,102; green,204; blue,255}, tikzit draw=black]
\tikzstyle{ket}=[fill=white, draw=black, shape=regular polygon, regular polygon sides=3, regular polygon rotate=-30, scale=0.7, inner sep=1pt, tikzit category=circuit, tikzit shape=rectangle, tikzit fill=green]
\tikzstyle{bra}=[fill=white, draw=black, shape=regular polygon, regular polygon sides=3, regular polygon rotate=30, scale=0.7, inner sep=1pt, tikzit category=circuit, tikzit shape=rectangle, tikzit fill=red]
\tikzstyle{scalar}=[shape=rectangle, text height=1.5ex, text depth=0.25ex, yshift=0.5mm, fill=white, draw=black, minimum height=5mm, yshift=-0.5mm, minimum width=5mm, font={\normalsize}]
\tikzstyle{clabel}=[fill=white, draw=none, shape=rectangle, tikzit fill={rgb,255: red,56; green,255; blue,242}, font={\footnotesize}, inner sep=1pt, tikzit category=labels]
\tikzstyle{empty diagram}=[draw={gray!40!white}, dashed, shape=rectangle, minimum width=1cm, minimum height=1cm, tikzit category=misc]
\tikzstyle{cluster small}=[fill=none, thick, draw={rgb,255: red,0; green,128; blue,128}, shape=circle, tikzit category=misc, tikzit shape=circle, minimum size=1.5mm,  inner sep=0.3mm, tikzit fill=white, tikzit draw={rgb,255: red,0; green,128; blue,128}, font={\footnotesize}]
\tikzstyle{cluster}=[fill=none, thick, draw={rgb,255: red,0; green,128; blue,128}, shape=circle, tikzit category=misc, tikzit shape=circle, minimum size=3.5mm, inner sep=0pt, tikzit fill=white, tikzit draw={rgb,255: red,0; green,128; blue,128}, font={\footnotesize}]
\tikzstyle{cluster big}=[fill=none, thick, draw={rgb,255: red,0; green,128; blue,128}, shape=circle, tikzit category=misc, tikzit shape=circle, minimum size=4.5mm, text width=2mm, inner sep=0pt, tikzit fill=white, tikzit draw={rgb,255: red,0; green,128; blue,128}, font={\footnotesize}]
\tikzstyle{trapez}=[trapezium angle=-85, minimum width=15mm, trapezium stretches = true, fill={rgb,255: red,191; green,191; blue,191}, draw=black, shape=trapezium, tikzit category=zx]

\tikzstyle{hadamard edge}=[-, dashed, dash pattern=on 2pt off 0.5pt, thick, draw={rgb,255: red,68; green,136; blue,255}]
\tikzstyle{box edge}=[-, dashed, dash pattern=on 2pt off 0.5pt, thick, draw={rgb,255: red,203; green,192; blue,225}]
\tikzstyle{brace edge}=[-, tikzit draw=blue, decorate, decoration={brace,amplitude=1mm,raise=-1mm}]
\tikzstyle{diredge}=[->]
\tikzstyle{double edge}=[-, double, shorten <=-1mm, shorten >=-1mm, double distance=2pt]
\tikzstyle{gray edge}=[-, {gray!70!white}, thick]
\tikzstyle{pointer edge}=[->, very thick, gray]
\tikzstyle{boldedge}=[-, line width=1.2pt, shorten <=-0.17mm, shorten >=-0.17mm]
\tikzstyle{boldedge red}=[-, line width=1.4pt, shorten <=-0.17mm, shorten >=-0.17mm, draw=red, tikzit draw=red]

\newcommand{\lmuPh}{Fakultät für Physik,  Ludwig-Maximilians-Universität München,  80799 Munich, Germany}
\newcommand{\lmuCS}{MNM-Team,  Ludwig-Maximilians-Universität München, 80538 Munich, Germany}
\newcommand{\mpq}{Max-Planck-Institut für Quantenoptik,  85748 Garching, Germany}
\newcommand{\tum}{Chair for Design Automation,  Technical University of Munich, 80333 Munich, Germany}
\newcommand{\mcqst}{Munich Center for Quantum Science  and Technology (MCQST),  80799 Munich, Germany}
\newcommand{\lrz}{Leibniz Supercomputing Centre (LRZ), 85748 Garching, Germany}
\newcommand{\scch}{Software Competence Center Hagenberg GmbH (SCCH), 4232 Hagenberg im Mühlkreis, Austria}

\title{Multi-controlled Phase Gate Synthesis with ZX-calculus \\ applied to Neutral Atom Hardware \vspace{-10pt}}

\author{Korbinian Staudacher$^{\,1\,\ast}$ \quad Ludwig Schmid$^{\,2}$ \quad Johannes Zeiher$^{\,3,4,5}$ \\ Robert Wille$^{\,2,6}$ \quad Dieter Kranzlmüller$^{\,1,7}$
 \email{${}^\ast$Korbinian.Staudacher@nm.ifi.lmu.de} \institute{$^{1}$\lmuCS \\ $^{2}$\tum \\ $^{3}$\lmuPh \\ $^{4}$\mpq \\ $^{5}$\mcqst  \\ $^{6}$\scch \\ $^{7}$\lrz}
}

\def\titlerunning{Multi-controlled Phase Gate Synthesis with ZX-calculus applied to Neutral Atom Hardware}
\def\authorrunning{K. Staudacher, L. Schmid, J. Zeiher, R. Wille, D. Kranzlmüller}

\begin{document}
\date{\today}
\maketitle

 \vspace{-10pt}
\begin{abstract}
  Quantum circuit synthesis describes the process of converting arbitrary unitary operations into a gate sequence of a fixed universal gate set, usually defined by the operations native to a given hardware platform.
  Most current synthesis algorithms are designed to synthesize towards a set of single-qubit rotations and an additional entangling two-qubit gate, such as CX, CZ, or the Mølmer–Sørensen gate. 
  However, with the emergence of neutral atom-based hardware and their native support for gates with more than two qubits, synthesis approaches tailored to these new gate sets become necessary.
  In this work, we present an approach to synthesize (multi-) controlled phase gates using ZX-calculus.
  By representing quantum circuits as graph-like ZX-diagrams, one can utilize the distinct graph structure of diagonal gates to identify multi-controlled phase gates inherently present in some quantum circuits even if none were explicitly defined in the original circuit.
  We evaluate the approach on a wide range of benchmark circuits and compare them to the standard Qiskit synthesis regarding its circuit execution time for neutral atom-based hardware with native support of multi-controlled gates.
  Our results show possible advantages for current state-of-the-art hardware and represent the first exact synthesis algorithm supporting arbitrary-sized multi-controlled phase gates.

\end{abstract}

\section{Introduction}
\label{sec:intro}
Compiling and optimizing quantum algorithms towards hardware-specific constraints is indispensable to efficiently use currently available noisy quantum hardware with limited gate fidelities and coherence times.
An important step of the compilation process is quantum circuit synthesis, converting arbitrary unitary operations to gate sequences natively supported by the hardware.
State-of-the-art synthesis algorithms, such as~\cite{shende2006synthesis,kliuchnikov2013fast}, are often focused on a superconducting hardware setting and synthesize towards singular two-qubit gates, e.g., CX, and single-qubit gates.
Such synthesis algorithms are less preferable for other hardware architectures, for instance, when gates acting on three or more qubits can be executed natively without decomposition, resulting in a reduced execution cost.

In this work, we propose an approach to synthesize quantum circuits towards single qubit gates and arbitrary-sized multi-controlled phase gates C${}_n$P$(\varphi)$.
To this end, we make use of the representation of a quantum circuit as a graph-like ZX-diagram where we can use powerful rewrite rules of the ZX-calculus to simplify diagrammatic structures~\cite{duncan-graph-theoretic-2020}.
This approach has shown to be a useful tool for tasks like hardware-agnostic circuit optimization or equivalence checking~\cite{kissingerReducingNumberNonClifford2020,staudacher2022reducing,peham2022equivalence}.
We show that C${}_n$P$(\varphi)$ have a distinct representation in graph-like ZX-diagrams as a combination of so-called phase gadgets, which occur naturally in the diagrams when using a simplification strategy proposed in~\cite{kissingerReducingNumberNonClifford2020}.
By modifying an existing extraction algorithm from~\cite{backens-there-2021} to translate graph-like diagrams back to quantum circuits, we can specifically optimize towards extracting phase gadget combinations corresponding to C${}_n$P$(\varphi)$ gates.

The benefit and potential of the resulting approach are shown by synthesizing gate functionality for the recently emerging neutral atoms platforms~\cite{saffmanQuantumInformationRydberg2010,saffmanQuantumComputingNeutral2019,henrietQuantumComputingNeutral2020,grahamMultiqubitEntanglementAlgorithms2022,morgado2021quantum,saffman2016quantum,schmidComputationalCapabilitiesCompiler2023}. Besides dynamic connectivity with atom rearrangements~\cite{bluvsteinQuantumProcessorBased2022,bluvsteinLogicalQuantumProcessor2023,shawMultiensembleMetrologyProgramming2024} and favorable properties regarding scalability and large-scale control~\cite{barredoAtombyatomAssemblerDefectfree2016,pauseSuperchargedTwodimensionalTweezer2023,bluvsteinQuantumProcessorBased2022,norciaIterativeAssembly1712024,gygerContinuousOperationLargescale2024}, this technology offers native support for multi-controlled gates such as $\mathrm{C}_{n}\mathrm{P}$, and $\mathrm{CZ}_{n}$~\cite{mullerMesoscopicRydbergGate2009,isenhowerMultibitCkNOTQuantum2011,dlaskaQuantumOptimizationFourBody2022,everedHighfidelityParallelEntangling2023,Jandura2022timeoptimaltwothree,cao2024multiqubit}.
We integrate our extraction scheme into a full gate synthesis and optimization process and compare total execution times on  hardware against Qiskit synthesis routines, considering current state-of-the-art parameters.
Our results show promising advantages in the form of reduced execution times on different benchmark circuits.

The paper is structured as follows: In the first part, we give a basic introduction to ZX-calculus, including graph-like ZX-diagram simplification, and show how multi-controlled phase gates can be identified and extracted from the diagrams.
In the second part, we focus on the application of the proposed approach to neutral-atom-specific gate synthesis and discuss its effect on the execution time. 

\section{Related work}
So far, algorithms supporting the synthesis of gates acting on more than two qubits are mostly centered around the generation of Toffoli gates. Ref.~\cite{grosse2009exact} introduces a synthesis algorithm for classical logic reversible functions using multi-control Toffoli gates and there exist algorithms for synthesizing towards universal Toffoli gate sets~\cite{amyImprovedSynthesisToffoliHadamard2023}, even with optimal numbers of Toffoli gates~\cite{mukhopadhyay2024synthesizing}. 
The synthesis of multi-controlled phase gates is less studied. Ref.~\cite{zhang2023characterization} proposes an optimal synthesis algorithm, but restricted to diagonal unitaries as an input. For universal circuits, a recent framework for neutral atom systems~\cite{patelGeyserCompilationFramework2022} is able to synthesize circuits with $CCZ$ gates. However, the synthesis process is based on non-exact numerical optimization procedures and does not consider more than three-qubit gates or arbitrary rotations.

\section{Preliminaries}
\label{sec:background}
In this section we introduce the ZX-calculus fundamentals and describe how graph-like diagrams can be simplified and extracted to quantum circuits, which represents the basis of our synthesis approach.
We only give a brief overview of ZX-calculus, for a more detailed introduction we refer to~\cite{van2020zx,coecke2017picturing,duncan-graph-theoretic-2020}.

\subsection{ZX-calculus}
ZX-calculus is a diagrammatic language for reasoning about linear maps in quantum computing where nodes (spiders) and edges (wires) form an undirected graph called ZX-diagram. There are two types of spiders: The green Z-spiders and the red X-spiders. Spiders can be parametrized with an angle $\alpha \in [0, 2\pi)$ and correspond to two-dimensional matrices in Hilbert space: 
\vspace{-3em}
\[
\begin{tikzpicture}
	\begin{pgfonlayer}{nodelayer}
		\node [style=Z phase dot] (0) at (-7.5, 0) {$\alpha$};
		\node [style=none] (1) at (-6.5, 0.75) {};
		\node [style=none] (2) at (-6.5, -0.75) {};
		\node [style=none] (3) at (-8.5, -0.75) {};
		\node [style=none] (4) at (-8.5, 0.75) {};
		\node [style=none] (5) at (-8.65, 0.25) {$\vdots$};
		\node [style=none] (6) at (-6.35, 0.25) {$\vdots$};
		\node [style=X phase dot] (7) at (2.5, 0) {$\alpha$};
		\node [style=none] (8) at (3.5, 0.75) {};
		\node [style=none] (9) at (3.5, -0.75) {};
		\node [style=none] (10) at (1.5, -0.75) {};
		\node [style=none] (11) at (1.5, 0.75) {};
		\node [style=none] (12) at (1.6, 0.25) {$\vdots$};
		\node [style=none] (13) at (3.4, 0.25) {$\vdots$};
		\node [style=none-small] (18) at (-3.25, 0.75) {$:=\ket{0}^{\otimes m}\!\bra{0}^{\otimes n}$};
		\node [style=none-small] (23) at (-2.3, -0.25) {\footnotesize$+ e^{i\alpha}\ket{1}^{\otimes m}\!\bra{1}^{\otimes n}$};
		\node [style=none-small] (24) at (6.75, 0.75) {$:=\ket{+}^{\otimes m}\!\bra{+}^{\otimes n}$};
		\node [style=none-small] (25) at (7.95, -0.25) {$+ e^{i\alpha}\ket{-}^{\otimes m}\!\bra{-}^{\otimes n}$};
		\node [style=none-small] (26) at (-9.25, 0) {$m$};
		\node [style=none-small] (27) at (-5.75, 0) {$n$};
		\node [style=none-small] (28) at (1, 0) {$m$};
		\node [style=none-small] (29) at (4, 0) {$n$};
		\node [style=none] (31) at (1.5, -0.25) {};
		\node [style=none] (32) at (0.25, 1.5) {};
		\node [style=none] (33) at (0.25, -1.25) {};
		\node [style=none-small] (34) at (13.2, 0) {$:=\begin{pmatrix}1 & 0 \\ 0 & 1\end{pmatrix}$};
		\node [style=none] (35) at (11.45, 1) {};
		\node [style=none] (36) at (11.45, -1) {};
		\node [style=none] (37) at (15.45, 1.5) {};
		\node [style=none] (38) at (15.45, -1.25) {};
		\node [style=none-small] (39) at (18.95, 0) {$:=\frac{1}{\sqrt{2}}\begin{pmatrix}1 & 1 \\ 1 & -1\end{pmatrix}$};
		\node [style=none] (40) at (16.275, 1) {};
		\node [style=none] (41) at (16.275, -1) {};
		\node [style=none] (42) at (10.7, 1.5) {};
		\node [style=none] (43) at (10.7, -1.25) {};
		\node [style=none] (44) at (11.2, 0) {};
		\node [style=hadamard] (45) at (16.275, 0) {};
	\end{pgfonlayer}
	\begin{pgfonlayer}{edgelayer}
		\draw [in=210, out=30] (0) to (1.center);
		\draw [in=150, out=-30] (0) to (2.center);
		\draw [in=30, out=-150] (0) to (3.center);
		\draw [in=330, out=150] (0) to (4.center);
		\draw [in=210, out=30] (7) to (8.center);
		\draw [in=150, out=-30] (7) to (9.center);
		\draw [in=30, out=-150] (7) to (10.center);
		\draw [in=330, out=150] (7) to (11.center);
		\draw [style=gray edge] (32.center) to (33.center);
		\draw [style=gray edge] (37.center) to (38.center);
		\draw (36.center) to (35.center);
		\draw [style=gray edge] (42.center) to (43.center);
		\draw (40.center) to (45);
		\draw (45) to (41.center);
	\end{pgfonlayer}
\end{tikzpicture}
\]
\vspace{-3em}\\
Spiders can have any number of ingoing and outgoing wires and we can compose two diagrams either horizontally by joining the outputs of one diagram with the inputs of the other (denoted by $\circ$), or vertically by placing them side by side (denoted by $\otimes$). This corresponds to the known dot and tensor product in Hilbert space.
For convenience, we distinguish between two types of wires: \textit{Normal wires}, representing the identity, and \textit{Hadamard wires}, representing the Hadamard matrix.
Wires entering the diagram from the left are called input wires, with the adjacent spiders defined as inputs $I$, and wires exiting to the right are called output wires, with adjacent spiders defined as outputs $O$.
We refer to the set of spiders $v \in I \cup O$ as \textit{boundary spiders} and the complementary set of spiders $v \in V \setminus (I \cup O)$ as the \textit{interior spiders}.
The complements of the inputs and the outputs are defined as $\overline{I} = V \setminus I$ and $\overline{O} = V \setminus O$ respectively.
We can write any quantum circuit as ZX-diagram by replacing gates with equivalent diagrams and use rules from ZX-calculus to modify them without changing the linear map. For instance, the following rules hold: 
\vspace{-.5em}
\[
\begin{tikzpicture}
	\begin{pgfonlayer}{nodelayer}
		\node [style=none] (0) at (0, 0) {};
		\node [style=none] (1) at (0, 0) {};
		\node [style=none] (2) at (0, 0) {$=$};
		\node [style=none] (4) at (-3, 0.45) {\vdots};
		\node [style=Z phase dot] (5) at (-3.75, 1) {$\alpha$};
		\node [style=Z phase dot] (6) at (-2.25, -0.5) {$\beta$};
		\node [style=none] (7) at (-1.25, -0.5) {};
		\node [style=none] (8) at (-1.25, -1.5) {};
		\node [style=none] (9) at (-4.75, -0.5) {};
		\node [style=none] (10) at (-4.75, -1.5) {};
		\node [style=none] (11) at (-4.75, 2) {};
		\node [style=none] (12) at (-4.75, 1) {};
		\node [style=none] (13) at (-1.25, 2) {};
		\node [style=none] (14) at (-1.25, 1) {};
		\node [style=none] (15) at (-4.5, -0.75) {\vdots};
		\node [style=none] (16) at (-1.375, -0.75) {\vdots};
		\node [style=none] (17) at (-1.5, 1.625) {\vdots};
		\node [style=none] (18) at (-4.5, 1.625) {\vdots};
		\node [style=Z phase dot] (19) at (3, 0) {};
		\node [style=none] (20) at (1, 1) {};
		\node [style=none] (21) at (1, -1) {};
		\node [style=none] (22) at (1.25, 0.25) {\vdots};
		\node [style=none] (23) at (5, 1) {};
		\node [style=none] (24) at (5, -1) {};
		\node [style=Z phase dot] (25) at (3, 0) {$\alpha+\beta$};
		\node [style=none] (26) at (4.75, 0.25) {\vdots};
		\node [style=X phase dot] (27) at (8.5, 0) {$\alpha$};
		\node [style=none] (28) at (9.75, 1.25) {};
		\node [style=none] (30) at (7.25, 1.25) {};
		\node [style=none] (31) at (7.25, -1.25) {};
		\node [style=none] (33) at (9.75, -1.25) {};
		\node [style=none] (34) at (8.5, -1) {$\ldots$};
		\node [style=none] (35) at (8.5, 1) {$\ldots$};
		\node [style=none] (37) at (14.75, 1.25) {};
		\node [style=none] (39) at (12.275, 1.25) {};
		\node [style=none] (40) at (12.25, -1.25) {};
		\node [style=none] (42) at (14.75, -1.25) {};
		\node [style=none] (43) at (13.5, -1) {$\ldots$};
		\node [style=none] (44) at (13.5, 1) {$\ldots$};
		\node [style=none] (45) at (11, 0) {$=$};
		\node [style=Z phase dot] (46) at (13.5, 0) {$\alpha$};
		\node [style=none] (47) at (6, 2) {};
		\node [style=none] (48) at (6, -1.5) {};
		\node [style=none] (49) at (0, 0.75) {$(f)$};
		\node [style=none] (50) at (11, 0.75) {$(h)$};
		\node [style=hadamard] (51) at (12.275, 0.75) {};
		\node [style=hadamard] (53) at (14.75, 0.75) {};
		\node [style=hadamard] (54) at (12.25, -0.75) {};
		\node [style=hadamard] (55) at (14.75, -0.75) {};
	\end{pgfonlayer}
	\begin{pgfonlayer}{edgelayer}
		\draw (6) to (7.center);
		\draw [in=165, out=-45, looseness=1.25] (6) to (8.center);
		\draw [in=0, out=-165] (6) to (10.center);
		\draw (6) to (9.center);
		\draw [bend left] (6) to (5);
		\draw (5) to (14.center);
		\draw [in=-180, out=15, looseness=1.25] (5) to (13.center);
		\draw [in=0, out=120] (5) to (11.center);
		\draw (5) to (12.center);
		\draw [bend left] (5) to (6);
		\draw [bend left=15] (20.center) to (25);
		\draw [bend right=15, looseness=1.25] (21.center) to (25);
		\draw [bend left=15] (25) to (23.center);
		\draw [bend right=15] (25) to (24.center);
		\draw [bend left=45] (5) to (6);
		\draw [in=-90, out=60, looseness=0.75] (27) to (28.center);
		\draw [in=-90, out=135, looseness=0.75] (27) to (30.center);
		\draw [in=90, out=-135, looseness=0.75] (27) to (31.center);
		\draw [in=90, out=-60, looseness=0.75] (27) to (33.center);
		\draw [style=gray edge] (47.center) to (48.center);
		\draw [in=-75, out=150, looseness=0.75] (46) to (51);
		\draw (51) to (39.center);
		\draw [in=-105, out=30, looseness=0.75] (46) to (53);
		\draw (53) to (37.center);
		\draw [in=75, out=-150, looseness=0.75] (46) to (54);
		\draw (54) to (40.center);
		\draw (42.center) to (55);
		\draw [in=330, out=105, looseness=0.75] (55) to (46);
	\end{pgfonlayer}
\end{tikzpicture}
\]
\vspace{-.75em}\\
The fusion rule ($f$) allows to merge spiders of the same color together if they are connected by at least one normal wire and ($h$) allows to change the colors of spiders by flipping normal and Hadamard wires. All rules hold in both directions and are also valid with interchanged colors, so we can also split up spiders with ($f$).
There exists a complete graphical rule set for transforming ZX-diagrams~\cite{vilmart-near-optimal-2018}.

\subsubsection{Graph-like diagrams}
In this work, we consider the class of \textit{graph-like} ZX-diagrams as introduced in~\cite{duncan-graph-theoretic-2020}, which allow us to represent any quantum computation as a graph of parametrized green Z-spiders and Hadamard wires. In those diagrams we represent Hadamard wires between spiders as dashed blue line instead of the yellow box for easier visualization.
One can transform any ZX-diagram into an equivalent graph-like ZX-diagram by repeatedly applying standard ZX-rules~\cite{duncan-graph-theoretic-2020} (c.f. \Cref{fig:grover}).
\begin{figure}
	\[
	\begin{tikzpicture}
	\begin{pgfonlayer}{nodelayer}
		\node [style=cnot ctrl] (0) at (-12.25, 0.625) {};
		\node [style=cnot ctrl] (1) at (-12.25, -0.625) {};
		\node [style=gate] (2) at (-13.25, 0.625) {H};
		\node [style=gate] (3) at (-13.25, -0.625) {H};
		\node [style=long gate] (4) at (-10.5, 0.625) {$R_z(-\gamma)$};
		\node [style=long gate] (5) at (-10.5, -0.625) {$R_z(-\delta)$};
		\node [style=gate] (6) at (-8.5, 0.625) {H};
		\node [style=gate] (7) at (-8.5, -0.625) {H};
		\node [style=gate] (8) at (-7.25, 0.625) {X};
		\node [style=gate] (9) at (-7.25, -0.625) {X};
		\node [style=gate] (10) at (-6, -0.625) {H};
		\node [style=cnot ctrl] (11) at (-4.75, 0.625) {};
		\node [style=vertex set] (12) at (-4.75, -0.625) {+};
		\node [style=gate] (13) at (-3.5, -0.625) {H};
		\node [style=gate] (14) at (-2.25, 0.625) {X};
		\node [style=gate] (15) at (-2.25, -0.625) {X};
		\node [style=gate] (16) at (-1, 0.625) {H};
		\node [style=gate] (17) at (-1, -0.625) {H};
		\node [style=none] (18) at (-14.75, 0.625) {$\ket{0}$};
		\node [style=none] (19) at (-14.75, -0.625) {$\ket{0}$};
		\node [style=none] (20) at (-14.25, 0.625) {};
		\node [style=none] (21) at (-14.25, -0.625) {};
		\node [style=none] (229) at (0, 0.625) {};
		\node [style=none] (230) at (0, -0.625) {};
		\node [style=Z dot] (231) at (3.75, 0.625) {};
		\node [style=Z dot] (232) at (3.75, -0.625) {};
		\node [style=Z phase dot] (235) at (4.575, 0.625) {$\gamma$};
		\node [style=Z phase dot] (236) at (4.575, -0.625) {$\delta$};
		\node [style=X phase dot] (239) at (6.25, 0.625) {$\pi$};
		\node [style=X phase dot] (240) at (6.25, -0.625) {$\pi$};
		\node [style=Z dot] (242) at (7.525, 0.625) {};
		\node [style=X dot] (243) at (7.525, -0.625) {};
		\node [style=X phase dot] (245) at (8.85, 0.625) {$\pi$};
		\node [style=X phase dot] (246) at (8.85, -0.625) {$\pi$};
		\node [style=X dot] (251) at (2.5, 0.625) {};
		\node [style=X dot] (252) at (2.5, -0.625) {};
		\node [style=none] (254) at (10.075, 0.625) {};
		\node [style=none] (255) at (10.075, -0.625) {};
		\node [style=Z phase dot] (256) at (13.25, 0.625) {$\gamma+\pi$};
		\node [style=Z phase dot] (257) at (13.25, -0.625) {$\delta+\pi$};
		\node [style=Z dot] (262) at (14.75, 0.625) {};
		\node [style=Z dot] (263) at (14.75, -0.625) {};
		\node [style=none] (266) at (16.5, -0.625) {};
		\node [style=none] (267) at (16.5, 0.625) {};
		\node [style=Z phase dot] (268) at (15.75, 0.625) {$\pi$};
		\node [style=Z phase dot] (269) at (15.75, -0.625) {$\pi$};
		\node [style=none] (270) at (0.5, -0.275) {};
		\node [style=none] (271) at (0.5, 0.225) {};
		\node [style=none] (272) at (1.5, -0.275) {};
		\node [style=none] (273) at (1.5, 0.225) {};
		\node [style=none] (274) at (1.75, -0.025) {};
		\node [style=none] (275) at (1.25, 0.475) {};
		\node [style=none] (276) at (1.25, -0.525) {};
		\node [style=none] (277) at (10.625, -0.25) {};
		\node [style=none] (278) at (10.625, 0.25) {};
		\node [style=none] (279) at (11.625, -0.25) {};
		\node [style=none] (280) at (11.625, 0.25) {};
		\node [style=none] (281) at (11.875, 0) {};
		\node [style=none] (282) at (11.375, 0.5) {};
		\node [style=none] (283) at (11.375, -0.5) {};
		\node [style=hadamard] (285) at (5.4, 0.625) {};
		\node [style=hadamard] (286) at (5.4, -0.625) {};
		\node [style=hadamard] (287) at (6.975, -0.625) {};
		\node [style=hadamard] (288) at (8.1, -0.625) {};
		\node [style=hadamard] (289) at (9.6, 0.625) {};
		\node [style=hadamard] (290) at (9.6, -0.625) {};
		\node [style=hadamard] (291) at (3.15, 0.625) {};
		\node [style=hadamard] (292) at (3.15, -0.625) {};
		\node [style=hadamard] (293) at (3.75, 0) {};
	\end{pgfonlayer}
	\begin{pgfonlayer}{edgelayer}
		\draw (11) to (12);
		\draw (6) to (8);
		\draw (8) to (11);
		\draw (11) to (14);
		\draw (14) to (16);
		\draw (20.center) to (2);
		\draw (21.center) to (3);
		\draw (7) to (9);
		\draw (9) to (10);
		\draw (10) to (12);
		\draw (12) to (13);
		\draw (13) to (15);
		\draw (15) to (17);
		\draw (2) to (0);
		\draw (0) to (4);
		\draw (3) to (1);
		\draw (1) to (0);
		\draw (1) to (5);
		\draw (5) to (7);
		\draw (4) to (6);
		\draw (229.center) to (16);
		\draw (230.center) to (17);
		\draw (242) to (243);
		\draw (239) to (242);
		\draw (242) to (245);
		\draw (231) to (235);
		\draw (232) to (236);
		\draw [style=hadamard edge] (256) to (262);
		\draw [style=hadamard edge] (262) to (263);
		\draw [style=hadamard edge] (263) to (257);
		\draw [style=hadamard edge, in=270, out=90] (257) to (256);
		\draw [style=hadamard edge] (268) to (262);
		\draw [style=hadamard edge] (269) to (263);
		\draw (268) to (267.center);
		\draw (266.center) to (269);
		\draw (270.center) to (272.center);
		\draw (271.center) to (273.center);
		\draw (276.center) to (274.center);
		\draw (274.center) to (275.center);
		\draw (277.center) to (279.center);
		\draw (278.center) to (280.center);
		\draw (283.center) to (281.center);
		\draw (281.center) to (282.center);
		\draw (251) to (291);
		\draw (291) to (231);
		\draw (232) to (292);
		\draw (292) to (252);
		\draw (236) to (286);
		\draw (286) to (240);
		\draw (235) to (285);
		\draw (285) to (239);
		\draw (245) to (289);
		\draw (289) to (254.center);
		\draw (255.center) to (290);
		\draw (290) to (246);
		\draw (246) to (288);
		\draw (288) to (243);
		\draw (243) to (287);
		\draw (287) to (240);
		\draw (231) to (293);
		\draw (293) to (232);
	\end{pgfonlayer}
\end{tikzpicture}
	\]
	\caption{Translation of a two-qubit Grover search into a graph like ZX-diagram. Gates are replaced by their ZX-calculus counterpart and the diagram is made graph-like by repeated application of ($f$) and ($h$).}
	\label{fig:grover}
\end{figure}
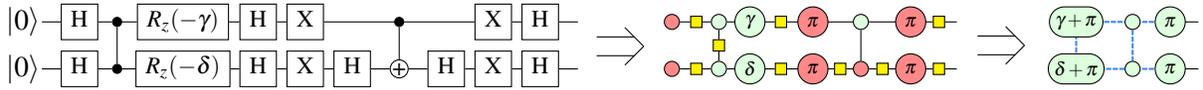
This formalism provides a link between quantum computing and graph theory since the entire computation is captured by the \textit{underlying graph} spanned by Hadamard wires, combined with phases of Z-spiders. Moreover, graph-like diagrams can be directly interpreted as measurement patterns in the model of measurement-based quantum computing (MBQC)~\cite{briegel2009measurement,backens-there-2021}.

\subsubsection{Diagram simplification}\label{sec:diag-simp}
We can rewrite graph-like diagrams into equivalent simplified versions (i.e., decreasing the number of spiders or wires), by using graph-theoretic rewrite rules as shown in~\cite{duncan-graph-theoretic-2020}: \paragraph{Local complementation} Given an undirected graph $G$, local complementation on a vertex $v$ (written $G\star v$) consists of flipping the edges between the neighbors of $v$. That is, after local complementation, every pair of neighbors of $v$ is connected iff it was not connected before. In graph-like ZX-diagrams, we can use a rewrite rule based on local complementation ($lc$) to eliminate spiders with a phase of $\pm\frac{\pi}{2}$:

\[    
    \begin{tikzpicture}
	\begin{pgfonlayer}{nodelayer}
		\node [style=none] (0) at (0, 0) {};
		\node [style=none] (1) at (0, 0) {$=$};
		\node [style=none] (2) at (-1, 0) {};
		\node [style=none] (3) at (-2, 0) {};
		\node [style=Z phase dot] (6) at (-2, 1) {$\alpha_n$};
		\node [style=none] (7) at (-2, -1.75) {};
		\node [style=Z phase dot] (8) at (-3, -0.75) {$\alpha_{n-1}$};
		\node [style=none] (9) at (-3, -1.75) {};
		\node [style=Z phase dot] (10) at (-4.25, 1.5) {$\pm\frac{\pi}{2}$};
		\node [style=Z phase dot] (11) at (-5.5, -0.75) {$\alpha_2$};
		\node [style=none] (12) at (-5.5, -1.75) {};
		\node [style=none] (13) at (-6.5, -1.75) {};
		\node [style=none] (14) at (-6.5, 0) {};
		\node [style=none] (15) at (-7.5, 0) {};
		\node [style=Z phase dot] (16) at (-6.5, 1) {$\alpha_1$};
		\node [style=none] (17) at (8.25, 0.25) {};
		\node [style=none] (18) at (7.25, 0.25) {};
		\node [style=Z phase dot] (19) at (7, 1.25) {$\alpha_n\mp\frac{\pi}{2}$};
		\node [style=none] (20) at (7.5, -1.75) {};
		\node [style=Z phase dot] (21) at (6.5, -0.75) {$\alpha_{n-1}\mp\frac{\pi}{2}$};
		\node [style=none] (22) at (6.5, -1.75) {};
		\node [style=Z phase dot] (24) at (2.5, -0.75) {$\alpha_2\mp\frac{\pi}{2}$};
		\node [style=none] (25) at (2.575, -1.75) {};
		\node [style=none] (26) at (1.5, -1.75) {};
		\node [style=none] (27) at (2, 0.25) {};
		\node [style=none] (28) at (0.95, 0.25) {};
		\node [style=Z phase dot] (29) at (2, 1.25) {$\alpha_1\mp\frac{\pi}{2}$};
		\node [style=none] (30) at (-4.25, 0) {\ldots};
		\node [style=none] (31) at (4.5, 0.5) {\ldots};
		\node [style=none] (32) at (7.725, 0.25) {$\ldots$};
		\node [style=none] (33) at (7, -1.75) {$\ldots$};
		\node [style=none] (34) at (1.45, 0.25) {$\ldots$};
		\node [style=none] (35) at (2, -1.75) {$\ldots$};
		\node [style=none] (36) at (-1.5, 0) {$\ldots$};
		\node [style=none] (37) at (-7, 0) {$\ldots$};
		\node [style=none] (38) at (-6, -1.75) {$\ldots$};
		\node [style=none] (39) at (-2.5, -1.75) {$\ldots$};
		\node [style=none] (40) at (0, 0.75) {$(lc)$};
	\end{pgfonlayer}
	\begin{pgfonlayer}{edgelayer}
		\draw [style=hadamard edge] (10) to (16);
		\draw [style=hadamard edge] (10) to (11);
		\draw [style=hadamard edge] (10) to (8);
		\draw [style=hadamard edge] (10) to (6);
		\draw (16) to (14.center);
		\draw (16) to (15.center);
		\draw (11) to (13.center);
		\draw (11) to (12.center);
		\draw (8) to (9.center);
		\draw (8) to (7.center);
		\draw (6) to (3.center);
		\draw (6) to (2.center);
		\draw (29) to (27.center);
		\draw (29) to (28.center);
		\draw (24) to (26.center);
		\draw (24) to (25.center);
		\draw (21) to (22.center);
		\draw (21) to (20.center);
		\draw (19) to (18.center);
		\draw (19) to (17.center);
		\draw [style=hadamard edge] (29) to (19);
		\draw [style=hadamard edge] (19) to (24);
		\draw [style=hadamard edge] (21) to (29);
		\draw [style=hadamard edge] (24) to (21);
		\draw [style=hadamard edge] (19) to (21);
		\draw [style=hadamard edge] (24) to (29);
	\end{pgfonlayer}
\end{tikzpicture}
\]
\vspace{-2em}
\paragraph{Pivoting} A Pivot $G\land uv$ consists of three local complementations ($G\star u\star v\star u$) applied on a pair of neighboring vertices $u,v$. We can use a similar rewrite rule ($p$) in graph-like ZX-diagrams to eliminate pairs of spiders with phase $0$ or $\pi$:
\vspace{-.5em}
\[
\begin{tikzpicture}
	\begin{pgfonlayer}{nodelayer}
		\node [style=none] (0) at (0, 0) {$=$};
		\node [style=Z phase dot] (1) at (-1.75, -1) {$\gamma_n$};
		\node [style=Z phase dot] (2) at (-1.75, 1.25) {$\gamma_1$};
		\node [style=none] (3) at (-1.25, 2) {};
		\node [style=none] (4) at (-2.25, 2) {};
		\node [style=Z phase dot] (5) at (-3, 1.5) {$k\pi$};
		\node [style=Z phase dot] (6) at (-6.5, 1.5) {$j\pi$};
		\node [style=Z phase dot] (7) at (-4.75, 0.5) {$\beta_1$};
		\node [style=none] (8) at (-4.25, 1.25) {};
		\node [style=none] (9) at (-5.25, 1.25) {};
		\node [style=Z phase dot] (10) at (-4.75, -1.25) {$\beta_n$};
		\node [style=none] (11) at (-4.25, -2) {};
		\node [style=none] (12) at (-5.25, -2) {};
		\node [style=Z phase dot] (13) at (-7.75, -1) {$\alpha_n$};
		\node [style=Z phase dot] (14) at (-7.75, 1.25) {$\alpha_1$};
		\node [style=none] (15) at (-7.25, 2) {};
		\node [style=none] (16) at (-8.25, 2) {};
		\node [style=none] (17) at (-1.25, -2) {};
		\node [style=none] (18) at (-2.25, -2) {};
		\node [style=none] (19) at (-7.25, -2) {};
		\node [style=none] (20) at (-8.25, -2) {};
		\node [style=none] (21) at (-4.75, 1.25) {$\ldots$};
		\node [style=none] (22) at (-4.75, -2) {$\ldots$};
		\node [style=none] (23) at (-1.75, 2) {$\ldots$};
		\node [style=none] (24) at (-1.75, -2) {$\ldots$};
		\node [style=none] (25) at (-7.75, -2) {$\ldots$};
		\node [style=none] (26) at (-7.75, 2) {$\ldots$};
		\node [style=none] (27) at (-1.75, 0.25) {$\vdots$};
		\node [style=none] (28) at (-7.75, 0.25) {$\vdots$};
		\node [style=Z phase dot] (29) at (7.75, 0.25) {$\gamma_n'$};
		\node [style=Z phase dot] (30) at (7.75, 1.9) {$\gamma_1'$};
		\node [style=none] (31) at (8.25, 2.65) {};
		\node [style=none] (32) at (7.25, 2.65) {};
		\node [style=Z phase dot] (35) at (4.75, -0.675) {$\beta_1'$};
		\node [style=none] (36) at (5.25, 0.075) {};
		\node [style=none] (37) at (4.25, 0.075) {};
		\node [style=Z phase dot] (38) at (4.75, -2.325) {$\beta_n'$};
		\node [style=none] (39) at (5.25, -3.075) {};
		\node [style=none] (40) at (4.25, -3.075) {};
		\node [style=Z phase dot] (41) at (1.75, 0.25) {$\alpha_n'$};
		\node [style=Z phase dot] (42) at (1.75, 1.9) {$\alpha_1'$};
		\node [style=none] (43) at (2.25, 2.65) {};
		\node [style=none] (44) at (1.25, 2.65) {};
		\node [style=none] (45) at (8.25, -0.75) {};
		\node [style=none] (46) at (7.25, -0.75) {};
		\node [style=none] (47) at (2.25, -0.75) {};
		\node [style=none] (48) at (1.25, -0.75) {};
		\node [style=none] (49) at (4.75, 0.075) {$\ldots$};
		\node [style=none] (50) at (4.75, -3.075) {$\ldots$};
		\node [style=none] (51) at (7.75, 2.65) {$\ldots$};
		\node [style=none] (52) at (7.75, -0.75) {$\ldots$};
		\node [style=none] (53) at (1.75, -0.75) {$\ldots$};
		\node [style=none] (54) at (1.75, 2.65) {$\ldots$};
		\node [style=none] (55) at (7.75, 1.275) {$\vdots$};
		\node [style=none] (56) at (1.75, 1.275) {$\vdots$};
		\node [style=none] (57) at (4.75, -1.325) {$\vdots$};
		\node [style=none] (59) at (-4.75, -0.175) {$\vdots$};
		\node [style=none] (60) at (13, 0.25) {$\beta_i'=\beta_i+(j+k+1)\pi$};
		\node [style=none] (61) at (11.5, 1.25) {$\alpha_i'=\alpha_i+k\pi$};
		\node [style=none] (62) at (11.5, -0.75) {$\gamma_i'=\gamma_i+j\pi$};
		\node [style=none] (63) at (0, 0.975) {$(p)$};
	\end{pgfonlayer}
	\begin{pgfonlayer}{edgelayer}
		\draw (3.center) to (2);
		\draw (2) to (4.center);
		\draw (1) to (17.center);
		\draw (1) to (18.center);
		\draw (10) to (11.center);
		\draw (10) to (12.center);
		\draw (7) to (9.center);
		\draw (7) to (8.center);
		\draw (14) to (15.center);
		\draw (14) to (16.center);
		\draw (13) to (20.center);
		\draw (13) to (19.center);
		\draw [style=hadamard edge] (7) to (6);
		\draw [style=hadamard edge] (6) to (14);
		\draw [style=hadamard edge] (6) to (13);
		\draw [style=hadamard edge] (6) to (10);
		\draw [style=hadamard edge] (10) to (5);
		\draw [style=hadamard edge] (7) to (5);
		\draw [style=hadamard edge] (1) to (5);
		\draw [style=hadamard edge] (5) to (2);
		\draw [style=hadamard edge] (6) to (5);
		\draw (31.center) to (30);
		\draw (30) to (32.center);
		\draw (29) to (45.center);
		\draw (29) to (46.center);
		\draw (38) to (39.center);
		\draw (38) to (40.center);
		\draw (35) to (37.center);
		\draw (35) to (36.center);
		\draw (42) to (43.center);
		\draw (42) to (44.center);
		\draw (41) to (48.center);
		\draw (41) to (47.center);
		\draw [style=hadamard edge] (42) to (30);
		\draw [style=hadamard edge] (42) to (29);
		\draw [style=hadamard edge] (29) to (41);
		\draw [style=hadamard edge] (42) to (35);
		\draw [style=hadamard edge] (42) to (38);
		\draw [style=hadamard edge] (38) to (30);
		\draw [style=hadamard edge] (38) to (29);
		\draw [style=hadamard edge] (35) to (30);
		\draw [style=hadamard edge] (30) to (41);
		\draw [style=hadamard edge] (41) to (35);
		\draw [style=hadamard edge] (41) to (38);
		\draw [style=hadamard edge] (29) to (35);
	\end{pgfonlayer}
\end{tikzpicture}
\]
\vspace{-1em}\\
By repeatedly applying those rules one can eliminate all interior spiders with phase $\pm\frac{\pi}{2}$ and every pair of interior spiders with phase $0$ or $\pi$~\cite{duncan-graph-theoretic-2020}.
\paragraph{Phase gadgets}
One can further simplify ZX-diagrams with a slightly modified version of the pivot rule if we allow one spider of a pair to have a non-Clifford phase $\sigma$~\cite{kissingerReducingNumberNonClifford2020}. The non-Clifford spider does not get removed but is transformed into a so called \textit{phase gadget}: 
\vspace{-.5em}
\[
\begin{tikzpicture}
	\begin{pgfonlayer}{nodelayer}
		\node [style=none] (0) at (0.25, 0) {$=$};
		\node [style=Z phase dot] (1) at (-1.25, -0.75) {$\gamma_n$};
		\node [style=Z phase dot] (2) at (-1.25, 1.25) {$\gamma_1$};
		\node [style=Z phase dot] (5) at (-2.75, 1.5) {$\sigma$};
		\node [style=Z phase dot] (6) at (-5.75, 1.5) {$j\pi$};
		\node [style=Z phase dot] (7) at (-4.25, -0.25) {$\beta_1$};
		\node [style=Z phase dot] (10) at (-4.25, -2.5) {$\beta_n$};
		\node [style=Z phase dot] (13) at (-7.25, -0.75) {$\alpha_n$};
		\node [style=Z phase dot] (14) at (-7.25, 1.25) {$\alpha_1$};
		\node [style=none] (27) at (-1.25, 0.5) {$\vdots$};
		\node [style=none] (28) at (-7.25, 0.5) {$\vdots$};
		\node [style=Z phase dot] (29) at (7, -0.75) {$\gamma_n'$};
		\node [style=Z phase dot] (30) at (7, 1) {$\gamma_1'$};
		\node [style=Z phase dot] (35) at (4, -1.75) {$\beta_1'$};
		\node [style=Z phase dot] (38) at (4, -3.4) {$\beta_n'$};
		\node [style=Z phase dot] (41) at (2, -0.75) {$\alpha_n'$};
		\node [style=Z phase dot] (42) at (2, 1) {$\alpha_1'$};
		\node [style=none] (55) at (7, 0.25) {$\vdots$};
		\node [style=none] (56) at (2, 0.25) {$\vdots$};
		\node [style=none] (57) at (4, -2.4) {$\vdots$};
		\node [style=Z phase dot] (58) at (4, 2.75) {$(-1)^j\sigma$};
		\node [style=Z dot] (59) at (4, 1.75) {};
		\node [style=none] (60) at (0.25, 1) {$(p)$};
		\node [style=none] (61) at (-4.25, -1.15) {$\vdots$};
		\node [style=none] (62) at (12.25, 0) {$\beta_i'=\beta_i+ \left(j+k+1\right)\pi$};
		\node [style=none] (63) at (10.75, 1) {$\alpha_i'=\alpha_i+k\pi$};
		\node [style=none] (64) at (10.75, -1) {$\gamma_i'=\gamma_i+j\pi$};
		\node [style=none] (65) at (7.5, 1.9) {};
		\node [style=none] (66) at (6.5, 1.9) {};
		\node [style=none] (67) at (2.25, 1.9) {};
		\node [style=none] (68) at (1.25, 1.9) {};
		\node [style=none] (69) at (7, 1.9) {$\ldots$};
		\node [style=none] (70) at (1.75, 1.9) {$\ldots$};
		\node [style=none] (71) at (-0.75, 2.15) {};
		\node [style=none] (72) at (-1.75, 2.15) {};
		\node [style=none] (73) at (-6.75, 2.15) {};
		\node [style=none] (74) at (-7.75, 2.15) {};
		\node [style=none] (75) at (-1.25, 2.15) {$\ldots$};
		\node [style=none] (76) at (-7.25, 2.15) {$\ldots$};
		\node [style=none] (77) at (-0.75, -1.6) {};
		\node [style=none] (78) at (-1.75, -1.6) {};
		\node [style=none] (79) at (-6.75, -1.6) {};
		\node [style=none] (80) at (-7.75, -1.6) {};
		\node [style=none] (81) at (-1.25, -1.6) {$\ldots$};
		\node [style=none] (82) at (-7.25, -1.6) {$\ldots$};
		\node [style=none] (83) at (-3.75, 0.65) {};
		\node [style=none] (84) at (-4.75, 0.65) {};
		\node [style=none] (85) at (-4.25, 0.65) {$\ldots$};
		\node [style=none] (86) at (-3.75, -3.35) {};
		\node [style=none] (87) at (-4.75, -3.35) {};
		\node [style=none] (88) at (-4.25, -3.35) {$\ldots$};
		\node [style=none] (89) at (4.5, -4.25) {};
		\node [style=none] (90) at (3.5, -4.25) {};
		\node [style=none] (91) at (4, -4.25) {$\ldots$};
		\node [style=none] (92) at (4.375, -0.975) {};
		\node [style=none] (93) at (3.625, -0.95) {};
		\node [style=none] (94) at (4, -0.925) {$\ldots$};
		\node [style=none] (95) at (7.5, -1.6) {};
		\node [style=none] (96) at (6.5, -1.6) {};
		\node [style=none] (97) at (7, -1.6) {$\ldots$};
		\node [style=none] (98) at (2.5, -1.85) {};
		\node [style=none] (99) at (1.5, -1.85) {};
		\node [style=none] (100) at (2, -1.85) {$\ldots$};
	\end{pgfonlayer}
	\begin{pgfonlayer}{edgelayer}
		\draw [style=hadamard edge] (7) to (6);
		\draw [style=hadamard edge] (6) to (14);
		\draw [style=hadamard edge] (6) to (13);
		\draw [style=hadamard edge] (6) to (10);
		\draw [style=hadamard edge] (10) to (5);
		\draw [style=hadamard edge] (7) to (5);
		\draw [style=hadamard edge] (1) to (5);
		\draw [style=hadamard edge] (5) to (2);
		\draw [style=hadamard edge] (6) to (5);
		\draw [style=hadamard edge] (42) to (30);
		\draw [style=hadamard edge] (42) to (29);
		\draw [style=hadamard edge] (29) to (41);
		\draw [style=hadamard edge] (42) to (35);
		\draw [style=hadamard edge] (42) to (38);
		\draw [style=hadamard edge] (38) to (30);
		\draw [style=hadamard edge] (38) to (29);
		\draw [style=hadamard edge] (35) to (30);
		\draw [style=hadamard edge] (30) to (41);
		\draw [style=hadamard edge] (41) to (35);
		\draw [style=hadamard edge] (41) to (38);
		\draw [style=hadamard edge] (29) to (35);
		\draw [style=hadamard edge] (42) to (59);
		\draw [style=hadamard edge] (41) to (59);
		\draw [style=hadamard edge] (59) to (58);
		\draw [style=hadamard edge, bend right=45] (35) to (59);
		\draw [style=hadamard edge, in=330, out=30] (38) to (59);
		\draw (68.center) to (42);
		\draw (42) to (67.center);
		\draw (66.center) to (30);
		\draw (30) to (65.center);
		\draw (74.center) to (14);
		\draw (14) to (73.center);
		\draw (72.center) to (2);
		\draw (2) to (71.center);
		\draw (13) to (80.center);
		\draw (79.center) to (13);
		\draw (1) to (78.center);
		\draw (1) to (77.center);
		\draw (10) to (87.center);
		\draw (10) to (86.center);
		\draw (7) to (84.center);
		\draw (7) to (83.center);
		\draw (90.center) to (38);
		\draw (38) to (89.center);
		\draw (35) to (93.center);
		\draw (35) to (92.center);
		\draw (99.center) to (41);
		\draw (41) to (98.center);
		\draw (96.center) to (29);
		\draw (29) to (95.center);
	\end{pgfonlayer}
\end{tikzpicture}
\]
\vspace{-1em}\\
In graph-like ZX-diagrams a phase gadget consists of a ``top'' spider exclusively connected to a phaseless ``root'' spider connected to other spiders. Simplifying graph-like diagrams with all three rules, we obtain a diagram where spiders either have a non-Clifford phase, are part of a phase gadget or a boundary.

\subsubsection{Gflow in graph-like diagrams}
Gflow is a graph-theoretic property for measurement patterns defined on labeled open graphs $(G,I,O,\lambda)$, where $G=(V,E)$ is an undirected graph with vertices $V$ and edges $E$, $I\subseteq V$, $O\subseteq V$ are the set of inputs resp. outputs, and $\lambda$ is a labeling function assigning each vertex a measurement plane of the Bloch sphere in $\{XY,XZ,YZ\}$~\cite{browne2007generalized}. A labeled open graph has gflow if there exists a map $g: \overline{O} \rightarrow \mathcal{P}(\overline{I})$ and a partial order $\prec$ over $V$, s.t. for all $v\in \overline{O}$:
\begin{itemize}
	\item If $w\in g(v)$ and $v\neq w$, then $v\prec w$.
	\item If $w\in Odd(g(v))$ and $v\neq w$, then $v\prec w$.
	\item If $\lambda(v) = XY$, then $v\notin g(v)$ and $v\in Odd(g(v))$.
	\item If $\lambda(v) = XZ$, then $v\in g(v)$ and $v\in Odd(g(v))$.
	\item If $\lambda(v) = YZ$, then $v\in g(v)$ and $v\notin Odd(g(v))$.
\end{itemize}
In graph-like diagrams, we interpret the underlying graph as a labeled open graph with phase gadgets corresponding to $YZ$ measurements\footnote{The root spider is labeled as $YZ$ measurement, while the top spider corresponds to a measurement effect which is omitted from the underlying graph.} and other spiders corresponding to $XY$ measurements~\cite{backens-there-2021}. Since the initial graph-like diagrams obtained from quantum circuits (as in \Cref{fig:grover}) have gflow~\cite{duncan-graph-theoretic-2020} and all above rules preserve gflow~\cite{backens-there-2021}, the simplified diagrams have gflow as well.

\subsubsection{Circuit extraction}\label{sec:basic-extraction}
Extracting quantum circuits back from graph-like ZX-diagrams where the circuit has only as many qubits as there are outputs/inputs, is so far only possible in polynomial time if the underlying graph has some kind of flow~\cite{backens-there-2021,simmons2021}. Here, we give a brief overview of the extraction algorithm for graph-like ZX-diagrams with gflow as described in~\cite{backens-there-2021}.
The algorithm extracts a quantum circuit from a ZX-diagram by taking suitable parts of the diagram and creating their equivalent representation as a quantum gate within the circuit at the corresponding position. These parts are then removed from the diagram, extracting one gate at a time, until only the inputs and outputs of the diagram remain. During the process, a set of green Z-spiders called the \textit{frontier} separates the extracted part of the diagram from the unextracted part.
Phases of frontier spiders can be directly extracted as $R_z$ gates, and Hadamard wires between frontier spiders as CZ gates. Furthermore, Hadamard wires where a frontier spider $w$ is exclusively connected to a non-frontier spider $v$ can be extracted as Hadamard gates with $v$ replacing $w$ in the frontier:
\vspace{-.5em}
\[
\begin{tikzpicture}
	\begin{pgfonlayer}{nodelayer}
		\node [style=none] (0) at (-8.75, 1.25) {};
		\node [style=none] (1) at (-8.75, 0) {};
		\node [style=none] (2) at (-8.75, -1.25) {};
		\node [style=Z phase dot] (3) at (-7.75, 1.25) {$\frac{3\pi}{2}$};
		\node [style=Z dot] (4) at (-7.75, 0) {};
		\node [style=Z phase dot] (5) at (-7.75, -1.25) {$\frac{\pi}{2}$};
		\node [style=Z dot] (6) at (-3.75, -1.25) {};
		\node [style=Z dot] (7) at (-6.25, 2) {};
		\node [style=Z dot] (8) at (-5, 1.25) {};
		\node [style=Z phase dot] (9) at (-3.75, 0) {$\frac{\pi}{4}$};
		\node [style=none] (10) at (-2, 1.25) {};
		\node [style=none] (11) at (-2, 0) {};
		\node [style=none] (12) at (-2, -1.25) {};
		\node [style=none] (13) at (-2.75, 2) {};
		\node [style=none] (14) at (-2.75, -1.5) {};
		\node [style=Z phase dot] (15) at (-6.25, 3) {$\frac{3\pi}{4}$};
		\node [style=Z phase dot] (16) at (-3.75, 1.25) {$\frac{\pi}{8}$};
		\node [style=none] (17) at (0, 1.25) {};
		\node [style=none] (18) at (0, 0) {};
		\node [style=none] (19) at (0, -1.25) {};
		\node [style=Z phase dot] (20) at (1, 1.25) {$\frac{3\pi}{2}$};
		\node [style=Z dot] (21) at (1, 0) {};
		\node [style=Z phase dot] (22) at (1, -1.25) {$\frac{\pi}{2}$};
		\node [style=Z dot] (23) at (3.5, -1.25) {};
		\node [style=Z dot] (24) at (2.5, 2) {};
		\node [style=Z dot] (25) at (3.5, 1.25) {};
		\node [style=Z dot] (26) at (3.5, 0) {};
		\node [style=none] (27) at (4.75, 1.25) {};
		\node [style=none] (28) at (4.75, 0) {};
		\node [style=none] (29) at (4.75, -1.25) {};
		\node [style=none] (30) at (4.25, 2) {};
		\node [style=none] (31) at (4.25, -1.5) {};
		\node [style=Z phase dot] (32) at (2.5, 3) {$\frac{3\pi}{4}$};
		\node [style=gate] (33) at (5.25, 1.25) {H};
		\node [style=gate] (34) at (7.25, 1.25) {$R_z(\frac{\pi}{8})$};
		\node [style=vertex] (35) at (5.25, 0) {};
		\node [style=vertex] (36) at (5.25, -1.25) {};
		\node [style=gate] (37) at (7.25, 0) {$R_z(\frac{\pi}{4})$};
		\node [style=none] (38) at (8.5, 1.25) {};
		\node [style=none] (39) at (8.5, 0) {};
		\node [style=none] (40) at (8.5, -1.25) {};
		\node [style=none] (41) at (-1, 0) {$=$};
	\end{pgfonlayer}
	\begin{pgfonlayer}{edgelayer}
		\draw (0.center) to (3);
		\draw (1.center) to (4);
		\draw [style=hadamard edge] (3) to (9);
		\draw [style=hadamard edge] (3) to (7);
		\draw [style=hadamard edge] (3) to (8);
		\draw [style=hadamard edge] (4) to (8);
		\draw [style=hadamard edge] (5) to (9);
		\draw [style=hadamard edge] (6) to (7);
		\draw [style=double edge] (14.center) to (13.center);
		\draw (11.center) to (9);
		\draw (12.center) to (6);
		\draw (5) to (2.center);
		\draw [style=hadamard edge] (15) to (7);
		\draw [style=hadamard edge] (6) to (3);
		\draw [style=hadamard edge] (4) to (7);
		\draw [style=hadamard edge] (5) to (7);
		\draw [style=hadamard edge] (9) to (6);
		\draw [style=hadamard edge] (8) to (16);
		\draw (16) to (10.center);
		\draw (17.center) to (20);
		\draw (18.center) to (21);
		\draw [style=hadamard edge] (20) to (26);
		\draw [style=hadamard edge] (20) to (24);
		\draw [style=hadamard edge] (20) to (25);
		\draw [style=hadamard edge] (21) to (25);
		\draw [style=hadamard edge] (22) to (26);
		\draw [style=hadamard edge] (23) to (24);
		\draw [style=double edge] (31.center) to (30.center);
		\draw (25) to (27.center);
		\draw (28.center) to (26);
		\draw (29.center) to (23);
		\draw (22) to (19.center);
		\draw [style=hadamard edge] (32) to (24);
		\draw [style=hadamard edge] (23) to (20);
		\draw [style=hadamard edge] (21) to (24);
		\draw [style=hadamard edge] (22) to (24);
		\draw (34) to (38.center);
		\draw (34) to (33);
		\draw (33) to (27.center);
		\draw (28.center) to (35);
		\draw (35) to (36);
		\draw (36) to (29.center);
		\draw (35) to (37);
		\draw (37) to (39.center);
		\draw (36) to (40.center);
	\end{pgfonlayer}
\end{tikzpicture}
\]
\vspace{-.5em}\\
If every spider in the frontier has at least two non-frontier neighbors we can add wires of a frontier spider to the wires of another one by placing a CX gate on the extracted circuit. If all neighbors are measured in the XY plane, gflow ensures that there exists a combination of additions so that there remains a frontier spider with only a single neighbor. We can obtain such a combination by applying Gaussian elimination on the biadjacency matrix between the frontier vertices and their neighbors.
Otherwise, if there are YZ-measured neighbors, we can transform them into XY measurements by applying a pivot on the neighbor and a connected frontier spider:
\vspace{-.5em}
\[
\begin{tikzpicture}
	\begin{pgfonlayer}{nodelayer}
		\node [style=none] (0) at (-11, 1.25) {};
		\node [style=none] (1) at (-11, 0) {};
		\node [style=none] (2) at (-11, -1.25) {};
		\node [style=Z phase dot] (3) at (-10, 1.25) {$\frac{3\pi}{2}$};
		\node [style=Z dot] (4) at (-10, 0) {};
		\node [style=Z phase dot] (11) at (-10, -1.25) {$\frac{\pi}{2}$};
		\node [style=Z dot] (35) at (-7.5, -1.25) {};
		\node [style=Z phase dot] (41) at (-8.75, 0) {$\frac{3\pi}{4}$};
		\node [style=Z dot] (54) at (-7.5, 1.25) {};
		\node [style=Z dot] (56) at (-7.5, 0) {};
		\node [style=none] (60) at (-5, 1.25) {};
		\node [style=none] (61) at (-5, 0) {};
		\node [style=none] (62) at (-6.25, -1.25) {};
		\node [style=none] (71) at (-6.75, 2) {};
		\node [style=none] (72) at (-6.75, -1.5) {};
		\node [style=none] (85) at (4, 1.25) {};
		\node [style=none] (86) at (4, 0) {};
		\node [style=none] (87) at (4, -1.25) {};
		\node [style=Z phase dot] (88) at (5, 1.25) {$\frac{3\pi}{2}$};
		\node [style=Z dot] (89) at (5, 0) {};
		\node [style=Z phase dot] (90) at (5, -1.25) {$\frac{\pi}{2}$};
		\node [style=Z dot] (91) at (7.5, -1.25) {};
		\node [style=Z phase dot] (92) at (6.5, -0.25) {$\frac{3\pi}{4}$};
		\node [style=Z dot] (93) at (7.5, 1.25) {};
		\node [style=Z dot] (94) at (7.5, 0) {};
		\node [style=none] (95) at (10.75, 1.25) {};
		\node [style=none] (96) at (10.75, 0) {};
		\node [style=none] (106) at (8.25, 2) {};
		\node [style=none] (107) at (8.25, -1.5) {};
		\node [style=none] (120) at (3.25, 0) {$=$};
		\node [style=vertex set] (121) at (10, 0) {+};
		\node [style=vertex] (122) at (10, 1.25) {};
		\node [style=vertex set] (123) at (9, 1.25) {+};
		\node [style=vertex] (124) at (9, -1.25) {};
		\node [style=none] (125) at (-3.5, 1.25) {};
		\node [style=none] (126) at (-3.5, 0) {};
		\node [style=none] (127) at (-3.5, -1.25) {};
		\node [style=Z phase dot] (128) at (-2.5, 1.25) {$\frac{3\pi}{2}$};
		\node [style=Z dot] (129) at (-2.5, 0) {};
		\node [style=Z phase dot] (130) at (-2.5, -1.25) {$\frac{\pi}{2}$};
		\node [style=Z dot] (131) at (0, -1.25) {};
		\node [style=Z phase dot] (132) at (-1, -0.25) {$\frac{3\pi}{4}$};
		\node [style=Z dot] (133) at (0, 1.25) {};
		\node [style=Z dot] (134) at (0, 0) {};
		\node [style=none] (135) at (2.5, 1.25) {};
		\node [style=none] (136) at (2.5, 0) {};
		\node [style=none] (137) at (1.25, -1.25) {};
		\node [style=none] (138) at (0.75, 2) {};
		\node [style=none] (139) at (0.75, -1.5) {};
		\node [style=none] (140) at (-4.25, 0) {$=$};
		\node [style=vertex set] (141) at (1.75, 0) {+};
		\node [style=vertex] (142) at (1.75, 1.25) {};
		\node [style=none] (143) at (-18, 1.25) {};
		\node [style=none] (144) at (-18, 0) {};
		\node [style=none] (145) at (-18, -1.25) {};
		\node [style=Z phase dot] (146) at (-17, 1.25) {$\frac{3\pi}{2}$};
		\node [style=Z dot] (147) at (-17, 0) {};
		\node [style=Z phase dot] (148) at (-17, -1.25) {$\frac{\pi}{2}$};
		\node [style=Z dot] (149) at (-14.5, -1.25) {};
		\node [style=Z dot] (150) at (-15.5, 2) {};
		\node [style=Z dot] (151) at (-14.5, 1.25) {};
		\node [style=Z dot] (152) at (-14.5, 0) {};
		\node [style=none] (153) at (-13, 1.25) {};
		\node [style=none] (154) at (-13, 0) {};
		\node [style=none] (155) at (-13, -1.25) {};
		\node [style=none] (156) at (-13.75, 2) {};
		\node [style=none] (157) at (-13.75, -1.5) {};
		\node [style=none] (158) at (-12, 0) {$=$};
		\node [style=Z phase dot] (159) at (-15.5, 3) {$\frac{3\pi}{4}$};
		\node [style=none] (160) at (-5, -1.25) {};
		\node [style=gate] (161) at (-5.75, -1.25) {$H$};
		\node [style=none] (162) at (2.5, -1.25) {};
		\node [style=gate] (163) at (1.75, -1.25) {$H$};
		\node [style=none] (164) at (10.75, -1.25) {};
		\node [style=gate] (165) at (10, -1.25) {$H$};
	\end{pgfonlayer}
	\begin{pgfonlayer}{edgelayer}
		\draw (0.center) to (3);
		\draw (1.center) to (4);
		\draw [style=hadamard edge] (3) to (56);
		\draw [style=hadamard edge] (3) to (41);
		\draw [style=hadamard edge, bend right] (3) to (11);
		\draw [style=hadamard edge] (3) to (54);
		\draw [style=hadamard edge] (4) to (35);
		\draw [style=hadamard edge] (4) to (54);
		\draw [style=hadamard edge] (11) to (35);
		\draw [style=hadamard edge] (11) to (56);
		\draw [style=hadamard edge] (35) to (41);
		\draw [style=double edge] (72.center) to (71.center);
		\draw (85.center) to (88);
		\draw (86.center) to (89);
		\draw [style=hadamard edge] (88) to (94);
		\draw [style=hadamard edge] (88) to (92);
		\draw [style=hadamard edge, bend right] (88) to (90);
		\draw [style=hadamard edge] (89) to (93);
		\draw [style=hadamard edge] (90) to (94);
		\draw [style=hadamard edge] (91) to (92);
		\draw [style=double edge] (107.center) to (106.center);
		\draw (121) to (94);
		\draw (122) to (121);
		\draw [style=hadamard edge, bend right=15, looseness=0.75] (93) to (90);
		\draw (93) to (123);
		\draw (123) to (122);
		\draw (124) to (91);
		\draw (124) to (123);
		\draw (54) to (60.center);
		\draw (61.center) to (56);
		\draw (62.center) to (35);
		\draw (122) to (95.center);
		\draw (96.center) to (121);
		\draw (90) to (87.center);
		\draw (11) to (2.center);
		\draw (125.center) to (128);
		\draw (126.center) to (129);
		\draw [style=hadamard edge] (128) to (134);
		\draw [style=hadamard edge] (128) to (132);
		\draw [style=hadamard edge, bend right] (128) to (130);
		\draw [style=hadamard edge] (129) to (133);
		\draw [style=hadamard edge] (130) to (134);
		\draw [style=hadamard edge] (131) to (132);
		\draw [style=double edge] (139.center) to (138.center);
		\draw (141) to (134);
		\draw (142) to (141);
		\draw [style=hadamard edge, bend right=15, looseness=0.75] (133) to (130);
		\draw (142) to (135.center);
		\draw (136.center) to (141);
		\draw (130) to (127.center);
		\draw (133) to (142);
		\draw (137.center) to (131);
		\draw [style=hadamard edge] (131) to (130);
		\draw [style=hadamard edge] (131) to (129);
		\draw (143.center) to (146);
		\draw (144.center) to (147);
		\draw [style=hadamard edge] (146) to (152);
		\draw [style=hadamard edge] (146) to (150);
		\draw [style=hadamard edge] (146) to (151);
		\draw [style=hadamard edge] (147) to (151);
		\draw [style=hadamard edge] (148) to (152);
		\draw [style=hadamard edge] (149) to (150);
		\draw [style=double edge] (157.center) to (156.center);
		\draw (151) to (153.center);
		\draw (154.center) to (152);
		\draw (155.center) to (149);
		\draw (148) to (145.center);
		\draw [style=hadamard edge] (159) to (150);
		\draw [style=hadamard edge] (149) to (146);
		\draw [style=hadamard edge] (147) to (150);
		\draw [style=hadamard edge] (148) to (150);
		\draw [style=hadamard edge] (3) to (4);
		\draw (62.center) to (161);
		\draw (161) to (160.center);
		\draw (163) to (162.center);
		\draw (163) to (137.center);
		\draw (165) to (164.center);
		\draw (165) to (124);
		\draw [style=hadamard edge] (128) to (129);
		\draw [style=hadamard edge] (88) to (89);
	\end{pgfonlayer}
\end{tikzpicture}
\]
\vspace{-.5em}\\
By repeating these procedures, we can transform the entire diagram into a quantum circuit.

\section{Extracting controlled phase gates from graph-like ZX-diagrams}
The process of transforming quantum circuits to graph-like ZX-diagrams, simplifying them, and re-extracting circuits can already be seen as an implicit synthesis algorithm to the gate set $\{R_z,H,CZ,CX\}$.
Given the requirements of neutral atom platforms, it may be desirable to incorporate two- and multi-controlled phase gates of arbitrary rotations into this gate set.
We first show how such gates are represented in graph-like ZX-diagrams, then how we incorporate this finding into the extraction algorithm.
\subsection{Graph-like representation of controlled phase gates}
The (multi-) controlled phase gate is a diagonal gate, meaning all non-zero entries of its corresponding matrix in the Z-basis are on its diagonal. Such gates can be represented as a semi-Boolean function $f : \{0, 1\}^n \rightarrow \mathbb{C}$ which assigns a complex number to each basis state. Ref.~\cite{kuijpers2019graphical} shows, that any semi-Boolean function $f(b) = a_b$ with $a_b\in\mathbb{C}$ and $b\in\mathbb{B}^n$ can be expressed in ZX-calculus as follows:
\vspace{-.5em}
\begin{equation}\label{eq:semibool}
	\begin{tikzpicture}
	\begin{pgfonlayer}{nodelayer}
		\node [style=Z phase dot] (0) at (0, -1.25) {$\alpha_c$};
		\node [style=X dot] (1) at (0, -0.25) {};
		\node [style=trapez] (2) at (0, 1.25) {$c$};
		\node [style=none] (3) at (1, 1.25) {};
		\node [style=none] (4) at (-1, 1.25) {};
		\node [style=Z dot] (5) at (1, 3) {};
		\node [style=Z dot] (6) at (-1, 3) {};
		\node [style=none] (7) at (0, 3) {$\ldots$};
		\node [style=none] (8) at (0, 0.65) {$\ldots$};
		\node [style=none] (9) at (0, 1.85) {$\ldots$};
		\node [style=none] (10) at (-2, 2.25) {};
		\node [style=none] (11) at (2, 2.25) {};
		\node [style=none] (12) at (2, -1.9) {};
		\node [style=none] (13) at (-2, -1.9) {};
		\node [style=none] (14) at (2.25, 2.725) {$c\in \mathbb{B}^n$};
		\node [style=none] (15) at (-1, 3.75) {};
		\node [style=none] (16) at (1, 3.75) {};
		\node [style=Z dot] (17) at (7, -0.25) {};
		\node [style=X dot] (18) at (7, -1) {};
		\node [style=none] (19) at (7, -1.75) {};
		\node [style=none] (20) at (7, 0.5) {};
		\node [style=none] (21) at (4, 4) {};
		\node [style=none] (22) at (4, -2) {};
		\node [style=none] (23) at (5.25, 3) {with};
		\node [style=none] (24) at (5.75, 1.75) {$c_i=$};
		\node [style=none] (25) at (7, 2.25) {};
		\node [style=none] (26) at (7, 1.5) {};
		\node [style=none] (27) at (9.75, 1.75) {if $c_i=1$, and};
		\node [style=none] (28) at (5.75, -0.5) {$c_i=$};
		\node [style=none] (31) at (9, -0.5) {if $c_i=0$};
		\node [style=none] (32) at (17, 0.25) {$\alpha_c=\frac{-1}{2^{n-1}}\sum_{b\in \mathbb{B}^n}f(b) \chi(b,c)$};
	\end{pgfonlayer}
	\begin{pgfonlayer}{edgelayer}
		\draw (0) to (1);
		\draw [in=135, out=-90, looseness=0.75] (4.center) to (1);
		\draw [in=-90, out=45, looseness=0.75] (1) to (3.center);
		\draw (5) to (3.center);
		\draw (4.center) to (6);
		\draw [style=box edge] (12.center) to (11.center);
		\draw [style=box edge, in=360, out=180] (11.center) to (10.center);
		\draw [style=box edge] (10.center) to (13.center);
		\draw [style=box edge] (13.center) to (12.center);
		\draw (5) to (16.center);
		\draw (6) to (15.center);
		\draw (20.center) to (17);
		\draw (19.center) to (18);
		\draw (21.center) to (22.center);
		\draw (25.center) to (26.center);
	\end{pgfonlayer}
\end{tikzpicture}
\end{equation}
\vspace{-1em}\\
The part in the dashed box is repeated for every Boolean vector $c$ in $\mathbb{B}^n$ and the grey box decomposes into $n$ subdiagrams either connecting the corresponding lower and upper wire with a normal wire if the $i$-th element of the Boolean vector is $1$, or disconnecting them if it is $0$.
Further, the phase $\alpha_c$ can be obtained by the formula on the right, where $\chi(b,c)=(-1)^{b\cdot c}$ corresponds to a parity function with $\cdot$ being the inner product: If $b$ and $c$ overlap in an odd number of elements it returns -1, else 1.
This rule yields a combination of phase gadgets, and when applying the color change rule on the middle red spider, we obtain the same graph structures as introduced in the previous section.
To model an $n$-controlled phase gate C${}_n$P$(\varphi)$ as a semi-Boolean function, we take $\alpha$ as an all-zero vector of length $2^n$ except for its last entry being $\varphi$.
Following \Cref{eq:semibool}, one can transform the function to a ZX-diagram which has $2^{n}-1$ phase gadgets split up into $\binom{k}{n}$ phase gadgets for $k\in\{1,\ldots n\}$.
For instance, the two and three-qubit-controlled phase gates have the following representation:\\
\noindent\begin{minipage}{.4\linewidth}\vspace{-1em}
	\begin{equation}\label{eq:extract-phases}
		\begin{tikzpicture}
	\begin{pgfonlayer}{nodelayer}
		\node [style=Z dot] (6) at (-3, -1.5) {};
		\node [style=Z phase dot] (8) at (-3, 1) {$\alpha$};
		\node [style=Z phase dot] (9) at (-3, -0.25) {$\alpha$};
		\node [style=none] (10) at (-1.5, 1) {};
		\node [style=none] (11) at (-1.5, -0.25) {};
		\node [style=none] (12) at (-1.5, -1.5) {};
		\node [style=Z dot] (13) at (-4.5, 2) {};
		\node [style=Z phase dot] (14) at (-4.5, 3.25) {$-\alpha$};
		\node [style=none] (15) at (-2, 2) {};
		\node [style=none] (16) at (-2, -1.75) {};
		\node [style=none] (23) at (-4.5, -0.25) {$\ldots$};
		\node [style=Z dot] (28) at (2, -1.5) {};
		\node [style=Z dot] (29) at (2, 1) {};
		\node [style=Z dot] (30) at (2, -0.25) {};
		\node [style=vertex] (31) at (3.75, 1) {};
		\node [style=vertex] (32) at (3.75, -0.25) {};
		\node [style=none] (33) at (3.75, -1.5) {};
		\node [style=none] (41) at (1, 0) {$\ldots$};
		\node [style=none] (44) at (-0.25, 0) {$=$};
		\node [style=none] (45) at (3, 2) {};
		\node [style=none] (46) at (3, -1.75) {};
		\node [style=none] (47) at (4.25, 0.25) {$2\alpha$};
		\node [style=none] (48) at (5, 1) {};
		\node [style=none] (49) at (5, -0.25) {};
		\node [style=none] (50) at (5, -1.5) {};
	\end{pgfonlayer}
	\begin{pgfonlayer}{edgelayer}
		\draw (6) to (12.center);
		\draw (8) to (10.center);
		\draw (9) to (11.center);
		\draw [style=hadamard edge] (14) to (13);
		\draw [style=hadamard edge] (13) to (8);
		\draw [style=hadamard edge] (9) to (13);
		\draw [style=double edge] (15.center) to (16.center);
		\draw (28) to (33.center);
		\draw (29) to (31);
		\draw (30) to (32);
		\draw [style=double edge] (45.center) to (46.center);
		\draw (31) to (32);
		\draw (31) to (48.center);
		\draw (49.center) to (32);
		\draw (50.center) to (33.center);
	\end{pgfonlayer}
\end{tikzpicture}
	\end{equation}
\end{minipage}%
\begin{minipage}{.6\linewidth}
	\begin{equation}
		\begin{tikzpicture}
	\begin{pgfonlayer}{nodelayer}
		\node [style=Z phase dot] (0) at (-2.75, -1.25) {$\alpha$};
		\node [style=Z phase dot] (1) at (-2.75, 1.25) {$\alpha$};
		\node [style=Z phase dot] (2) at (-2.75, 0) {$\alpha$};
		\node [style=none] (3) at (-1.25, 1.25) {};
		\node [style=none] (4) at (-1.25, 0) {};
		\node [style=none] (5) at (-1.25, -1.25) {};
		\node [style=Z dot] (6) at (-4, 2) {};
		\node [style=Z phase dot] (7) at (-4, 3.25) {$-\alpha$};
		\node [style=none] (8) at (-1.75, 2.25) {};
		\node [style=none] (9) at (-1.75, -1.5) {};
		\node [style=none] (15) at (-6, 0) {$\ldots$};
		\node [style=Z dot] (20) at (2.25, -1.25) {};
		\node [style=Z dot] (21) at (2.25, 1.25) {};
		\node [style=Z dot] (22) at (2.25, 0) {};
		\node [style=vertex] (23) at (4, 1.25) {};
		\node [style=vertex] (24) at (4, 0) {};
		\node [style=vertex] (25) at (4, -1.25) {};
		\node [style=none] (31) at (1, 0) {$\ldots$};
		\node [style=none] (32) at (0, 0) {$=$};
		\node [style=none] (33) at (3.25, 2.25) {};
		\node [style=none] (34) at (3.25, -1.5) {};
		\node [style=none] (35) at (4.5, -0.75) {$4\alpha$};
		\node [style=none] (36) at (5.25, 1.25) {};
		\node [style=none] (37) at (5.25, 0) {};
		\node [style=none] (38) at (5.25, -1.25) {};
		\node [style=Z dot] (39) at (-5.25, 2) {};
		\node [style=Z phase dot] (40) at (-5.25, 3.25) {$-\alpha$};
		\node [style=Z dot] (41) at (-6.5, 2) {};
		\node [style=Z phase dot] (42) at (-6.5, 3.25) {$-\alpha$};
		\node [style=Z dot] (43) at (-7.75, 2) {};
		\node [style=Z phase dot] (44) at (-7.75, 3.25) {$\alpha$};
	\end{pgfonlayer}
	\begin{pgfonlayer}{edgelayer}
		\draw (0) to (5.center);
		\draw (1) to (3.center);
		\draw (2) to (4.center);
		\draw [style=hadamard edge] (7) to (6);
		\draw [style=hadamard edge] (6) to (1);
		\draw [style=hadamard edge] (2) to (6);
		\draw [style=double edge] (8.center) to (9.center);
		\draw (20) to (25);
		\draw (21) to (23);
		\draw (22) to (24);
		\draw [style=double edge] (33.center) to (34.center);
		\draw (23) to (24);
		\draw (23) to (36.center);
		\draw (37.center) to (24);
		\draw (38.center) to (25);
		\draw [style=hadamard edge] (40) to (39);
		\draw [style=hadamard edge] (42) to (41);
		\draw [style=hadamard edge] (44) to (43);
		\draw [style=hadamard edge] (39) to (0);
		\draw [style=hadamard edge] (39) to (2);
		\draw [style=hadamard edge] (41) to (1);
		\draw [style=hadamard edge] (41) to (0);
		\draw [style=hadamard edge] (1) to (43);
		\draw [style=hadamard edge] (43) to (2);
		\draw [style=hadamard edge] (0) to (43);
		\draw (24) to (25);
	\end{pgfonlayer}
\end{tikzpicture}
	\end{equation}
\end{minipage}\\

For arbitrary-sized multi-controlled phase gates, this generalizes to the following theorem:

\begin{theorem}[Multi-controlled phase gates]
  \label{theo:mcp-gates}
  	Let $\binom{S}{k}$ denote the set of all k-combinations of a set $S$ and PG$(\alpha,N)$ denote a phase gadget with phase $\alpha$ connected to neighbors $N$ which are empty Z-spiders. A n-qubit controlled phase gate C${}_n$P$(\varphi)$ is equivalent to a graph-like ZX-diagram with outputs $O,|O|=n$ having a phase of $\alpha$ and phase gadgets\footnote{Note, that the product notation here corresponds to the composition $\circ$ of ZX-diagrams.} \[
  \prod_{k=2}^{n}\prod_{s\in\binom{O}{k}} PG((-1)^{k+1}\alpha,s),\qquad \alpha = \frac{\varphi}{2^{n-1}}
  \]
  \begin{proof}
  	A graphical proof is given by the ZX representation of the diagonal gate and the corresponding proof in~\cite{kuijpers2019graphical}, a combinatorial proof can be found in~\cite{amy2019}. We give an alternative combinatorial proof in~\Cref{sec:proof-theorem-x}.
  \end{proof}
\end{theorem}

\subsection{Adaption of the extraction algorithm}
\label{sec:algorithm}
We adapt the algorithm described in~\ref{sec:basic-extraction} by including an additional C${}_n$P gate extraction step between CZ and R${}_z$ extraction.
For that, we carry out a pattern match on the phase gadgets which are exclusively connected to the outputs, i.e., the frontier. 
If we find a graph structure as described in~\Cref{theo:mcp-gates} for $n$ frontier spiders, we take the phase of the gadget which is connected to all $n$ spiders as the desired phase $\alpha$ if $n$ is odd, or $-\alpha$ if $n$ is even, and adjust the phases of all other gadgets using the following two rewrites:

\noindent\begin{minipage}{.5\linewidth}
	\begin{equation}\label{eq:output-unfusion}
		\begin{tikzpicture}
	\begin{pgfonlayer}{nodelayer}
		\node [style=none] (0) at (-0.5, 0.825) {};
		\node [style=none] (1) at (-0.5, -0.75) {};
		\node [style=none] (3) at (-3.5, 1.5) {};
		\node [style=Z phase dot] (4) at (-1.75, 0.825) {$\beta$};
		\node [style=Z phase dot] (5) at (-1.75, -0.75) {$\alpha$};
		\node [style=none] (6) at (-1.75, 0.25) {$\vdots$};
		\node [style=none] (7) at (-3.5, 0.75) {};
		\node [style=none] (9) at (7.5, -0.75) {};
		\node [style=none] (10) at (1, 1.325) {};
		\node [style=Z phase dot] (11) at (2.75, 0.825) {$\alpha$};
		\node [style=Z phase dot] (12) at (2.75, -0.75) {$\alpha$};
		\node [style=none] (13) at (2.75, 0.25) {$\vdots$};
		\node [style=none] (14) at (1, 0.5) {};
		\node [style=none] (15) at (0, 0) {$=$};
		\node [style=none] (16) at (7.5, 0.825) {};
		\node [style=none] (17) at (-3, 1) {$\vdots$};
		\node [style=none] (18) at (1.5, 0.825) {$\vdots$};
		\node [style=none] (19) at (-1, 1.575) {};
		\node [style=none] (20) at (-1, -1) {};
		\node [style=none] (21) at (3.5, 1.575) {};
		\node [style=none] (22) at (3.5, -1) {};
		\node [style=gate] (23) at (5.5, 0.825) {$R_z(\beta-\alpha)$};
	\end{pgfonlayer}
	\begin{pgfonlayer}{edgelayer}
		\draw (4) to (0.center);
		\draw (1.center) to (5);
		\draw [style=hadamard edge] (4) to (3.center);
		\draw [style=hadamard edge] (5) to (7.center);
		\draw (9.center) to (12);
		\draw [style=hadamard edge] (11) to (10.center);
		\draw [style=hadamard edge] (12) to (14.center);
		\draw [style=double edge] (19.center) to (20.center);
		\draw (16.center) to (23);
		\draw (23) to (11);
		\draw [style=double edge] (21.center) to (22.center);
	\end{pgfonlayer}
\end{tikzpicture}
	\end{equation}
\end{minipage}%
\begin{minipage}{.5\linewidth}
	\begin{equation}\label{eq:gadget-unfusion}
		\begin{tikzpicture}
	\begin{pgfonlayer}{nodelayer}
		\node [style=none] (0) at (-0.5, 0.85) {};
		\node [style=none] (2) at (-0.5, -0.75) {};
		\node [style=none] (3) at (0, 0) {$=$};
		\node [style=Z phase dot] (4) at (-3, 2.6) {$\beta$};
		\node [style=Z dot] (5) at (-3, 1.6) {};
		\node [style=Z phase dot] (6) at (-1.75, 0.85) {$\alpha$};
		\node [style=Z phase dot] (8) at (-1.75, -0.75) {$\alpha$};
		\node [style=none] (9) at (-1.75, 0.25) {$\vdots$};
		\node [style=none] (10) at (5.25, 0.85) {};
		\node [style=none] (11) at (5.25, -0.75) {};
		\node [style=Z phase dot] (12) at (2.75, 2.6) {$\pm\alpha$};
		\node [style=Z dot] (13) at (2.75, 1.6) {};
		\node [style=Z phase dot] (14) at (4, 0.85) {$\alpha$};
		\node [style=Z phase dot] (15) at (4, -0.75) {$\alpha$};
		\node [style=none] (16) at (4, 0.25) {$\vdots$};
		\node [style=Z phase dot] (17) at (1.25, 2.6) {$\beta\mp\alpha$};
		\node [style=Z dot] (18) at (1.25, 1.6) {};
		\node [style=none] (19) at (-1, 1.35) {};
		\node [style=none] (20) at (-1, -1) {};
		\node [style=none] (21) at (4.75, 1.35) {};
		\node [style=none] (22) at (4.75, -1) {};
	\end{pgfonlayer}
	\begin{pgfonlayer}{edgelayer}
		\draw (6) to (0.center);
		\draw (2.center) to (8);
		\draw [style=hadamard edge] (8) to (5);
		\draw [style=hadamard edge] (6) to (5);
		\draw [style=hadamard edge] (5) to (4);
		\draw (14) to (10.center);
		\draw (11.center) to (15);
		\draw [style=hadamard edge] (15) to (13);
		\draw [style=hadamard edge] (14) to (13);
		\draw [style=hadamard edge] (13) to (12);
		\draw [style=hadamard edge] (18) to (17);
		\draw [style=hadamard edge] (14) to (18);
		\draw [style=hadamard edge] (15) to (18);
		\draw [style=double edge] (19.center) to (20.center);
		\draw [style=double edge] (21.center) to (22.center);
	\end{pgfonlayer}
\end{tikzpicture}
	\end{equation}
\end{minipage}\\

With~\Cref{eq:output-unfusion}, we extract the unwanted part of an output phase as R${}_z$ gate, and with~\Cref{eq:gadget-unfusion}, we split up phase gadgets into a part with the desired phase and another gadget so that the sum of the phases yields the original one. Both rewrites are sound in ZX-calculus: The first corresponds to an application of the fusion rule as mentioned in~\Cref{sec:background} and the second is a reversed version of the gadget fusion rule as shown in~\cite[Section D]{kissingerReducingNumberNonClifford2020}. With adjusted phases, we extract the entire graph structure by removing it from the diagram and placing a C${}_n$P$(\varphi)$ gate with $\varphi=\alpha\cdot 2^{n-1}$ on the circuit. \\
We can extend this procedure by also allowing the extraction of graph structures where some phase gadgets required for a C${}_n$P extraction are missing in the diagram. Consider the following example, where we are initially missing two 2-ary phase gadgets with $-\frac{\pi}{4}$ to complete a C${}_2$P$(\pi)$ structure:
\begin{equation}\label{eq:complete-mcp-structure}
	\begin{tikzpicture}
	\begin{pgfonlayer}{nodelayer}
		\node [style=Z dot] (20) at (9, -1) {};
		\node [style=Z dot] (21) at (9, 1.5) {};
		\node [style=Z dot] (22) at (9, 0.25) {};
		\node [style=vertex] (23) at (10.75, 1.5) {};
		\node [style=vertex] (24) at (10.75, 0.25) {};
		\node [style=vertex] (25) at (10.75, -1) {};
		\node [style=none] (33) at (10, 2.5) {};
		\node [style=none] (34) at (10, -1.25) {};
		\node [style=none] (35) at (11.25, -0.5) {$\pi$};
		\node [style=none] (36) at (12, 1.5) {};
		\node [style=none] (37) at (12, 0.25) {};
		\node [style=none] (38) at (12, -1) {};
		\node [style=Z phase dot] (45) at (-8.75, -1.25) {$\frac{\pi}{4}$};
		\node [style=Z phase dot] (46) at (-8.75, 1.25) {$\frac{\pi}{4}$};
		\node [style=Z phase dot] (47) at (-8.75, 0) {$\frac{\pi}{4}$};
		\node [style=none] (48) at (-7.25, 1.25) {};
		\node [style=none] (49) at (-7.25, 0) {};
		\node [style=none] (50) at (-7.25, -1.25) {};
		\node [style=none] (53) at (-7.75, 2) {};
		\node [style=none] (54) at (-7.75, -1.5) {};
		\node [style=none] (56) at (5.25, 0) {$=$};
		\node [style=Z dot] (59) at (-11.25, 2.25) {};
		\node [style=Z phase dot] (60) at (-11.25, 3.5) {$-\frac{\pi}{4}$};
		\node [style=Z dot] (61) at (-12.5, 2.25) {};
		\node [style=Z phase dot] (62) at (-12.5, 3.5) {$\frac{\pi}{4}$};
		\node [style=Z phase dot] (63) at (2.25, -1.25) {$\frac{\pi}{4}$};
		\node [style=Z phase dot] (64) at (2.25, 1.25) {$\frac{\pi}{4}$};
		\node [style=Z phase dot] (65) at (2.25, 0) {$\frac{\pi}{4}$};
		\node [style=none] (66) at (3.75, 1.25) {};
		\node [style=none] (67) at (3.75, 0) {};
		\node [style=none] (68) at (3.75, -1.25) {};
		\node [style=none] (69) at (3.25, 2) {};
		\node [style=none] (70) at (3.25, -1.5) {};
		\node [style=Z dot] (72) at (-4, 2.25) {};
		\node [style=Z phase dot] (73) at (-4, 3.5) {$-\frac{\pi}{4}$};
		\node [style=Z dot] (74) at (-5.25, 2.25) {};
		\node [style=Z phase dot] (75) at (-5.25, 3.5) {$\frac{\pi}{4}$};
		\node [style=none] (76) at (-6, 0) {$=$};
		\node [style=Z dot] (77) at (-2.5, 2.25) {};
		\node [style=Z phase dot] (78) at (-2.5, 3.5) {$-\frac{\pi}{4}$};
		\node [style=Z dot] (79) at (-1.25, 2.25) {};
		\node [style=Z phase dot] (80) at (-1.25, 3.5) {$\frac{\pi}{4}$};
		\node [style=Z dot] (81) at (0, 2.25) {};
		\node [style=Z phase dot] (82) at (0, 3.5) {$-\frac{\pi}{4}$};
		\node [style=Z dot] (83) at (1.25, 2.25) {};
		\node [style=Z phase dot] (84) at (1.25, 3.5) {$\frac{\pi}{4}$};
		\node [style=Z dot] (87) at (6, 2.25) {};
		\node [style=Z phase dot] (88) at (6, 3.5) {$\frac{\pi}{4}$};
		\node [style=Z dot] (91) at (7.5, 2.25) {};
		\node [style=Z phase dot] (92) at (7.5, 3.5) {$\frac{\pi}{4}$};
	\end{pgfonlayer}
	\begin{pgfonlayer}{edgelayer}
		\draw (20) to (25);
		\draw (21) to (23);
		\draw (22) to (24);
		\draw [style=double edge] (33.center) to (34.center);
		\draw (23) to (24);
		\draw (23) to (36.center);
		\draw (37.center) to (24);
		\draw (38.center) to (25);
		\draw (24) to (25);
		\draw (45) to (50.center);
		\draw (46) to (48.center);
		\draw (47) to (49.center);
		\draw [style=double edge] (53.center) to (54.center);
		\draw [style=hadamard edge] (60) to (59);
		\draw [style=hadamard edge] (62) to (61);
		\draw [style=hadamard edge] (59) to (46);
		\draw [style=hadamard edge] (59) to (45);
		\draw [style=hadamard edge] (46) to (61);
		\draw [style=hadamard edge] (61) to (47);
		\draw [style=hadamard edge] (45) to (61);
		\draw (63) to (68.center);
		\draw (64) to (66.center);
		\draw (65) to (67.center);
		\draw [style=double edge] (69.center) to (70.center);
		\draw [style=hadamard edge] (73) to (72);
		\draw [style=hadamard edge] (75) to (74);
		\draw [style=hadamard edge] (72) to (64);
		\draw [style=hadamard edge] (72) to (63);
		\draw [style=hadamard edge] (64) to (74);
		\draw [style=hadamard edge] (74) to (65);
		\draw [style=hadamard edge] (63) to (74);
		\draw [style=hadamard edge] (78) to (77);
		\draw [style=hadamard edge] (80) to (79);
		\draw [style=hadamard edge] (82) to (81);
		\draw [style=hadamard edge] (84) to (83);
		\draw [style=hadamard edge] (77) to (65);
		\draw [style=hadamard edge] (65) to (79);
		\draw [style=hadamard edge] (63) to (77);
		\draw [style=hadamard edge] (63) to (79);
		\draw [style=hadamard edge] (65) to (81);
		\draw [style=hadamard edge] (65) to (83);
		\draw [style=hadamard edge] (64) to (81);
		\draw [style=hadamard edge] (64) to (83);
		\draw [style=hadamard edge] (88) to (87);
		\draw [style=hadamard edge] (92) to (91);
		\draw [style=hadamard edge] (87) to (22);
		\draw [style=hadamard edge] (20) to (87);
		\draw [style=hadamard edge] (91) to (21);
		\draw [style=hadamard edge] (22) to (91);
	\end{pgfonlayer}
\end{tikzpicture}
\end{equation}
By adding pairs of phase gadgets with opposite phases corresponding to the identity, we complete the required graph structure to extract the gate. Some inserted gadgets then remain in the diagram and are extracted later. 
If we always take the gadget with the most neighbors, complete the diagram to match a C${}_n$P gate, and extract it, every phase gadget will get extracted as part of a C${}_n$P gate at some point and we can entirely omit $YZ$ spider eliminations via pivoting. 

\subsubsection{Preservation of gflow}
For a complete translation of graph-like diagrams into quantum circuits, it is essential that all operations preserve gflow on the diagram. The rewrites of Equations 3-7 essentially reduce to deletions and insertions of phase gadgets, i.e., $YZ$ measurements, connected to only outputs. Since the original extraction algorithm preserves gflow and it has been shown in ~\cite[Lemma 3.4.]{backens-there-2021} that the deletion of arbitrary $YZ$ measurements preserves gflow, the same remains to be shown for the insertion case: 
\begin{lemma}[Insertion of $YZ$ measurements on outputs]
	Let $(g,\prec)$ be a gflow for $(G,I,O,\lambda)$ and let $W\subseteq O$. Then $(G',I,O,\lambda')$, where $G' = (V',E')$ with $V'=V\cup \{x\}$, $\lambda'(x) = YZ$ and $E'=E\cup \{(x,w) | w\in W\}$ has a gflow.
	\begin{proof}
		We provide the detailed proof in \Cref{sec:gadget-insertion}.
	\end{proof}
\end{lemma}

\subsubsection{Time complexity}
\label{sec:time-complexity}
The time complexity of the proposed approach in terms of elementary graph operations depends on whether we allow additional insertions of phase gadgets or not. Let $k$ denote the number of spiders in a diagram and $n$ the number of outputs: 
\begin{itemize}
	\item If we do not allow additional insertions of phase gadgets, we have approximately the same runtime as the original algorithm, namely $O(n^2k^2+k^3)$ which is summed from the runtime for Gaussian elimination $O(n^2k)$, pivoting $YZ$ measurements $O(k^2)$ and $k$ steps in total~\cite{backens-there-2021}. Additionally, for our approach, we have to split at most $k$ phase gadgets at each step, which adds another $O(nk)$ term to the elementary graph operations. Yet, this term gets absorbed by the complexity of the Gaussian elimination.
	\item If we allow phase gadget insertions to complete the graph structures corresponding to a C${}_n$P gate, the complexity essentially becomes $O(2^{n+1}k)$. This is because, in the worst case, we would complete structures where there is only a single phase gadget connected to all $n$ outputs, and we need to add $2\cdot 2^n -n-2$  additional gadgets. We want to emphasize that this worst-case complexity is unlikely to occur in practice. 
 Yet, for larger circuits, it may be useful to limit the size of extractable C${}_n$P gates to a constant.
\end{itemize}

\section{Neutral Atom Circuit Synthesis}
\label{sec:evaluation}
In the following, we want to apply the proposed extraction scheme to circuit synthesis for neutral atom (NA)-based hardware due to their native support of $\mathrm{C}_n\mathrm{P}$ gates.
Therefore, we briefly introduce the hardware capabilities~\cite{schmidComputationalCapabilitiesCompiler2023}, embed our proposed scheme into a complete synthesis procedure, and evaluate the effect of the $\mathrm{C}_n\mathrm{P}$ extraction regarding the circuit execution time in comparison to Qiskits internal synthesis algorithm.

\subsection{Neutral Atom Background}
\label{sec:neutr-atom-backgr}

For NA-based quantum computers, qubit registers are realized by placing single atoms in optical dipole traps created by laser beams, referred to as optical lattices or optical tweezers.
While arbitrary atom arrangements are possible, we assume a rectangular grid as illustrated in \Cref{fig:neutral-atoms}.
The qubit states can be encoded in long-lived internal atomic states such as hyperfine or nuclear spin states.
Commonly employed atomic species include alkali or alkaline-earth(-like) atoms such as Rb and Sr, which provide suitable internal states with long coherence times.
Gates are realized with specific laser pulses on the atoms using uniform global beams, with the possibility of addressing a whole register or individual qubits~\cite{grahamMultiqubitEntanglementAlgorithms2022,levineDispersiveOpticalSystems2022}.
Multi-qubit gates are based on the long-range interaction between close-by atoms excited to high-lying Rydberg states~\cite{mullerMesoscopicRydbergGate2009,isenhowerMultibitCkNOTQuantum2011,everedHighfidelityParallelEntangling2023,levineParallelImplementationHighFidelity2019,saffmanQuantumInformationRydberg2010,saffman2016quantum,morgado2021quantum}.
There exist different protocols to realize both single- and multi-qubit gates, and the preferred implementation depends on many parameters such as the chosen atom species, the respective qubit encoding, and the experimental setup.
In this work, we focus on individually addressable $\mathrm{C}_{n} \mathrm{P}(\varphi)$ gates between neighboring atoms as a generalization of the often implemented CZ gate, which can be realized by tuning the accumulated phase during the Rydberg interaction to an arbitrary angle $\varphi$ instead of $\pi$~\cite{levineParallelImplementationHighFidelity2019}.
This might even result in improved gate times and fidelity due to the shorter time of the atom spent in the Rydberg state~\cite{everedHighfidelityParallelEntangling2023}.
Regarding single-qubit gates, we assume a scheme between fast, individually addressable AC Stark shift beams, realizing local $R_{z}$ rotations and slow, globally addressing microwave pulses performing a rotation about an axis in the XY-plane for all qubits simultaneously.
In particular, we assume the following gate definitions, where single-qubit gates are equivalent to the ones from \cite{nottinghamDecomposingRoutingQuantum2023} with GR as global XY-rotations applied to all $n$~qubits, and $\hat{Y}$ being the Pauli-Y matrix:
\begin{align}
  \label{eq:gate-definitions}
  \begin{split}
  \mathrm{C}_{n} \mathrm{P}(\varphi) \equiv \mathrm{diag}(1,\, \dots,1,\, e^{i \varphi}), \quad \mathrm{R}_{z}(\gamma^{\pm}) \equiv \mathrm{diag}(e^{-i\gamma^{\pm} /2} ,e^{-i\gamma^{\pm} /2} ), \quad
    \mathrm{GR}(\theta_{\mathrm{max}})\equiv  \exp \left( - i \frac{\theta}{2}\, \sum_{i=1}^{n} \hat{Y}_{i}  \right)
     \end{split}
\end{align}
In this setting, sets of arbitrary but simultaneous single-qubit gates on different qubits can be converted into two global illuminations interleaved with single-qubit Z-rotations.
On unused qubits, the two complementary global rotations cancel out, effectively applying an identity operation.
To convert single-qubit gates to this setting, we consider the transversal decomposition scheme introduced in~\cite{nottinghamDecomposingRoutingQuantum2023} which is optimal in terms of global pulse time. An illustration of the gate capabilities and how to synthesize the respective single- and multi-qubit gates is shown in \Cref{fig:neutral-atoms}.

\begin{figure}
	\centering
	\includegraphics[width=0.99\textwidth]{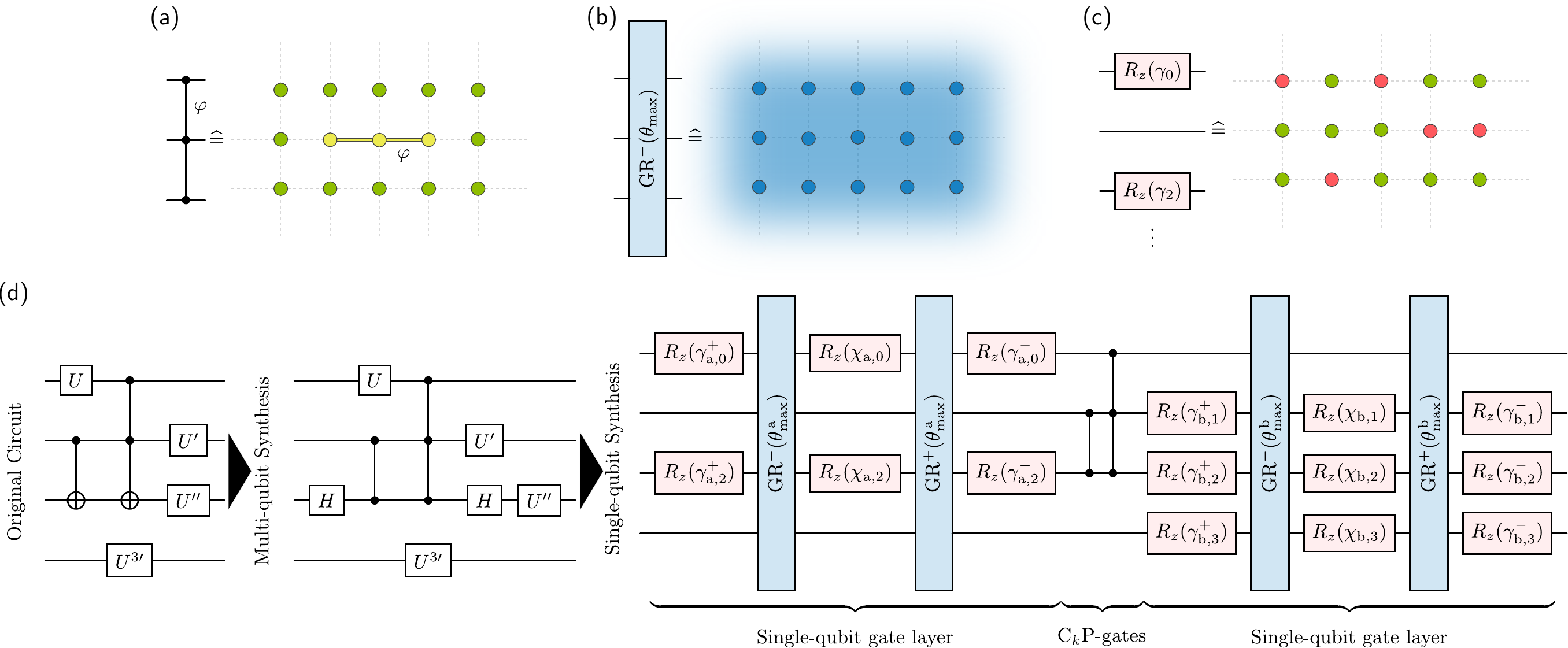}
	\caption{Illustration of the neutral atom gate capabilities and the process of synthesizing and scheduling quantum circuits to the hardware. \textbf{(a)} Native (multi-) controlled phase gates ($\mathrm{C}_{n} \mathrm{P}(\varphi)$), here shown for three qubits. \textbf{(b)} Global single-qubit rotations in the XY-plane. \textbf{(c)} Individually addressable Z-rotations ($\mathrm{R}_{z}(\gamma)$). \textbf{(d)} Synthesis to alternating single- and multi-qubit layers: First, the synthesis of multi-qubit gates to $\mathrm{C}_{n} \mathrm{P}$ gates. Second, the synthesis and scheduling of single-qubit gates into global XY-rotations and individually addressable Z-rotations according to the transversal decomposition of \cite{nottinghamDecomposingRoutingQuantum2023}.}
	\label{fig:neutral-atoms}
\end{figure}

\subsection{Related Work}
\label{sec:related-work}
Recently, there has been a fast development of NA-specific compilation methods~\cite{bakerExploitingLongDistanceInteractions2021,liTimingAwareQubitMapping2023,patelGeyserCompilationFramework2022,brandhoferOptimalMappingNearTerm2021,tanCompilingQuantumCircuits2023,nottinghamDecomposingRoutingQuantum2023,schmidHybridCircuitMapping2023,wangFPQACCompilationFramework2023,wangQPilotFieldProgrammable2023,tanDepthOptimalAddressing2D2024} focusing almost exclusively on the mapping and scheduling tasks within the compilation.
The only exceptions are \cite{nottinghamDecomposingRoutingQuantum2023}, proposing the single-qubit synthesis also used in this work, and the Geyser framework~\cite{patelGeyserCompilationFramework2022} using numerical optimization to introduce additional CCZ gates.
Nevertheless, both neglect the capability of NAs to natively execute controlled gates with arbitrary phases and, furthermore, the ability to directly execute multi-controlled gates for more than one control qubit.
Although convenient, this neglects potentially shorter or simpler circuits using these specific capabilities of NAs.
Therefore, in the following, we evaluate five different synthesis schemes to compare the commonly used naive synthesis algorithms with an NA-specific synthesis employing the ZX-approach of~\Cref{sec:algorithm}.

\subsection{Evaluation Setup}
\label{sec:evaluation-setup}
The total gate synthesis process consists of the two steps illustrated in \Cref{fig:neutral-atoms}:
First, the synthesis of the multi-qubit gates into $\mathrm{C}_{n} \mathrm{P}$ gates and arbitrary single-qubit rotations, and second, the synthesis and scheduling of single-qubit gates into the alternating global vs. local scheme, resulting in the fully synthesized circuit containing only gates from the native gate set of \Cref{eq:gate-definitions}.
For the first step, we consider the following five schemes:
\begin{enumerate}
  \itemsep0em
  \item \textbf{Qiskit-default:} Circuits are converted into the Qiskit-supported native gate set of $\{\mathrm{U}_{3}, CZ\}$ using the internal \texttt{transpile} function and setting the optimization level to three. This approach uses Qiskit internal schemes to decompose gates with more than two-qubit gates.
  \item \textbf{No-decomp:} The Qiskit decomposition introduces a large overhead that can be bypassed for NAs. For better comparison, we propose this scheme which synthesizes to $\{\mathrm{U}_{3}, \mathrm{C}_{n}\mathrm{Z}\}$ by replacing all (multi-) controlled gates by their $\mathrm{C}_{n}\mathrm{Z}\}$ equivalent and using the Qiskit-default approach for the single qubit gates. 
  \item \textbf{ZX-default:} The circuit is converted into a graph-like ZX diagram, and the default extraction algorithm of PyZX~\cite{kissinger2020Pyzx} is used to recreate a circuit.
  \item \textbf{ZX-no-insert:} Similar to the default but the extraction scheme from \Cref{sec:algorithm} is used to synthesize $\mathrm{C}_{n}\mathrm{P}$ gates.
  \item \textbf{ZX-with-insert:} In addition to ZX-no-insert, we allow the insertion of additional phase gadgets, resulting in possibly more and larger phase gate extractions.
\end{enumerate}
Since the ZX extraction sometimes produces redundant gates, we additionally apply a basic gate cancellation algorithm afterwards.
In the second step, the transverse decomposition according to~\cite{nottinghamDecomposingRoutingQuantum2023} is used to synthesize the single qubit gates.
In this scheme, $\theta_{\mathrm{max}} \equiv \max_{i}\theta_{i}$ is defined as the maximum of the first Euler angle $\theta_i$ of any single qubit rotation in this layer.
According to the discussion in \cite{nottinghamDecomposingRoutingQuantum2023}, the total gate execution time scales linearly in this angle with the maximal duration at $\theta_{\mathrm{max}}=\pi$. \\
As for many single-qubit gates the actual moment of execution is not unique, it can be added to different layers.
We thus use an additional greedy optimization step, not performed in~\cite{nottinghamDecomposingRoutingQuantum2023}, to check the possible positions of the single-qubit gates and assign them such as to minimize the overall $\theta_{\mathrm{max}}$.
In particular, a gate with Euler angle $\theta$ is preferably assigned to a layer with $\theta_{\mathrm{max}} > \theta$, allowing the gate to be executed without increasing the gate time.
As evaluation metrics, we compute both simple gate counts and the total circuit execution time $T$ by scheduling the gates according to the illustration in \Cref{fig:neutral-atoms}. We sum individual gate times, where we assume the gate execution time increases linearly with the rotation angle as follows:
\begin{equation}
	T = \sum_{i=0}^{d}\frac{|max_\gamma(R_z(\gamma),i)|}{\pi}100ns \; + \; \sum_{GR(\theta_{\mathrm{max}})}\frac{|\theta_{\mathrm{max}}|}{\pi}100\mu s \; + \; \sum_{C_1P(\varphi)}\frac{|\varphi|}{\pi}100ns \; + \; \sum_{C_nP(\varphi),n>1}\frac{|\varphi|}{\pi}400ns
\end{equation}

Here $d$ denotes circuit depth, meaning we assume full parallel execution of the $R_z$ if possible by taking the maximum angle of each layer. Multi-qubit gates are assumed to be executed in a sequential way.
The gate times are $\SI{0.1}{\micro\s}$ for the $R_{z}$~\cite{shawMultiensembleMetrologyProgramming2024} and the $\mathrm{CP}(\varphi=\pi)$ gate~\cite{madjarovHighfidelityEntanglementDetection2020}.
For all higher-weight controlled phase gates $\mathrm{C}_{\geq 2}\mathrm{P}$ we assume $\SI{0.4}{\micro\s}$ for $\varphi=\pi$.
The dominating factor for circuit execution time are the slow global illuminations GR with $\SI{100}{\micro\s}$~\cite{shawMultiensembleMetrologyProgramming2024}.
Therefore, our main aim to use the scheme of \Cref{sec:algorithm} is to lower the number of global GR gates and, in this way, reduce to overall execution time.

For a comprehensive and rigorous evaluation, we chose circuits from three different benchmark collections, with their descriptions available online: 
\textbf{QASM-Bench(small)}~\cite{liQASMBenchLowlevelQASM2022} and \textbf{MQT-Bench}~\cite{quetschlichMQTBenchBenchmarking2023} contain various low-level benchmark circuits of different sizes and types with common quantum subroutines and algorithms.
Additionally, we also consider the \textbf{Feynman-Bench}~\cite{amyLargescaleFunctionalVerification2019} collection. Created for formal methods, it contains different arithmetic circuits usually based on Toffoli gates. 

The code used for the evaluations is available with an MIT license at Zenodo~\cite{staudachercode} allowing reproducibility and possible usage or integration into other compilation projects.

\subsection{Results \& Discussion}
\label{sec:results}

\begin{table}[t]
\resizebox{\textwidth}{!}{\begin{threeparttable}
		\centering
		\caption{Averaged reduction of execution time $T$ relative to Qiskit-default. Negative percentages indicate an increased execution time.}
		\label{tab:avg-time-reduction}
		\begin{tabular}{@{}lcccccc@{}}
			\toprule
			&Circuits & \multicolumn{1}{|c}{Qiskit} &  \multicolumn{1}{c}{No-decomp}  &  \multicolumn{1}{|c}{ZX-default}  &  \multicolumn{1}{c}{ZX-no-insert}  &  \multicolumn{1}{c}{ZX-with-insert}    \\
			\midrule
			QASM-Bench~\cite{liQASMBenchLowlevelQASM2022}~${}^1$                  & 35 & $0\%$ & $8\%$ & $2\%$ & $14\%$ & $26\%$ \\
			MQT-Bench~\cite{quetschlichMQTBenchBenchmarking2023}~${}^2$           & 11 & $0\%$ & $0\%$ & $-44\%$ & $-16\%$ & $26\%$  \\
			Feynman-Bench~\cite{amyLargescaleFunctionalVerification2019}~${}^3$    & 26 & $0\%$ & $63\%$ & $-23\%$ & $-16\%$ & $40\%$ \\
			\bottomrule

		\end{tabular}
		\begin{tablenotes}\footnotesize
			\item ${}^1$~\href{https://github.com/pnnl/QASMBench}{https://github.com/pnnl/QASMBench} \hfill ${}^2$~\href{https://www.cda.cit.tum.de/mqtbench/}{https://www.cda.cit.tum.de/mqtbench/} \hfill ${}^3$~\href{https://github.com/meamy/feynman}{https://github.com/meamy/feynman}
		\end{tablenotes}
	\end{threeparttable}}
	\vspace{-0.5cm}
	
\end{table}

\begin{table}[t]
	\centering
	\caption{Evaluation results for six benchmarks, selected to illustrate both good and poor performance. Numbers after the names indicate the corresponding benchmark collection. The first table shows gate counts corresponding to the native gate set of \Cref{eq:gate-definitions}. The second table contains the total execution time $T [\SI{}{\milli\s}]$ and the synthesis algorithm runtime $r [\SI{}{\s}]$ on a consumer notebook.}
	\resizebox{\textwidth}{!}{
		\begin{tabular}{@{}lccccccccccccccccccccccccc@{}}
			\toprule
			&  \multicolumn{11}{|c|}{ZX}  &  \multicolumn{2}{c|}{Qiskit} & \multicolumn{3}{c}{Own alternative}  \\
			&  \multicolumn{2}{|c}{Default}  &  \multicolumn{3}{c}{No-insert}  &  \multicolumn{6}{c|}{With-insert}  &  \multicolumn{2}{c|}{Default}  &  \multicolumn{3}{c}{No-decomp}  \\
			\midrule
			&
			\multicolumn{1}{|c}{\rotatebox{90}{GR}} &
			\multicolumn{1}{c|}{\rotatebox{90}{CP}} &
			\rotatebox{90}{GR} &
			\rotatebox{90}{CP} &
			\multicolumn{1}{c|}{\rotatebox{90}{$\mathrm{C}_{2}\mathrm{P}$}} &
			\rotatebox{90}{GR} &
			\rotatebox{90}{CP} &
			\rotatebox{90}{$\mathrm{C}_2\mathrm{P}$} &
			\rotatebox{90}{$\mathrm{C}_3\mathrm{P}$} &
			\rotatebox{90}{$\mathrm{C}_4\mathrm{P}$} &
			\multicolumn{1}{c|}{\rotatebox{90}{$\mathrm{C}_{5}\mathrm{P}$}} &
			\rotatebox{90}{GR} &
			\multicolumn{1}{c|}{\rotatebox{90}{CP}} &
			\rotatebox{90}{GR} &
			\rotatebox{90}{CP} &
			\multicolumn{1}{c}{\rotatebox{90}{$\mathrm{C}_{2}\mathrm{P}$}} \\
			\midrule
			\midrule
			hhl_n7~(1)                      & \multicolumn{1}{|c}{448}   & \multicolumn{1}{c|}{296} &          362 & 241   & \multicolumn{1}{c|}{- }  & \bfseries 306 & 207 & 42  & 29  & 6 & \multicolumn{1}{c|}{-} & 356 & \multicolumn{1}{c|}{196} &            356 & 196 & -   \\
			qft_10~(2)                      & \multicolumn{1}{|c}{90 }   & \multicolumn{1}{c|}{140 } &          36  & 75    & \multicolumn{1}{c|}{- }  & \bfseries 14  & 62  & 14  & -   & - & \multicolumn{1}{c|}{-} & 44 & \multicolumn{1}{c|}{105} &            44 & 105 & -   \\
			qnn_10~(2)                      & \multicolumn{1}{|c}{106}   & \multicolumn{1}{c|}{334} &          62  & 199   & \multicolumn{1}{c|}{- }  & \bfseries 28  & 159 & 48  & 26  & 8 & \multicolumn{1}{c|}{1} & 76 & \multicolumn{1}{c|}{188} &            76 & 188 &-    \\
			gf2\textasciicircum 7\_mult~(3) & \multicolumn{1}{|c}{214}   & \multicolumn{1}{c|}{956} &          134 & 447   & \multicolumn{1}{c|}{19}  & \bfseries 14  & 22  & 114 & -   & - & \multicolumn{1}{c|}{-} & 194 & \multicolumn{1}{c|}{300} &            18  & 6   & 49  \\
			rc_adder_6~(3)                  & \multicolumn{1}{|c}{76 }   & \multicolumn{1}{c|}{100 } &          84 & 95    & \multicolumn{1}{c|}{- }   &           62  & 91  & 4   & -   & - & \multicolumn{1}{c|}{-} & 86 & \multicolumn{1}{c|}{93 } & \bfseries 28   & 27  & 11  \\
			qcla_adder_10~(3)               & \multicolumn{1}{|c}{128}   & \multicolumn{1}{c|}{331} &          138 & 462   & \multicolumn{1}{c|}{1 }  &           24  & 153 & 121 & -   & - & \multicolumn{1}{c|}{-} & 74 & \multicolumn{1}{c|}{233} & \bfseries  18  & 29  & 34  \\
			\bottomrule
		\end{tabular}
	}
	\bigskip

	\resizebox{\textwidth}{!}{
		\begin{tabular}{@{}lSSSSSSSSSSc@{}}
			\toprule
			&  \multicolumn{6}{|c|}{ZX}  &  \multicolumn{2}{c|}{Qiskit} &  \multicolumn{2}{c}{Own alternative} \\
			&  \multicolumn{2}{|c}{Default}  &  \multicolumn{2}{c}{No-insert}  &  \multicolumn{2}{c|}{With-insert}  &  \multicolumn{2}{c|}{Default}  &  \multicolumn{2}{c}{No-decomp} &  \\
			\midrule
			&
			\multicolumn{1}{|c}{$T$} & \multicolumn{1}{c|}{$r$} & $T$ & \multicolumn{1}{c|}{$r$} & $T$ & \multicolumn{1}{c|}{$r$} & $T$ & \multicolumn{1}{c|}{$r$} & $T$ & \hspace{10pt}$r$ \\

			\midrule
			\midrule
			hhl_n7~(1)                      & \multicolumn{1}{|S}{11.07 }  & \multicolumn{1}{S|}{0.085}&           8.71 & \multicolumn{1}{S|}{0.156}  &           7.70  & \multicolumn{1}{S|}{0.239} & 5.45 & \multicolumn{1}{S|}{0.122} & \bfseries 5.45 & 0.277    \\
			qft_10~(2)                      & \multicolumn{1}{|S}{2.16  } & \multicolumn{1}{S|}{0.036} &           0.91 & \multicolumn{1}{S|}{0.029}  & \bfseries 0.35  & \multicolumn{1}{S|}{0.039} & 0.89 & \multicolumn{1}{S|}{0.043} &           0.89 & 0.103   \\
			qnn_10~(2)                      & \multicolumn{1}{|S}{3.12  } & \multicolumn{1}{S|}{0.031} &           1.53 & \multicolumn{1}{S|}{0.108}  & \bfseries 0.74  & \multicolumn{1}{S|}{0.214} & 3.06 & \multicolumn{1}{S|}{0.117} &           3.06 & 0.151   \\
			gf2\textasciicircum 7\_mult~(3) & \multicolumn{1}{|S}{5.84  } & \multicolumn{1}{S|}{0.218} &           3.31 & \multicolumn{1}{S|}{1.286}  & \bfseries 0.40  & \multicolumn{1}{S|}{3.421} & 4.20 & \multicolumn{1}{S|}{0.083} &           0.47  & 0.041     \\
			rc_adder_6~(3)                  & \multicolumn{1}{|S}{1.89  } & \multicolumn{1}{S|}{0.030} &           2.04 & \multicolumn{1}{S|}{0.049}  &           1.56  & \multicolumn{1}{S|}{0.052} & 1.67 & \multicolumn{1}{S|}{0.033} & \bfseries 0.71  & 0.041    \\
			qcla_adder_10~(3)               & \multicolumn{1}{|S}{3.26  } & \multicolumn{1}{S|}{1.326} &           3.48 & \multicolumn{1}{S|}{0.448}  &           0.66  & \multicolumn{1}{S|}{2.063} & 1.63 & \multicolumn{1}{S|}{0.073} & \bfseries 0.47  & 0.078    \\
			\bottomrule
		\end{tabular}

	}
	\label{tab:results-gate-count}
\end{table}
The five compilation schemes are evaluated on gate count and execution time $T [\SI{}{\milli\s}]$ of the synthesized circuits, together with the algorithm runtime $r [\SI{}{\s}]$.
We first discuss time reduction averaged over all circuits, summarized in \Cref{tab:avg-time-reduction}, then, we highlight six examples shown in~\Cref{tab:results-gate-count}, which have been selected to best illustrate different cases within the dataset.
The full dataset with all raw data is available at Zenodo~\cite{staudachercode}.
On the QASM-Bench collection $(1)$, the ZX-with-insert approach results in an average $26\%$ reduction of execution time compared to the Qiskit internal synthesis improving 23 of the 35 circuits.
Similar numbers result for the MQT-Bench $(2)$ circuits with $26\%$ reduction of execution time, improving 7 out of 11 circuits.
For the Feynman benchmarks $(3)$, results are mixed: While our scheme achieves a $40\%$ reduction compared to Qiskit, improving 24 of 26 circuits, the No-decomp scheme has an even higher average reduction of execution time with $63\%$.

Considering the above part of \Cref{tab:results-gate-count} one can see how the synthesis approach described in this work is capable of successfully synthesizing $\mathrm{C}_{n}\mathrm{P}$ gates.
Since No-decomp just replaces multi-controlled gates by their $\mathrm{C}_{n}\mathrm{Z}\}$ equivalent, the corresponding column indicates the number of multi-controlled gates present in the original circuit. In comparison, one can then see that while No-insert only resynthesizes a few of the original gates, With-insert synthesizes more multi-qubit gates and is often able to create even higher-dimensioned controlled gates.
This higher-controlled gate synthesis appears very dominantly in dense circuits such as qnn\_10, corresponding to a quantum neural network circuit, but also in circuits that are natively built on controlled phase gates such as HHL.

Due to the controlled gates, the proposed approach is capable of effectively reducing the number of slow global GR gates in comparison to the regular ZX extraction scheme and Qiskit, which are not capable of synthesizing multi-controlled gates.
Generally, a lower number of GR gates also results in a shorter circuit execution time with some exceptions, such as the hhl\_n7, where the ZX synthesis has a longer execution time than the No-decomp scheme, although the number of absolute GR pulses is lower. This is likely due to an increased pulse time of the individual GR gates.

The ZX approaches do not perform well on circuits that already contain close to optimal multi-controlled gates, for instance on circuits of the Feynman benchmark.
Here, the approaches extract gates in a less efficient way, resulting in a gate and time overhead. In such cases, replacing multi-controlled gates without changing circuit structure as in the No-decomp scheme is the best option. This can also be seen when comparing the number of GR gates for the two adder circuits to the No-decomp scheme, where the ZX approaches are not capable of reconstructing a similar efficient circuit structure. Since all ZX strategies yield inefficient circuits, it may be that in such cases the ZX-diagram simplification creates too complex graph structures.

Regarding algorithmic runtime, ZX performs similarly to the Qiskit internal synthesis.
Note, however, the significant increase in runtime for qcla\_adder and gf2\textasciicircum t\_mult for the With-insert synthesis.
This is likely the overhead due to the insertion of additional phase gadgets, resulting in worst-case exponential runtime as discussed in \Cref{sec:time-complexity}.

\section{Conclusion}
In this work we introduced a novel approach to synthesize quantum circuits to the universal gate set $\{$H,R${}_z$,C${}_n$P$\}$.
As a key contribution, our approach is able to efficiently identify structures in graph-like ZX-diagrams that correspond to multi-controlled phase gates and extract them to quantum circuits. This allows the synthesis of such gates even if they were not present in the original circuit. Together with existing simplification strategies for ZX-diagrams, our approach can be used to synthesize arbitrary quantum circuits towards neutral atom architectures. Here, our synthesis often trades slow global pulse rotations for fast multi-controlled qubit gates and we are thus able to reduce execution time significantly for many common circuits.
Further, this could also help hardware developers to evaluate whether increasing the number of qubits supported by multi-controlled phase gates is beneficial for certain problems in terms of execution time and fidelity.
In cases where the circuit already consists of optimized multi-controlled gates, such as circuits based on arithmetic functions, the synthesis may result in less efficient quantum circuits. This is likely due to overly complex graph structures resulting from ZX-diagram simplification. \\
We leave it as a topic for further research whether in those cases more sophisticated strategies allow exploiting the phase gadget structures for multi-controlled phase gates synthesis without increasing the underlying graph structure complexity. Possible approaches include advanced heuristics applying the proposed scheme only to cases where it is likely to improve the circuit structure. 
We also want to mention that a similar synthesis approach could be done without ZX-calculus using the Pauli Dependency DAG representation of quantum circuits~\cite{simmons2021}. It has been shown that the diagram simplification rules from Section~\ref{sec:diag-simp} are equivalent to reordering Pauli terms in a Pauli Dependency DAG and by identifying patterns of individual Pauli-Z terms similar to Theorem~\ref{theo:mcp-gates} we can then synthesize C$_{n}$P gates. As future work, it would be interesting to see how these two versions compare.

\section*{Acknowledgments}
This work is partially supported by the German Federal Ministry of Education and Research (BMBF) under the funding program Quantum Technologies - From Basic Research to Market under contract number 13N16070.

The authors acknowledge funding from the Munich Quantum Valley initiative (K3, K5), which is supported by the Bavarian state government with funds from the Hightech Agenda Bayern Plus.

J.Z. acknowledges funding by the Max Planck Society (MPG), the Deutsche Forschungsgemeinschaft (DFG, German Research Foundation) under Germany’s Excellence Strategy EXC-2111-39081486, and acknowledges support from the German Federal Ministry of Education and Research (BMBF) through the program ``Quantum technologies - from basic research to market'' (SNAQC, Grant No. 13N16265).

L.S. and R.W. acknowledge funding from the European Research Council (ERC) under the European Union's Horizon 2020 research and innovation program (Grant Agreement No. 101001318).

The authors would like to thank Nathalia Nottingham for providing an implementation of the transversal layer decomposition and Miriam Backens for helpful discussions.

\bibliographystyle{eptcs}
\bibliography{refs}

\begin{thebibliography}{10}
\providecommand{\bibitemdeclare}[2]{}
\providecommand{\surnamestart}{}
\providecommand{\surnameend}{}
\providecommand{\urlprefix}{Available at }
\providecommand{\url}[1]{\texttt{#1}}
\providecommand{\href}[2]{\texttt{#2}}
\providecommand{\urlalt}[2]{\href{#1}{#2}}
\providecommand{\doi}[1]{doi:\urlalt{https://doi.org/#1}{#1}}
\providecommand{\eprint}[1]{arXiv:\urlalt{https://arxiv.org/abs/#1}{#1}}
\providecommand{\bibinfo}[2]{#2}

\bibitemdeclare{article}{amyLargescaleFunctionalVerification2019}
\bibitem{amyLargescaleFunctionalVerification2019}
\bibinfo{author}{Matthew \surnamestart Amy\surnameend} (\bibinfo{year}{2019}):
  \emph{\bibinfo{title}{Towards Large-scale Functional Verification of
  Universal Quantum Circuits}}.
\newblock {\slshape \bibinfo{journal}{Electronic Proceedings in Theoretical
  Computer Science}} \bibinfo{volume}{287}, pp. \bibinfo{pages}{1--21},
  \doi{10.4204/EPTCS.287.1}.

\bibitemdeclare{article}{amy2019}
\bibitem{amy2019}
\bibinfo{author}{Matthew \surnamestart Amy\surnameend},
  \bibinfo{author}{Parsiad \surnamestart Azimzadeh\surnameend} \&
  \bibinfo{author}{Michele \surnamestart Mosca\surnameend}
  (\bibinfo{year}{2018}): \emph{\bibinfo{title}{On the controlled-NOT
  complexity of controlled-NOT–phase circuits}}.
\newblock {\slshape \bibinfo{journal}{Quantum Science and Technology}}
  \bibinfo{volume}{4}(\bibinfo{number}{1}), p. \bibinfo{pages}{015002},
  \doi{10.1088/2058-9565/aad8ca}.
\newblock \urlprefix\url{https://dx.doi.org/10.1088/2058-9565/aad8ca}.

\bibitemdeclare{article}{amyImprovedSynthesisToffoliHadamard2023}
\bibitem{amyImprovedSynthesisToffoliHadamard2023}
\bibinfo{author}{Matthew \surnamestart Amy\surnameend},
  \bibinfo{author}{Andrew~N. \surnamestart Glaudell\surnameend},
  \bibinfo{author}{Sarah~Meng \surnamestart Li\surnameend} \&
  \bibinfo{author}{Neil~J. \surnamestart Ross\surnameend}
  (\bibinfo{year}{2023}): \emph{\bibinfo{title}{Improved {{Synthesis}}
  of~{{Toffoli-Hadamard Circuits}}}}, pp. \bibinfo{pages}{169--209}.
\newblock \doi{10.1007/978-3-031-38100-3_12}.
\newblock \eprint{2305.11305}.

\bibitemdeclare{article}{backens-there-2021}
\bibitem{backens-there-2021}
\bibinfo{author}{Miriam \surnamestart Backens\surnameend},
  \bibinfo{author}{Hector \surnamestart Miller-Bakewell\surnameend},
  \bibinfo{author}{Giovanni \surnamestart de~Felice\surnameend},
  \bibinfo{author}{Leo \surnamestart Lobski\surnameend} \&
  \bibinfo{author}{John \surnamestart van~de Wetering\surnameend}
  (\bibinfo{year}{2021}): \emph{\bibinfo{title}{There and back again: {A}
  circuit extraction tale}}.
\newblock {\slshape \bibinfo{journal}{Quantum}} \bibinfo{volume}{5}, p.
  \bibinfo{pages}{421}, \doi{10.22331/q-2021-03-25-421}.
\newblock \urlprefix\url{http://arxiv.org/abs/2003.01664}.
\newblock \bibinfo{note}{ArXiv:2003.01664 [quant-ph]}.

\bibitemdeclare{article}{bakerExploitingLongDistanceInteractions2021}
\bibitem{bakerExploitingLongDistanceInteractions2021}
\bibinfo{author}{Jonathan~M. \surnamestart Baker\surnameend},
  \bibinfo{author}{Andrew \surnamestart Litteken\surnameend},
  \bibinfo{author}{Casey \surnamestart Duckering\surnameend},
  \bibinfo{author}{Henry \surnamestart Hoffmann\surnameend},
  \bibinfo{author}{Hannes \surnamestart Bernien\surnameend} \&
  \bibinfo{author}{Frederic~T. \surnamestart Chong\surnameend}
  (\bibinfo{year}{2021}): \emph{\bibinfo{title}{Exploiting {{Long-Distance
  Interactions}} and {{Tolerating Atom Loss}} in {{Neutral Atom Quantum
  Architectures}}}}, pp. \bibinfo{pages}{818--831}.
\newblock \doi{10.1109/ISCA52012.2021.00069}.

\bibitemdeclare{article}{barredoAtombyatomAssemblerDefectfree2016}
\bibitem{barredoAtombyatomAssemblerDefectfree2016}
\bibinfo{author}{Daniel \surnamestart Barredo\surnameend},
  \bibinfo{author}{Sylvain \surnamestart {de L{\'e}s{\'e}leuc}\surnameend},
  \bibinfo{author}{Vincent \surnamestart Lienhard\surnameend},
  \bibinfo{author}{Thierry \surnamestart Lahaye\surnameend} \&
  \bibinfo{author}{Antoine \surnamestart Browaeys\surnameend}
  (\bibinfo{year}{2016}): \emph{\bibinfo{title}{An Atom-by-Atom Assembler of
  Defect-Free Arbitrary Two-Dimensional Atomic Arrays}}.
\newblock {\slshape \bibinfo{journal}{Science}}
  \bibinfo{volume}{354}(\bibinfo{number}{6315}), pp.
  \bibinfo{pages}{1021--1023}, \doi{10.1126/science.aah3778}.

\bibitemdeclare{article}{bluvsteinLogicalQuantumProcessor2023}
\bibitem{bluvsteinLogicalQuantumProcessor2023}
\bibinfo{author}{Dolev \surnamestart Bluvstein\surnameend},
  \bibinfo{author}{Simon~J. \surnamestart Evered\surnameend},
  \bibinfo{author}{Alexandra~A. \surnamestart Geim\surnameend},
  \bibinfo{author}{Sophie~H. \surnamestart Li\surnameend},
  \bibinfo{author}{Hengyun \surnamestart Zhou\surnameend}, \bibinfo{author}{Tom
  \surnamestart Manovitz\surnameend}, \bibinfo{author}{Sepehr \surnamestart
  Ebadi\surnameend}, \bibinfo{author}{Madelyn \surnamestart Cain\surnameend},
  \bibinfo{author}{Marcin \surnamestart Kalinowski\surnameend},
  \bibinfo{author}{Dominik \surnamestart Hangleiter\surnameend},
  \bibinfo{author}{J.~Pablo~Bonilla \surnamestart Ataides\surnameend},
  \bibinfo{author}{Nishad \surnamestart Maskara\surnameend},
  \bibinfo{author}{Iris \surnamestart Cong\surnameend}, \bibinfo{author}{Xun
  \surnamestart Gao\surnameend}, \bibinfo{author}{Pedro~Sales \surnamestart
  Rodriguez\surnameend}, \bibinfo{author}{Thomas \surnamestart
  Karolyshyn\surnameend}, \bibinfo{author}{Giulia \surnamestart
  Semeghini\surnameend}, \bibinfo{author}{Michael~J. \surnamestart
  Gullans\surnameend}, \bibinfo{author}{Markus \surnamestart
  Greiner\surnameend}, \bibinfo{author}{Vladan \surnamestart
  Vuleti{\'c}\surnameend} \& \bibinfo{author}{Mikhail~D. \surnamestart
  Lukin\surnameend} (\bibinfo{year}{2023}): \emph{\bibinfo{title}{Logical
  Quantum Processor Based on Reconfigurable Atom Arrays}}.
\newblock {\slshape \bibinfo{journal}{Nature}}, pp. \bibinfo{pages}{1--3},
  \doi{10.1038/s41586-023-06927-3}.

\bibitemdeclare{article}{bluvsteinQuantumProcessorBased2022}
\bibitem{bluvsteinQuantumProcessorBased2022}
\bibinfo{author}{Dolev \surnamestart Bluvstein\surnameend},
  \bibinfo{author}{Harry \surnamestart Levine\surnameend},
  \bibinfo{author}{Giulia \surnamestart Semeghini\surnameend},
  \bibinfo{author}{Tout~T. \surnamestart Wang\surnameend},
  \bibinfo{author}{Sepehr \surnamestart Ebadi\surnameend},
  \bibinfo{author}{Marcin \surnamestart Kalinowski\surnameend},
  \bibinfo{author}{Alexander \surnamestart Keesling\surnameend},
  \bibinfo{author}{Nishad \surnamestart Maskara\surnameend},
  \bibinfo{author}{Hannes \surnamestart Pichler\surnameend},
  \bibinfo{author}{Markus \surnamestart Greiner\surnameend},
  \bibinfo{author}{Vladan \surnamestart Vuleti{\'c}\surnameend} \&
  \bibinfo{author}{Mikhail~D. \surnamestart Lukin\surnameend}
  (\bibinfo{year}{2022}): \emph{\bibinfo{title}{A Quantum Processor Based on
  Coherent Transport of Entangled Atom Arrays}}.
\newblock {\slshape \bibinfo{journal}{Nature}}
  \bibinfo{volume}{604}(\bibinfo{number}{7906}), pp. \bibinfo{pages}{451--456},
  \doi{10.1038/s41586-022-04592-6}.

\bibitemdeclare{inproceedings}{brandhoferOptimalMappingNearTerm2021}
\bibitem{brandhoferOptimalMappingNearTerm2021}
\bibinfo{author}{Sebastian \surnamestart Brandhofer\surnameend},
  \bibinfo{author}{Ilia \surnamestart Polian\surnameend} \&
  \bibinfo{author}{Hans~Peter \surnamestart B{\"u}chler\surnameend}
  (\bibinfo{year}{2021}): \emph{\bibinfo{title}{Optimal {{Mapping}} for
  {{Near-Term Quantum Architectures}} Based on {{Rydberg Atoms}}}}.
\newblock In: {\slshape \bibinfo{booktitle}{2021 {{IEEE}}/{{ACM International
  Conference On Computer Aided Design}} ({{ICCAD}})}}, pp.
  \bibinfo{pages}{1--7}, \doi{10.1109/ICCAD51958.2021.9643490}.

\bibitemdeclare{article}{briegel2009measurement}
\bibitem{briegel2009measurement}
\bibinfo{author}{Hans~J \surnamestart Briegel\surnameend},
  \bibinfo{author}{David~E \surnamestart Browne\surnameend},
  \bibinfo{author}{Wolfgang \surnamestart D{\"u}r\surnameend},
  \bibinfo{author}{Robert \surnamestart Raussendorf\surnameend} \&
  \bibinfo{author}{Maarten \surnamestart Van~den Nest\surnameend}
  (\bibinfo{year}{2009}): \emph{\bibinfo{title}{Measurement-based quantum
  computation}}.
\newblock {\slshape \bibinfo{journal}{Nature Physics}}
  \bibinfo{volume}{5}(\bibinfo{number}{1}), pp. \bibinfo{pages}{19--26},
  \doi{10.1038/nphys1157}.

\bibitemdeclare{article}{browne2007generalized}
\bibitem{browne2007generalized}
\bibinfo{author}{Daniel~E \surnamestart Browne\surnameend},
  \bibinfo{author}{Elham \surnamestart Kashefi\surnameend},
  \bibinfo{author}{Mehdi \surnamestart Mhalla\surnameend} \&
  \bibinfo{author}{Simon \surnamestart Perdrix\surnameend}
  (\bibinfo{year}{2007}): \emph{\bibinfo{title}{Generalized flow and
  determinism in measurement-based quantum computation}}.
\newblock {\slshape \bibinfo{journal}{New Journal of Physics}}
  \bibinfo{volume}{9}(\bibinfo{number}{8}), p. \bibinfo{pages}{250},
  \doi{10.1088/1367-2630/9/8/250}.
\newblock \urlprefix\url{https://dx.doi.org/10.1088/1367-2630/9/8/250}.

\bibitemdeclare{misc}{cao2024multiqubit}
\bibitem{cao2024multiqubit}
\bibinfo{author}{Alec \surnamestart Cao\surnameend},
  \bibinfo{author}{William~J. \surnamestart Eckner\surnameend},
  \bibinfo{author}{Theodor~Lukin \surnamestart Yelin\surnameend},
  \bibinfo{author}{Aaron~W. \surnamestart Young\surnameend},
  \bibinfo{author}{Sven \surnamestart Jandura\surnameend},
  \bibinfo{author}{Lingfeng \surnamestart Yan\surnameend},
  \bibinfo{author}{Kyungtae \surnamestart Kim\surnameend},
  \bibinfo{author}{Guido \surnamestart Pupillo\surnameend},
  \bibinfo{author}{Jun \surnamestart Ye\surnameend},
  \bibinfo{author}{Nelson~Darkwah \surnamestart Oppong\surnameend} \&
  \bibinfo{author}{Adam~M. \surnamestart Kaufman\surnameend}
  (\bibinfo{year}{2024}): \emph{\bibinfo{title}{Multi-qubit gates and
  'Schr\"odinger cat' states in an optical clock}}.
\newblock \eprint{2402.16289}.

\bibitemdeclare{book}{coecke2017picturing}
\bibitem{coecke2017picturing}
\bibinfo{author}{Bob \surnamestart Coecke\surnameend} \& \bibinfo{author}{Aleks
  \surnamestart Kissinger\surnameend} (\bibinfo{year}{2017}):
  \emph{\bibinfo{title}{Picturing Quantum Processes}}.
\newblock \bibinfo{publisher}{Cambridge University Press},
  \doi{10.1017/9781316219317}.

\bibitemdeclare{article}{dlaskaQuantumOptimizationFourBody2022}
\bibitem{dlaskaQuantumOptimizationFourBody2022}
\bibinfo{author}{Clemens \surnamestart Dlaska\surnameend},
  \bibinfo{author}{Kilian \surnamestart Ender\surnameend},
  \bibinfo{author}{Glen~Bigan \surnamestart Mbeng\surnameend},
  \bibinfo{author}{Andreas \surnamestart Kruckenhauser\surnameend},
  \bibinfo{author}{Wolfgang \surnamestart Lechner\surnameend} \&
  \bibinfo{author}{Rick \surnamestart {van Bijnen}\surnameend}
  (\bibinfo{year}{2022}): \emph{\bibinfo{title}{Quantum {{Optimization}} via
  {{Four-Body Rydberg Gates}}}}.
\newblock {\slshape \bibinfo{journal}{Physical Review Letters}}
  \bibinfo{volume}{128}(\bibinfo{number}{12}), p. \bibinfo{pages}{120503},
  \doi{10.1103/PhysRevLett.128.120503}.

\bibitemdeclare{article}{duncan-graph-theoretic-2020}
\bibitem{duncan-graph-theoretic-2020}
\bibinfo{author}{Ross \surnamestart Duncan\surnameend}, \bibinfo{author}{Aleks
  \surnamestart Kissinger\surnameend}, \bibinfo{author}{Simon \surnamestart
  Perdrix\surnameend} \& \bibinfo{author}{John van~de \surnamestart
  Wetering\surnameend} (\bibinfo{year}{2020}):
  \emph{\bibinfo{title}{Graph-theoretic {Simplification} of {Quantum}
  {Circuits} with the {ZX}-calculus}}.
\newblock {\slshape \bibinfo{journal}{Quantum}} \bibinfo{volume}{4}, p.
  \bibinfo{pages}{279}, \doi{10.22331/q-2020-06-04-279}.
\newblock \urlprefix\url{https://quantum-journal.org/papers/q-2020-06-04-279/}.
\newblock \bibinfo{note}{Publisher: Verein zur Förderung des Open Access
  Publizierens in den Quantenwissenschaften}.

\bibitemdeclare{article}{everedHighfidelityParallelEntangling2023}
\bibitem{everedHighfidelityParallelEntangling2023}
\bibinfo{author}{Simon~J. \surnamestart Evered\surnameend},
  \bibinfo{author}{Dolev \surnamestart Bluvstein\surnameend},
  \bibinfo{author}{Marcin \surnamestart Kalinowski\surnameend},
  \bibinfo{author}{Sepehr \surnamestart Ebadi\surnameend}, \bibinfo{author}{Tom
  \surnamestart Manovitz\surnameend}, \bibinfo{author}{Hengyun \surnamestart
  Zhou\surnameend}, \bibinfo{author}{Sophie~H. \surnamestart Li\surnameend},
  \bibinfo{author}{Alexandra~A. \surnamestart Geim\surnameend},
  \bibinfo{author}{Tout~T. \surnamestart Wang\surnameend},
  \bibinfo{author}{Nishad \surnamestart Maskara\surnameend},
  \bibinfo{author}{Harry \surnamestart Levine\surnameend},
  \bibinfo{author}{Giulia \surnamestart Semeghini\surnameend},
  \bibinfo{author}{Markus \surnamestart Greiner\surnameend},
  \bibinfo{author}{Vladan \surnamestart Vuleti{\'c}\surnameend} \&
  \bibinfo{author}{Mikhail~D. \surnamestart Lukin\surnameend}
  (\bibinfo{year}{2023}): \emph{\bibinfo{title}{High-Fidelity Parallel
  Entangling Gates on a Neutral-Atom Quantum Computer}}.
\newblock {\slshape \bibinfo{journal}{Nature}}
  \bibinfo{volume}{622}(\bibinfo{number}{7982}), pp. \bibinfo{pages}{268--272},
  \doi{10.1038/s41586-023-06481-y}.
\newblock \eprint{2304.05420}.

\bibitemdeclare{article}{grahamMultiqubitEntanglementAlgorithms2022}
\bibitem{grahamMultiqubitEntanglementAlgorithms2022}
\bibinfo{author}{T.~M. \surnamestart Graham\surnameend},
  \bibinfo{author}{Y.~\surnamestart Song\surnameend},
  \bibinfo{author}{J.~\surnamestart Scott\surnameend},
  \bibinfo{author}{C.~\surnamestart Poole\surnameend},
  \bibinfo{author}{L.~\surnamestart Phuttitarn\surnameend},
  \bibinfo{author}{K.~\surnamestart Jooya\surnameend},
  \bibinfo{author}{P.~\surnamestart Eichler\surnameend},
  \bibinfo{author}{X.~\surnamestart Jiang\surnameend},
  \bibinfo{author}{A.~\surnamestart Marra\surnameend},
  \bibinfo{author}{B.~\surnamestart Grinkemeyer\surnameend},
  \bibinfo{author}{M.~\surnamestart Kwon\surnameend},
  \bibinfo{author}{M.~\surnamestart Ebert\surnameend},
  \bibinfo{author}{J.~\surnamestart Cherek\surnameend}, \bibinfo{author}{M.~T.
  \surnamestart Lichtman\surnameend}, \bibinfo{author}{M.~\surnamestart
  Gillette\surnameend}, \bibinfo{author}{J.~\surnamestart Gilbert\surnameend},
  \bibinfo{author}{D.~\surnamestart Bowman\surnameend},
  \bibinfo{author}{T.~\surnamestart Ballance\surnameend},
  \bibinfo{author}{C.~\surnamestart Campbell\surnameend},
  \bibinfo{author}{E.~D. \surnamestart Dahl\surnameend},
  \bibinfo{author}{O.~\surnamestart Crawford\surnameend},
  \bibinfo{author}{N.~S. \surnamestart Blunt\surnameend},
  \bibinfo{author}{B.~\surnamestart Rogers\surnameend},
  \bibinfo{author}{T.~\surnamestart Noel\surnameend} \&
  \bibinfo{author}{M.~\surnamestart Saffman\surnameend} (\bibinfo{year}{2022}):
  \emph{\bibinfo{title}{Multi-Qubit Entanglement and Algorithms on a
  Neutral-Atom Quantum Computer}}.
\newblock {\slshape \bibinfo{journal}{Nature}}
  \bibinfo{volume}{604}(\bibinfo{number}{7906}), pp. \bibinfo{pages}{457--462},
  \doi{10.1038/s41586-022-04603-6}.

\bibitemdeclare{article}{grosse2009exact}
\bibitem{grosse2009exact}
\bibinfo{author}{Daniel \surnamestart Gro{\ss}e\surnameend},
  \bibinfo{author}{Robert \surnamestart Wille\surnameend},
  \bibinfo{author}{Gerhard~W \surnamestart Dueck\surnameend} \&
  \bibinfo{author}{Rolf \surnamestart Drechsler\surnameend}
  (\bibinfo{year}{2009}): \emph{\bibinfo{title}{Exact multiple-control Toffoli
  network synthesis with SAT techniques}}.
\newblock {\slshape \bibinfo{journal}{IEEE Transactions on Computer-Aided
  Design of Integrated Circuits and Systems}}
  \bibinfo{volume}{28}(\bibinfo{number}{5}), pp. \bibinfo{pages}{703--715},
  \doi{10.1109/TCAD.2009.2017215}.

\bibitemdeclare{misc}{gygerContinuousOperationLargescale2024}
\bibitem{gygerContinuousOperationLargescale2024}
\bibinfo{author}{Flavien \surnamestart Gyger\surnameend},
  \bibinfo{author}{Maximilian \surnamestart Ammenwerth\surnameend},
  \bibinfo{author}{Renhao \surnamestart Tao\surnameend},
  \bibinfo{author}{Hendrik \surnamestart Timme\surnameend},
  \bibinfo{author}{Stepan \surnamestart Snigirev\surnameend},
  \bibinfo{author}{Immanuel \surnamestart Bloch\surnameend} \&
  \bibinfo{author}{Johannes \surnamestart Zeiher\surnameend}
  (\bibinfo{year}{2024}): \emph{\bibinfo{title}{Continuous Operation of
  Large-Scale Atom Arrays in Optical Lattices}}.
\newblock \eprint{2402.04994}.

\bibitemdeclare{article}{henrietQuantumComputingNeutral2020}
\bibitem{henrietQuantumComputingNeutral2020}
\bibinfo{author}{Lo\"{\i}c \surnamestart Henriet\surnameend},
  \bibinfo{author}{Lucas \surnamestart Beguin\surnameend},
  \bibinfo{author}{Adrien \surnamestart Signoles\surnameend},
  \bibinfo{author}{Thierry \surnamestart Lahaye\surnameend},
  \bibinfo{author}{Antoine \surnamestart Browaeys\surnameend},
  \bibinfo{author}{Georges-Olivier \surnamestart Reymond\surnameend} \&
  \bibinfo{author}{Christophe \surnamestart Jurczak\surnameend}
  (\bibinfo{year}{2020}): \emph{\bibinfo{title}{Quantum Computing with Neutral
  Atoms}}.
\newblock {\slshape \bibinfo{journal}{Quantum}} \bibinfo{volume}{4}, p.
  \bibinfo{pages}{327}, \doi{10.22331/q-2020-09-21-327}.

\bibitemdeclare{article}{isenhowerMultibitCkNOTQuantum2011}
\bibitem{isenhowerMultibitCkNOTQuantum2011}
\bibinfo{author}{L.~\surnamestart Isenhower\surnameend},
  \bibinfo{author}{M.~\surnamestart Saffman\surnameend} \&
  \bibinfo{author}{K.~\surnamestart M{\o}lmer\surnameend}
  (\bibinfo{year}{2011}): \emph{\bibinfo{title}{Multibit {{CkNOT}} Quantum
  Gates via {{Rydberg}} Blockade}}.
\newblock {\slshape \bibinfo{journal}{Quantum Information Processing}}
  \bibinfo{volume}{10}(\bibinfo{number}{6}), p. \bibinfo{pages}{755},
  \doi{10.1007/s11128-011-0292-4}.

\bibitemdeclare{article}{Jandura2022timeoptimaltwothree}
\bibitem{Jandura2022timeoptimaltwothree}
\bibinfo{author}{Sven \surnamestart Jandura\surnameend} \&
  \bibinfo{author}{Guido \surnamestart Pupillo\surnameend}
  (\bibinfo{year}{2022}): \emph{\bibinfo{title}{Time-{O}ptimal {T}wo- and
  {T}hree-{Q}ubit {G}ates for {R}ydberg {A}toms}}.
\newblock {\slshape \bibinfo{journal}{{Quantum}}} \bibinfo{volume}{6}, p.
  \bibinfo{pages}{712}, \doi{10.22331/q-2022-05-13-712}.

\bibitemdeclare{article}{kissingerReducingNumberNonClifford2020}
\bibitem{kissingerReducingNumberNonClifford2020}
\bibinfo{author}{Aleks \surnamestart Kissinger\surnameend} \&
  \bibinfo{author}{John \surnamestart {van de Wetering}\surnameend}
  (\bibinfo{year}{2020}): \emph{\bibinfo{title}{Reducing the Number of
  Non-{{Clifford}} Gates in Quantum Circuits}}.
\newblock {\slshape \bibinfo{journal}{Physical Review A}}
  \bibinfo{volume}{102}(\bibinfo{number}{2}), p. \bibinfo{pages}{022406},
  \doi{10.1103/PhysRevA.102.022406}.

\bibitemdeclare{article}{kissinger2020Pyzx}
\bibitem{kissinger2020Pyzx}
\bibinfo{author}{Aleks \surnamestart Kissinger\surnameend} \&
  \bibinfo{author}{John \surnamestart van~de Wetering\surnameend}
  (\bibinfo{year}{2020}): \emph{\bibinfo{title}{{PyZX: Large Scale Automated
  Diagrammatic Reasoning}}} \bibinfo{volume}{318}, pp.
  \bibinfo{pages}{229--241}.
\newblock \doi{10.4204/EPTCS.318.14}.

\bibitemdeclare{article}{kliuchnikov2013fast}
\bibitem{kliuchnikov2013fast}
\bibinfo{author}{Vadym \surnamestart Kliuchnikov\surnameend},
  \bibinfo{author}{Dmitri \surnamestart Maslov\surnameend} \&
  \bibinfo{author}{Michele \surnamestart Mosca\surnameend}
  (\bibinfo{year}{2013}): \emph{\bibinfo{title}{Fast and efficient exact
  synthesis of single-qubit unitaries generated by clifford and T gates}}.
\newblock {\slshape \bibinfo{journal}{Quantum Information \& Computation}}
  \bibinfo{volume}{13}(\bibinfo{number}{7-8}), pp. \bibinfo{pages}{607--630},
  \doi{10.5555/2535649.2535653}.

\bibitemdeclare{misc}{kuijpers2019graphical}
\bibitem{kuijpers2019graphical}
\bibinfo{author}{Stach \surnamestart Kuijpers\surnameend},
  \bibinfo{author}{John \surnamestart van~de Wetering\surnameend} \&
  \bibinfo{author}{Aleks \surnamestart Kissinger\surnameend}:
  \emph{\bibinfo{title}{Graphical fourier theory and the cost of quantum
  addition}}.
\newblock \urlprefix\url{https://doi.org/10.48550/arXiv.1904.07551}.

\bibitemdeclare{article}{levineDispersiveOpticalSystems2022}
\bibitem{levineDispersiveOpticalSystems2022}
\bibinfo{author}{Harry \surnamestart Levine\surnameend}, \bibinfo{author}{Dolev
  \surnamestart Bluvstein\surnameend}, \bibinfo{author}{Alexander \surnamestart
  Keesling\surnameend}, \bibinfo{author}{Tout~T. \surnamestart
  Wang\surnameend}, \bibinfo{author}{Sepehr \surnamestart Ebadi\surnameend},
  \bibinfo{author}{Giulia \surnamestart Semeghini\surnameend},
  \bibinfo{author}{Ahmed \surnamestart Omran\surnameend},
  \bibinfo{author}{Markus \surnamestart Greiner\surnameend},
  \bibinfo{author}{Vladan \surnamestart Vuleti{\'c}\surnameend} \&
  \bibinfo{author}{Mikhail~D. \surnamestart Lukin\surnameend}
  (\bibinfo{year}{2022}): \emph{\bibinfo{title}{Dispersive Optical Systems for
  Scalable {{Raman}} Driving of Hyperfine Qubits}}.
\newblock {\slshape \bibinfo{journal}{Physical Review A}}
  \bibinfo{volume}{105}(\bibinfo{number}{3}), p. \bibinfo{pages}{032618},
  \doi{10.1103/PhysRevA.105.032618}.

\bibitemdeclare{article}{levineParallelImplementationHighFidelity2019}
\bibitem{levineParallelImplementationHighFidelity2019}
\bibinfo{author}{Harry \surnamestart Levine\surnameend},
  \bibinfo{author}{Alexander \surnamestart Keesling\surnameend},
  \bibinfo{author}{Giulia \surnamestart Semeghini\surnameend},
  \bibinfo{author}{Ahmed \surnamestart Omran\surnameend},
  \bibinfo{author}{Tout~T. \surnamestart Wang\surnameend},
  \bibinfo{author}{Sepehr \surnamestart Ebadi\surnameend},
  \bibinfo{author}{Hannes \surnamestart Bernien\surnameend},
  \bibinfo{author}{Markus \surnamestart Greiner\surnameend},
  \bibinfo{author}{Vladan \surnamestart Vuleti{\'c}\surnameend},
  \bibinfo{author}{Hannes \surnamestart Pichler\surnameend} \&
  \bibinfo{author}{Mikhail~D. \surnamestart Lukin\surnameend}
  (\bibinfo{year}{2019}): \emph{\bibinfo{title}{Parallel {{Implementation}} of
  {{High-Fidelity Multiqubit Gates}} with {{Neutral Atoms}}}}.
\newblock {\slshape \bibinfo{journal}{Physical Review Letters}}
  \bibinfo{volume}{123}(\bibinfo{number}{17}), p. \bibinfo{pages}{170503},
  \doi{10.1103/PhysRevLett.123.170503}.

\bibitemdeclare{misc}{liQASMBenchLowlevelQASM2022}
\bibitem{liQASMBenchLowlevelQASM2022}
\bibinfo{author}{Ang \surnamestart Li\surnameend}, \bibinfo{author}{Samuel
  \surnamestart Stein\surnameend}, \bibinfo{author}{Sriram \surnamestart
  Krishnamoorthy\surnameend} \& \bibinfo{author}{James \surnamestart
  Ang\surnameend} (\bibinfo{year}{2022}): \emph{\bibinfo{title}{{{QASMBench}}:
  {{A Low-level QASM Benchmark Suite}} for {{NISQ Evaluation}} and
  {{Simulation}}}}, \doi{10.48550/arXiv.2005.13018}.
\newblock \eprint{2005.13018}.

\bibitemdeclare{article}{liTimingAwareQubitMapping2023}
\bibitem{liTimingAwareQubitMapping2023}
\bibinfo{author}{Yongshang \surnamestart Li\surnameend},
  \bibinfo{author}{Yu~\surnamestart Zhang\surnameend}, \bibinfo{author}{Mingyu
  \surnamestart Chen\surnameend}, \bibinfo{author}{Xiangyang \surnamestart
  Li\surnameend} \& \bibinfo{author}{Peng \surnamestart Xu\surnameend}
  (\bibinfo{year}{2023}): \emph{\bibinfo{title}{Timing-{{Aware Qubit Mapping}}
  and {{Gate Scheduling Adapted}} to {{Neutral Atom Quantum Computing}}}}.
\newblock {\slshape \bibinfo{journal}{IEEE Transactions on Computer-Aided
  Design of Integrated Circuits and Systems}}, pp. \bibinfo{pages}{1--1},
  \doi{10.1109/TCAD.2023.3261244}.

\bibitemdeclare{article}{madjarovHighfidelityEntanglementDetection2020}
\bibitem{madjarovHighfidelityEntanglementDetection2020}
\bibinfo{author}{Ivaylo~S. \surnamestart Madjarov\surnameend},
  \bibinfo{author}{Jacob~P. \surnamestart Covey\surnameend},
  \bibinfo{author}{Adam~L. \surnamestart Shaw\surnameend},
  \bibinfo{author}{Joonhee \surnamestart Choi\surnameend},
  \bibinfo{author}{Anant \surnamestart Kale\surnameend},
  \bibinfo{author}{Alexandre \surnamestart Cooper\surnameend},
  \bibinfo{author}{Hannes \surnamestart Pichler\surnameend},
  \bibinfo{author}{Vladimir \surnamestart Schkolnik\surnameend},
  \bibinfo{author}{Jason~R. \surnamestart Williams\surnameend} \&
  \bibinfo{author}{Manuel \surnamestart Endres\surnameend}
  (\bibinfo{year}{2020}): \emph{\bibinfo{title}{High-Fidelity Entanglement and
  Detection of Alkaline-Earth {{Rydberg}} Atoms}}.
\newblock {\slshape \bibinfo{journal}{Nature Physics}}
  \bibinfo{volume}{16}(\bibinfo{number}{8}), pp. \bibinfo{pages}{857--861},
  \doi{10.1038/s41567-020-0903-z}.

\bibitemdeclare{article}{morgado2021quantum}
\bibitem{morgado2021quantum}
\bibinfo{author}{M~\surnamestart Morgado\surnameend} \&
  \bibinfo{author}{S~\surnamestart Whitlock\surnameend} (\bibinfo{year}{2021}):
  \emph{\bibinfo{title}{Quantum simulation and computing with
  Rydberg-interacting qubits}}.
\newblock {\slshape \bibinfo{journal}{AVS Quantum Science}}
  \bibinfo{volume}{3}(\bibinfo{number}{2}), \doi{10.1116/5.0036562}.
\newblock
  \eprint{https://pubs.aip.org/avs/aqs/article-pdf/doi/10.1116/5.0036562/19739152/023501\_1\_online.pdf}.

\bibitemdeclare{article}{mukhopadhyay2024synthesizing}
\bibitem{mukhopadhyay2024synthesizing}
\bibinfo{author}{Priyanka \surnamestart Mukhopadhyay\surnameend}
  (\bibinfo{year}{2024}): \emph{\bibinfo{title}{Synthesizing Toffoli-optimal
  quantum circuits for arbitrary multi-qubit unitaries}}.
\newblock {\slshape \bibinfo{journal}{arXiv preprint arXiv:2401.08950}}.
\newblock \urlprefix\url{https://doi.org/10.48550/arXiv.2401.08950}.

\bibitemdeclare{article}{mullerMesoscopicRydbergGate2009}
\bibitem{mullerMesoscopicRydbergGate2009}
\bibinfo{author}{M.~\surnamestart M{\"u}ller\surnameend},
  \bibinfo{author}{I.~\surnamestart Lesanovsky\surnameend},
  \bibinfo{author}{H.~\surnamestart Weimer\surnameend}, \bibinfo{author}{H.~P.
  \surnamestart B{\"u}chler\surnameend} \& \bibinfo{author}{P.~\surnamestart
  Zoller\surnameend} (\bibinfo{year}{2009}): \emph{\bibinfo{title}{Mesoscopic
  {{Rydberg Gate Based}} on {{Electromagnetically Induced Transparency}}}}.
\newblock {\slshape \bibinfo{journal}{Physical Review Letters}}
  \bibinfo{volume}{102}(\bibinfo{number}{17}), p. \bibinfo{pages}{170502},
  \doi{10.1103/PhysRevLett.102.170502}.

\bibitemdeclare{misc}{norciaIterativeAssembly1712024}
\bibitem{norciaIterativeAssembly1712024}
\bibinfo{author}{M.~A. \surnamestart Norcia\surnameend},
  \bibinfo{author}{H.~\surnamestart Kim\surnameend}, \bibinfo{author}{W.~B.
  \surnamestart Cairncross\surnameend}, \bibinfo{author}{M.~\surnamestart
  Stone\surnameend}, \bibinfo{author}{A.~\surnamestart Ryou\surnameend},
  \bibinfo{author}{M.~\surnamestart Jaffe\surnameend}, \bibinfo{author}{M.~O.
  \surnamestart Brown\surnameend}, \bibinfo{author}{K.~\surnamestart
  Barnes\surnameend}, \bibinfo{author}{P.~\surnamestart Battaglino\surnameend},
  \bibinfo{author}{A.~\surnamestart Brown\surnameend},
  \bibinfo{author}{K.~\surnamestart Cassella\surnameend},
  \bibinfo{author}{C.-A. \surnamestart Chen\surnameend},
  \bibinfo{author}{R.~\surnamestart Coxe\surnameend},
  \bibinfo{author}{D.~\surnamestart Crow\surnameend},
  \bibinfo{author}{J.~\surnamestart Epstein\surnameend},
  \bibinfo{author}{C.~\surnamestart Griger\surnameend},
  \bibinfo{author}{E.~\surnamestart Halperin\surnameend},
  \bibinfo{author}{F.~\surnamestart Hummel\surnameend},
  \bibinfo{author}{A.~M.~W. \surnamestart Jones\surnameend},
  \bibinfo{author}{J.~M. \surnamestart Kindem\surnameend},
  \bibinfo{author}{J.~\surnamestart King\surnameend},
  \bibinfo{author}{K.~\surnamestart Kotru\surnameend},
  \bibinfo{author}{J.~\surnamestart Lauigan\surnameend},
  \bibinfo{author}{M.~\surnamestart Li\surnameend},
  \bibinfo{author}{M.~\surnamestart Lu\surnameend},
  \bibinfo{author}{E.~\surnamestart Megidish\surnameend},
  \bibinfo{author}{J.~\surnamestart Marjanovic\surnameend},
  \bibinfo{author}{M.~\surnamestart McDonald\surnameend},
  \bibinfo{author}{T.~\surnamestart Mittiga\surnameend}, \bibinfo{author}{J.~A.
  \surnamestart Muniz\surnameend}, \bibinfo{author}{S.~\surnamestart
  Narayanaswami\surnameend}, \bibinfo{author}{C.~\surnamestart
  Nishiguchi\surnameend}, \bibinfo{author}{T.~\surnamestart Paule\surnameend},
  \bibinfo{author}{K.~A. \surnamestart Pawlak\surnameend},
  \bibinfo{author}{L.~S. \surnamestart Peng\surnameend}, \bibinfo{author}{K.~L.
  \surnamestart Pudenz\surnameend}, \bibinfo{author}{A.~\surnamestart
  Smull\surnameend}, \bibinfo{author}{D.~\surnamestart Stack\surnameend},
  \bibinfo{author}{M.~\surnamestart Urbanek\surnameend},
  \bibinfo{author}{R.~J.~M. \surnamestart {van de Veerdonk}\surnameend},
  \bibinfo{author}{Z.~\surnamestart Vendeiro\surnameend},
  \bibinfo{author}{L.~\surnamestart Wadleigh\surnameend},
  \bibinfo{author}{T.~\surnamestart Wilkason\surnameend},
  \bibinfo{author}{T.-Y. \surnamestart Wu\surnameend},
  \bibinfo{author}{X.~\surnamestart Xie\surnameend},
  \bibinfo{author}{E.~\surnamestart {Zalys-Geller}\surnameend},
  \bibinfo{author}{X.~\surnamestart Zhang\surnameend} \& \bibinfo{author}{B.~J.
  \surnamestart Bloom\surnameend} (\bibinfo{year}{2024}):
  \emph{\bibinfo{title}{Iterative Assembly of \$\^{}\{171\}\${{Yb}} Atom Arrays
  in Cavity-Enhanced Optical Lattices}}, \doi{10.48550/arXiv.2401.16177}.
\newblock \eprint{2401.16177}.

\bibitemdeclare{misc}{nottinghamDecomposingRoutingQuantum2023}
\bibitem{nottinghamDecomposingRoutingQuantum2023}
\bibinfo{author}{Natalia \surnamestart Nottingham\surnameend},
  \bibinfo{author}{Michael~A. \surnamestart Perlin\surnameend},
  \bibinfo{author}{Ryan \surnamestart White\surnameend},
  \bibinfo{author}{Hannes \surnamestart Bernien\surnameend},
  \bibinfo{author}{Frederic~T. \surnamestart Chong\surnameend} \&
  \bibinfo{author}{Jonathan~M. \surnamestart Baker\surnameend}
  (\bibinfo{year}{2023}): \emph{\bibinfo{title}{Decomposing and {{Routing
  Quantum Circuits Under Constraints}} for {{Neutral Atom Architectures}}}},
  \doi{10.48550/arXiv.2307.14996}.
\newblock \eprint{2307.14996}.

\bibitemdeclare{inproceedings}{patelGeyserCompilationFramework2022}
\bibitem{patelGeyserCompilationFramework2022}
\bibinfo{author}{Tirthak \surnamestart Patel\surnameend},
  \bibinfo{author}{Daniel \surnamestart Silver\surnameend} \&
  \bibinfo{author}{Devesh \surnamestart Tiwari\surnameend}
  (\bibinfo{year}{2022}): \emph{\bibinfo{title}{Geyser: A Compilation Framework
  for Quantum Computing with Neutral Atoms}}.
\newblock In: {\slshape \bibinfo{booktitle}{Proceedings of the 49th {{Annual
  International Symposium}} on {{Computer Architecture}}}},
  \bibinfo{series}{{{ISCA}} '22}, \bibinfo{publisher}{{Association for
  Computing Machinery}}, \bibinfo{address}{{New York, NY, USA}}, pp.
  \bibinfo{pages}{383--395}, \doi{10.1145/3470496.3527428}.

\bibitemdeclare{misc}{pauseSuperchargedTwodimensionalTweezer2023}
\bibitem{pauseSuperchargedTwodimensionalTweezer2023}
\bibinfo{author}{Lars \surnamestart Pause\surnameend}, \bibinfo{author}{Lukas
  \surnamestart Sturm\surnameend}, \bibinfo{author}{Marcel \surnamestart
  Mittenb{\"u}hler\surnameend}, \bibinfo{author}{Stephan \surnamestart
  Amann\surnameend}, \bibinfo{author}{Tilman \surnamestart
  Preuschoff\surnameend}, \bibinfo{author}{Dominik \surnamestart
  Sch{\"a}ffner\surnameend}, \bibinfo{author}{Malte \surnamestart
  Schlosser\surnameend} \& \bibinfo{author}{Gerhard \surnamestart
  Birkl\surnameend} (\bibinfo{year}{2023}): \emph{\bibinfo{title}{Supercharged
  Two-Dimensional Tweezer Array with More than 1000 Atomic Qubits}},
  \doi{10.48550/arXiv.2310.09191}.
\newblock \eprint{2310.09191}.

\bibitemdeclare{article}{peham2022equivalence}
\bibitem{peham2022equivalence}
\bibinfo{author}{Tom \surnamestart Peham\surnameend}, \bibinfo{author}{Lukas
  \surnamestart Burgholzer\surnameend} \& \bibinfo{author}{Robert \surnamestart
  Wille\surnameend} (\bibinfo{year}{2022}): \emph{\bibinfo{title}{Equivalence
  checking of quantum circuits with the ZX-calculus}}.
\newblock {\slshape \bibinfo{journal}{IEEE Journal on Emerging and Selected
  Topics in Circuits and Systems}} \bibinfo{volume}{12}(\bibinfo{number}{3}),
  pp. \bibinfo{pages}{662--675}, \doi{10.1109/JETCAS.2022.3202204}.

\bibitemdeclare{article}{quetschlichMQTBenchBenchmarking2023}
\bibitem{quetschlichMQTBenchBenchmarking2023}
\bibinfo{author}{Nils \surnamestart Quetschlich\surnameend},
  \bibinfo{author}{Lukas \surnamestart Burgholzer\surnameend} \&
  \bibinfo{author}{Robert \surnamestart Wille\surnameend}
  (\bibinfo{year}{2023}): \emph{\bibinfo{title}{{{MQT Bench}}: {{Benchmarking
  Software}} and {{Design Automation Tools}} for {{Quantum Computing}}}}.
\newblock {\slshape \bibinfo{journal}{Quantum}} \bibinfo{volume}{7}, p.
  \bibinfo{pages}{1062}, \doi{10.22331/q-2023-07-20-1062}.

\bibitemdeclare{article}{saffman2016quantum}
\bibitem{saffman2016quantum}
\bibinfo{author}{M~\surnamestart Saffman\surnameend} (\bibinfo{year}{2016}):
  \emph{\bibinfo{title}{Quantum computing with atomic qubits and Rydberg
  interactions: progress and challenges}}.
\newblock {\slshape \bibinfo{journal}{Journal of Physics B: Atomic, Molecular
  and Optical Physics}} \bibinfo{volume}{49}(\bibinfo{number}{20}), p.
  \bibinfo{pages}{202001}, \doi{10.1088/0953-4075/49/20/202001}.
\newblock \urlprefix\url{https://dx.doi.org/10.1088/0953-4075/49/20/202001}.

\bibitemdeclare{article}{saffmanQuantumInformationRydberg2010}
\bibitem{saffmanQuantumInformationRydberg2010}
\bibinfo{author}{M.~\surnamestart Saffman\surnameend}, \bibinfo{author}{T.~G.
  \surnamestart Walker\surnameend} \& \bibinfo{author}{K.~\surnamestart
  M{\o}lmer\surnameend} (\bibinfo{year}{2010}): \emph{\bibinfo{title}{Quantum
  Information with {{Rydberg}} Atoms}}.
\newblock {\slshape \bibinfo{journal}{Reviews of Modern Physics}}
  \bibinfo{volume}{82}(\bibinfo{number}{3}), pp. \bibinfo{pages}{2313--2363},
  \doi{10.1103/RevModPhys.82.2313}.

\bibitemdeclare{article}{saffmanQuantumComputingNeutral2019}
\bibitem{saffmanQuantumComputingNeutral2019}
\bibinfo{author}{Mark \surnamestart Saffman\surnameend} (\bibinfo{year}{2019}):
  \emph{\bibinfo{title}{Quantum Computing with Neutral Atoms}}.
\newblock {\slshape \bibinfo{journal}{National Science Review}}
  \bibinfo{volume}{6}(\bibinfo{number}{1}), pp. \bibinfo{pages}{24--25},
  \doi{10.1093/nsr/nwy088}.

\bibitemdeclare{article}{schmidComputationalCapabilitiesCompiler2023}
\bibitem{schmidComputationalCapabilitiesCompiler2023}
\bibinfo{author}{Ludwig \surnamestart Schmid\surnameend},
  \bibinfo{author}{David~F \surnamestart Locher\surnameend},
  \bibinfo{author}{Manuel \surnamestart Rispler\surnameend},
  \bibinfo{author}{Sebastian \surnamestart Blatt\surnameend},
  \bibinfo{author}{Johannes \surnamestart Zeiher\surnameend},
  \bibinfo{author}{Markus \surnamestart Müller\surnameend} \&
  \bibinfo{author}{Robert \surnamestart Wille\surnameend}
  (\bibinfo{year}{2024}): \emph{\bibinfo{title}{Computational capabilities and
  compiler development for neutral atom quantum processors—connecting tool
  developers and hardware experts}}.
\newblock {\slshape \bibinfo{journal}{Quantum Science and Technology}}
  \bibinfo{volume}{9}(\bibinfo{number}{3}), p. \bibinfo{pages}{033001},
  \doi{10.1088/2058-9565/ad33ac}.
\newblock \urlprefix\url{https://dx.doi.org/10.1088/2058-9565/ad33ac}.

\bibitemdeclare{misc}{schmidHybridCircuitMapping2023}
\bibitem{schmidHybridCircuitMapping2023}
\bibinfo{author}{Ludwig \surnamestart Schmid\surnameend},
  \bibinfo{author}{Sunghye \surnamestart Park\surnameend},
  \bibinfo{author}{Seokhyeong \surnamestart Kang\surnameend} \&
  \bibinfo{author}{Robert \surnamestart Wille\surnameend}
  (\bibinfo{year}{2023}): \emph{\bibinfo{title}{Hybrid {{Circuit Mapping}}:
  {{Leveraging}} the {{Full Spectrum}} of {{Computational Capabilities}} of
  {{Neutral Atom Quantum Computers}}}}, \doi{10.48550/arXiv.2311.14164}.
\newblock \eprint{2311.14164}.

\bibitemdeclare{article}{shawMultiensembleMetrologyProgramming2024}
\bibitem{shawMultiensembleMetrologyProgramming2024}
\bibinfo{author}{Adam~L. \surnamestart Shaw\surnameend}, \bibinfo{author}{Ran
  \surnamestart Finkelstein\surnameend}, \bibinfo{author}{Richard Bing-Shiun
  \surnamestart Tsai\surnameend}, \bibinfo{author}{Pascal \surnamestart
  Scholl\surnameend}, \bibinfo{author}{Tai~Hyun \surnamestart Yoon\surnameend},
  \bibinfo{author}{Joonhee \surnamestart Choi\surnameend} \&
  \bibinfo{author}{Manuel \surnamestart Endres\surnameend}
  (\bibinfo{year}{2024}): \emph{\bibinfo{title}{Multi-Ensemble Metrology by
  Programming Local Rotations with Atom Movements}}.
\newblock {\slshape \bibinfo{journal}{Nature Physics}}, pp.
  \bibinfo{pages}{1--7}, \doi{10.1038/s41567-023-02323-w}.

\bibitemdeclare{article}{shende2006synthesis}
\bibitem{shende2006synthesis}
\bibinfo{author}{VV~\surnamestart Shende\surnameend},
  \bibinfo{author}{SS~\surnamestart Bullock\surnameend} \&
  \bibinfo{author}{IL~\surnamestart Markov\surnameend} (\bibinfo{year}{2006}):
  \emph{\bibinfo{title}{Synthesis of quantum-logic circuits}}.
\newblock {\slshape \bibinfo{journal}{IEEE Transactions on Computer-Aided
  Design of Integrated Circuits and Systems}}
  \bibinfo{volume}{25}(\bibinfo{number}{6}), pp. \bibinfo{pages}{1000--1010},
  \doi{10.1109/TCAD.2005.855930}.

\bibitemdeclare{article}{simmons2021}
\bibitem{simmons2021}
\bibinfo{author}{Will \surnamestart Simmons\surnameend} (\bibinfo{year}{2021}):
  \emph{\bibinfo{title}{Relating Measurement Patterns to Circuits via Pauli
  Flow}}.
\newblock {\slshape \bibinfo{journal}{Electronic Proceedings in Theoretical
  Computer Science}} \bibinfo{volume}{343}, p. \bibinfo{pages}{50–101},
  \doi{10.4204/eptcs.343.4}.

\bibitemdeclare{article}{staudacher2022reducing}
\bibitem{staudacher2022reducing}
\bibinfo{author}{Korbinian \surnamestart Staudacher\surnameend},
  \bibinfo{author}{Tobias \surnamestart Guggemos\surnameend},
  \bibinfo{author}{Sophia \surnamestart Grundner-Culemann\surnameend} \&
  \bibinfo{author}{Wolfgang \surnamestart Gehrke\surnameend}
  (\bibinfo{year}{2023}): \emph{\bibinfo{title}{Reducing 2-QuBit Gate Count for
  ZX-Calculus based Quantum Circuit Optimization}}.
\newblock {\slshape \bibinfo{journal}{Electronic Proceedings in Theoretical
  Computer Science}} \bibinfo{volume}{394}, pp. \bibinfo{pages}{29--45},
  \doi{10.4204/EPTCS.394.3}.

\bibitemdeclare{misc}{staudachercode}
\bibitem{staudachercode}
\bibinfo{author}{Korbinian \surnamestart Staudacher\surnameend},
  \bibinfo{author}{Ludwig \surnamestart Schmid\surnameend},
  \bibinfo{author}{Johannes \surnamestart Zeiher\surnameend},
  \bibinfo{author}{Robert \surnamestart Wille\surnameend} \&
  \bibinfo{author}{Dieter \surnamestart Kranzlm{\"u}ller\surnameend}
  (\bibinfo{year}{2024}): \emph{\bibinfo{title}{Multi-Controlled Phase Gate
  Synthesis with {{ZX-}} Calculus}}, \doi{10.5281/zenodo.10730427}.

\bibitemdeclare{article}{tanCompilingQuantumCircuits2023}
\bibitem{tanCompilingQuantumCircuits2023}
\bibinfo{author}{Daniel~Bochen \surnamestart Tan\surnameend},
  \bibinfo{author}{Dolev \surnamestart Bluvstein\surnameend},
  \bibinfo{author}{Mikhail~D. \surnamestart Lukin\surnameend} \&
  \bibinfo{author}{Jason \surnamestart Cong\surnameend} (\bibinfo{year}{2024}):
  \emph{\bibinfo{title}{Compiling {Q}uantum {C}ircuits for {D}ynamically
  {F}ield-{P}rogrammable {N}eutral {A}toms {A}rray {P}rocessors}}.
\newblock {\slshape \bibinfo{journal}{{Quantum}}} \bibinfo{volume}{8}, p.
  \bibinfo{pages}{1281}, \doi{10.22331/q-2024-03-14-1281}.

\bibitemdeclare{misc}{tanDepthOptimalAddressing2D2024}
\bibitem{tanDepthOptimalAddressing2D2024}
\bibinfo{author}{Daniel~Bochen \surnamestart Tan\surnameend},
  \bibinfo{author}{Shuohao \surnamestart Ping\surnameend} \&
  \bibinfo{author}{Jason \surnamestart Cong\surnameend} (\bibinfo{year}{2024}):
  \emph{\bibinfo{title}{Depth-{{Optimal Addressing}} of {{2D Qubit Array}} with
  {{1D Controls Based}} on {{Exact Binary Matrix Factorization}}}}.
\newblock \eprint{2401.13807}.

\bibitemdeclare{article}{vilmart-near-optimal-2018}
\bibitem{vilmart-near-optimal-2018}
\bibinfo{author}{Renaud \surnamestart Vilmart\surnameend}
  (\bibinfo{year}{2019}): \emph{\bibinfo{title}{A Near-Minimal Axiomatisation
  of ZX-Calculus for Pure Qubit Quantum Mechanics}}, pp.
  \bibinfo{pages}{1--10}.
\newblock \doi{10.1109/LICS.2019.8785765}.

\bibitemdeclare{misc}{wangFPQACCompilationFramework2023}
\bibitem{wangFPQACCompilationFramework2023}
\bibinfo{author}{Hanrui \surnamestart Wang\surnameend}, \bibinfo{author}{Pengyu
  \surnamestart Liu\surnameend}, \bibinfo{author}{Bochen \surnamestart
  Tan\surnameend}, \bibinfo{author}{Yilian \surnamestart Liu\surnameend},
  \bibinfo{author}{Jiaqi \surnamestart Gu\surnameend},
  \bibinfo{author}{David~Z. \surnamestart Pan\surnameend},
  \bibinfo{author}{Jason \surnamestart Cong\surnameend}, \bibinfo{author}{Umut
  \surnamestart Acar\surnameend} \& \bibinfo{author}{Song \surnamestart
  Han\surnameend} (\bibinfo{year}{2023}): \emph{\bibinfo{title}{{{FPQA-C}}: {{A
  Compilation Framework}} for {{Field Programmable Qubit Array}}}},
  \doi{10.48550/arXiv.2311.15123}.
\newblock \eprint{2311.15123}.

\bibitemdeclare{misc}{wangQPilotFieldProgrammable2023}
\bibitem{wangQPilotFieldProgrammable2023}
\bibinfo{author}{Hanrui \surnamestart Wang\surnameend}, \bibinfo{author}{Bochen
  \surnamestart Tan\surnameend}, \bibinfo{author}{Pengyu \surnamestart
  Liu\surnameend}, \bibinfo{author}{Yilian \surnamestart Liu\surnameend},
  \bibinfo{author}{Jiaqi \surnamestart Gu\surnameend}, \bibinfo{author}{Jason
  \surnamestart Cong\surnameend} \& \bibinfo{author}{Song \surnamestart
  Han\surnameend} (\bibinfo{year}{2023}): \emph{\bibinfo{title}{Q-{{Pilot}}:
  {{Field Programmable Quantum Array Compilation}} with {{Flying Ancillas}}}},
  \doi{10.48550/arXiv.2311.16190}.
\newblock \eprint{2311.16190}.

\bibitemdeclare{article}{van2020zx}
\bibitem{van2020zx}
\bibinfo{author}{John \surnamestart van~de Wetering\surnameend}
  (\bibinfo{year}{2020}): \emph{\bibinfo{title}{ZX-calculus for the working
  quantum computer scientist}}.
\newblock {\slshape \bibinfo{journal}{arXiv preprint arXiv:2012.13966}}.
\newblock \urlprefix\url{https://doi.org/10.48550/arXiv.2012.13966}.

\bibitemdeclare{article}{zhang2023characterization}
\bibitem{zhang2023characterization}
\bibinfo{author}{Shihao \surnamestart Zhang\surnameend}, \bibinfo{author}{Junda
  \surnamestart Wu\surnameend} \& \bibinfo{author}{Lvzhou \surnamestart
  Li\surnameend} (\bibinfo{year}{2023}):
  \emph{\bibinfo{title}{Characterization, synthesis, and optimization of
  quantum circuits over multiple-control $\mathit{Z}$-rotation gates: A
  systematic study}}.
\newblock {\slshape \bibinfo{journal}{Phys. Rev. A}} \bibinfo{volume}{108}, p.
  \bibinfo{pages}{022603}, \doi{10.1103/PhysRevA.108.022603}.
\newblock \urlprefix\url{https://link.aps.org/doi/10.1103/PhysRevA.108.022603}.

\end{thebibliography}

\appendix
\newpage
\section{Gadget insertion}
\label{sec:gadget-insertion}
Inserting $YZ$ measurements on the outputs preserves gflow:
\begin{corollary}
	Let $(g,\prec)$ be a gflow for $(G,I,O,\lambda)$ and let $W\subseteq O$. Then $(G',I,O,\lambda')$, where $G' = (V',E')$ with $V'=V\cup \{x\}$, $\lambda'(x) = YZ$ and $E'=E\cup \{(x,w) | w\in W\}$ has a gflow $(g',\prec')$ with following properties:
	\begin{itemize}
		\item $g'(x) = \{x\}$,
		\item $\forall v \in V$: $g'(v) = g(v)$,
		\item $\prec'$ is the transitive closure of $\prec \cup \{(x,v)|v\in O\}\cup \{(v,x)|v\in V\backslash O\}$.
	\end{itemize}
	\begin{proof}
		The only new correction set in $g'$ is $g'(x) = \{x\}$, for all other vertices, it is the same as in $g$. Therefore, all conditions except $(g2)$ are trivially satisfied. For $(g2)$, we need to distinguish two cases for the new vertex $x$ and vertices $v\in V\backslash\{x\}$: \begin{itemize}
			\item $x\in Odd(g(v))$: By definition, $x$ is the last element in the partial order $\prec'$ of all non-outputs, thus $(g2)$ holds. 
			\item $v\in Odd(g(x))$: $x\prec v$ holds, because $g(x)=\{x\}$ and $Odd(\{x\})$ only contains outputs which we chose to be after $x$ in the partial order.
		\end{itemize}
		Note that $x\notin Odd(g(x))$ by definition. 
	\end{proof}

\end{corollary}

\section{Alternative proof of \Cref{theo:mcp-gates}}
\label{sec:proof-theorem-x}

\begin{figure}[h]
  \centering
  \includegraphics[width=\textwidth]{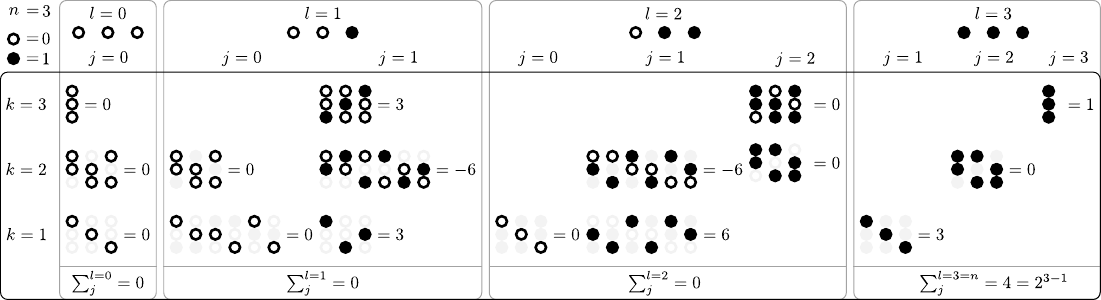}
  \caption{Illustration of the possible combinations and their contribution in \Cref{eq:lemma-n-ueq-l} for $n=3$. For each possible value of $l=0,\, \dots \, , 3$ the combinations for the possible $k \leq n$ and $j \leq l$ are illustrated as three-circle circles, and their contribution to the sum is computed. The final row shows that the sums of the contributions fulfill the condition of \Cref{lemma:multi-controlled-phase-gate}. }
  \label{fig:combinatorics}
\end{figure}

This section provides an alternative, combinatorial proof of \Cref{theo:mcp-gates} instead of using the graphical approach discussed in the main part of the work.
The overarching idea is to find a closed formula for the unitary defined by the ZX illustration by summing the corresponding phase contributions and showing that this corresponds to $\mathrm{diag}(1, 1, \dots , e^{i\alpha})$ for an arbitrary number of qubits $n$.

To find a closed formula for \Cref{theo:mcp-gates} consider the definition of a single phase gadget and its corresponding unitary action on the \mbox{$n$-qubit} basis states according to~\cite{kissingerReducingNumberNonClifford2020}:
\begin{equation}
  \label{eq:phase-gadget}
\begin{tikzpicture}
	\begin{pgfonlayer}{nodelayer}
		\node [style=Z dot] (8) at (-3, 1) {};
		\node [style=Z dot] (9) at (-3, -0.5) {};
		\node [style=none] (10) at (-1.5, 1) {};
		\node [style=none] (11) at (-1.5, -0.5) {};
		\node [style=Z dot] (13) at (-4.5, 2.25) {};
		\node [style=Z phase dot] (14) at (-4.5, 3.5) {$\alpha$};
		\node [style=none] (23) at (-2.25, -1.5) {$\vdots$};
		\node [style=Z dot] (24) at (-3, -2.5) {};
		\node [style=none] (25) at (-1.5, -2.5) {};
	\end{pgfonlayer}
	\begin{pgfonlayer}{edgelayer}
		\draw (8) to (10.center);
		\draw (9) to (11.center);
		\draw [style=hadamard edge] (14) to (13);
		\draw [style=hadamard edge] (13) to (8);
		\draw [style=hadamard edge] (9) to (13);
		\draw (24) to (25.center);
		\draw [style=hadamard edge](13) to (24);
	\end{pgfonlayer}
\end{tikzpicture} \qquad  U \ket{x_{1},\, \dots,\, x_{n}} = e^{i \alpha (x_{1} \bigoplus\, \dots\, \bigoplus x_{n})} \ket{x_{1},\, \dots,\, x_{n}} \quad ,
\end{equation}
where the binary $x_{i} \in \{0,1\},\, i=1,\, \dots,\, n$ label the basis states and $\oplus$ is the binary sum modulo two, i.e. XOR.
Note that all phase gadgets of \Cref{theo:mcp-gates} can be written in such a way, summing only the $x_{i}$ connected by Hadamard wires to the phase.
The single qubit phases give an additional contribution of $e^{i \alpha \, x_{j}} \ket{x_{1},\, \dots,\, x_{n}}$ for each applied qubit~$j$, corresponding to a single-qubit phase gadget.

As the binary sum does not depend on the order of the $x_{i}$ but only on their value, we introduce the following notation, where we assume that $l$ entries in the sum are non-zero, resulting in a non-zero-sum whenever $l$ is odd:
\begin{equation}
  \label{eq:sum-delta}
  x_{1} \oplus \, \dots\, \oplus x_{n} \quad = \quad \frac{(-\!1)^{l + 1} + 1}{2} \quad = \quad \mathrm{mod}_{2}(l) \quad = \quad \begin{cases}
	1 &, l \text{ is odd} \\
	0 &, l \text{ is even }
  \end{cases} \, .
\end{equation}

Considering again \Cref{eq:extract-phases} one can see that there are two contributions to the total accumulated phase.
First, for each qubit, a single-qubit phase $\alpha$ is added.
Second, for each possible combination of length $k$ of all the $x_{i}$, there is a phase gadget with phase $(-\!1)^{k+1}\,\alpha$.
For the $\mathrm{C}_{1}\mathrm{P}(2\alpha)$ gate, this reduces to a single $k=2$ phase gadget of phase $-\alpha$.
For the $\mathrm{C}_{2}\mathrm{P}(4\alpha)$ gate, on the other hand, there are $\binom{3}{2}=3$ phase gadgets of size $k=2$ and angle $-\alpha$ and a single ($\binom{3}{3}=1$) $k=3$ gadget with angle $\alpha$.
In general, for a $n$ qubit gate, there are $\binom{n}{k}$ combinations for phase gadgets of size $k=1,\, \dots,\, n$.
A combination contributes to the total phase if the number $l$ of non-zero entries in the direct sum of \Cref{eq:phase-gadget} is odd, resulting in an additional phase $\pm \alpha$.
Otherwise, the combination does not contribute to the phase.
To express the number of possible combinations depending on $l$, the $\binom{n}{k}$ combinations for a length $k$ can also be expressed as choosing $j$ variables from the $l$ one-valued variables and choosing $k-j$ variables from the $n-l$ zero-valued variables and summing over all possible $j$:
\begin{equation}
  \label{eq:vandermonde}
  \binom{n}{k} = \sum_{j=0}^{k} \binom{l}{j}\binom{n-l}{k-j} \, .
\end{equation}
This relation is known as the \textit{Vandermonde identity}. An illustration of the possible combinations depending on $k$ and $l$ is shown in \Cref{fig:combinatorics} for the simple case $n=3$.

Multiplying the unitaries of all these phase gadgets corresponds to summing the accumulated phases with the appropriate sign, converting the problem of \Cref{theo:mcp-gates} into a summation of the appropriate phases with a corresponding sign.
For the total structure to represent a multi-controlled phase gate, the phases have to vanish for all possible basis states $\ket{x_{1},\, \dots,\, x_{n}}$ except for $\ket{1, \, \dots\, 1}$ where they have to sum to $2^{n-1}\alpha$.

Based on these considerations, an equivalent statement of \Cref{theo:mcp-gates} can be formulated, dropping the illustrations of the ZX-calculus and formulating the multi-controlled phase gate extraction as a purely combinatorial problem, focusing on the accumulated phase.
\Cref{theo:mcp-gates} then directly follows from this Lemma based on the considerations above and using $e^{i\alpha \cdot 0} = 1$.
\begin{lemma}[Multi-controlled phase gate]
  \label{lemma:multi-controlled-phase-gate}
  For $n$ binary variables $x_{1}\, \dots\, x_{n}$ of which $l$ is non-zero, summing the modulo two sum over all possible combinations of length $k$ with sign $(-1)^{k+1}$, it holds:

  \begin{equation}
	\label{eq:lemma-n-ueq-l}
	\sum_{k=1}^{n} (-\!1)^{k+1} \sum_{j=0}^{\mathrm{min}(k,l)} \binom{l}{j} \binom{n-l}{k-j} \mathrm{mod}_{2}(j) = \begin{cases}
	  2^{n-1} &, \text{ if } $n=l$ \\
	  0 &, \text{ else}
	\end{cases} \quad .
  \end{equation}
  Where the $\mathrm{min}(k,l)$ results from the fact that the number of ones in the current combination $j$ cannot be larger than the length $k$ of the combination, nor the total number of ones $l$ available.
\end{lemma}

\begin{proof}
  The proof is two-fold. First, the $n=l$ case is shown explicitly, while the case $n \neq l$ is shown by induction in both variables $n$ and $l$.
  Also, note that
  \begin{equation}
	\label{eq:mun-sum-id}
	\sum_{j=0}^{\mathrm{min}(k,l)} \binom{l}{j} \binom{n-l}{k-j} = \sum_{j=0}^{k} \binom{l}{j} \binom{n-l}{k-j}  =  \sum_{j=0}^{l} \binom{l}{j} \binom{n-l}{k-j}
  \end{equation}
  as for $k>l$ the first binomial coefficient vanishes in all additional cases, and for $l>k$ the second, as $k-j<0$ in these cases.
  This also becomes clear from the illustration in \Cref{fig:combinatorics} where the vanishing combinations are either non-existent or only contain zero entries.
  These identities are used multiple times in the following proof.

  \subsubsection*{Case $n=l$:}
\label{sec:n-eq-l}
If $n=l$ all variables are one and, therefore, $\mathrm{min}(k,l)=k$. Furthermore, the second binomial coefficient is non-zero only in the $j=k$ case, where it equals 1. This results in
\begin{equation*}
  \sum_{k=1}^{n}(-\!1)^{k+1} \binom{n}{k} \frac{(-\!1)^{k+1}+1}{2} = \frac{1}{2} \left[ \sum_{k=1}^{n}\binom{n}{k} (-\!1)^{k+1} + \sum_{k=1}^{n}\binom{n}{k} \right] = \frac{1}{2} \left[ 1 + 2^{n} -1 \right] = 2^{n-1} \, ,
\end{equation*}
using the regular and the alternating binomial sum, directly showing the first part of \Cref{lemma:multi-controlled-phase-gate}.

  \subsubsection*{Case $n \neq l$:}
\label{sec:n-neq-l}
Proving \Cref{eq:lemma-n-ueq-l} for arbitrary $n$ and $l < n$ is done by induction.
Therefore, showing the term to be zero for $l=0$ and arbitrary $n$ as the base case and then performing the induction step both in $n$ and in $l$.

\textit{Base case $l=0,\, n$:} In this case the second sum reduces to the $j=0$ case, trivially giving zero, independetly for all $n$. In other words, as all variables are zero, the sum in \Cref{eq:phase-gadget} always gives zero.

\textit{Induction step $n \rightarrow n+1$:} Inserting this step into \Cref{eq:phase-gadget} and using the recurrence relation of the binomial coefficient $\binom{n+1}{k} = \binom{n}{k-1} + \binom{n}{k}$ and the abbreviation $\xi \coloneqq \mathrm{mod}_{2}(j)$ one gets
\begingroup
\allowdisplaybreaks
  \begin{align*}
	\sum_{k=1}^{n+1}& (-\!1)^{k+1} \sum_{j=0}^{l} \binom{l}{j}\binom{n - l}{k-j-1} \xi + (-\!1)^{n} + \underbrace{\sum_{k=1}^{n} (-\!1)^{k+1} \sum_{j=0}^{l} \binom{l}{j}\binom{n - l}{k-j} \xi}_{=0 \text{ (Base case)}} \\
	& \qquad + \cancel{(-\!1)^{n} \sum_{j=0}^{l} \binom{l}{j}\binom{n - l}{n + 1-j} \xi} \\
								  &=\sum_{k=0}^{n} (-\!1)^{k} \sum_{j=0}^{l} \binom{l}{j}\binom{n - l}{k-j} \frac{(-\!1)^{j} + 1}{2} =\sum_{k=0}^{n} (-\!1)^{k+1} \sum_{j=0}^{k} \binom{l}{j}\binom{n - l}{k-j} \frac{(-\!1)^{j+1} + 1 - 2}{2} \\
								  & = \underbrace{(-\!1)^{1} \binom{l}{0}\binom{0-l}{0-0} \frac{-\!2}{2} \vphantom{\sum_{k=1}^{n} (-\!1)^{k+1} \sum_{j=0}^{l} \binom{l}{j}\binom{n - l}{k-j} \xi}}_{k=0 \text{ case }} + \underbrace{\sum_{k=1}^{n} (-\!1)^{k+1} \sum_{j=0}^{l} \binom{l}{j}\binom{n - l}{k-j} \xi}_{=0 \text{ (Base case)}} + \sum_{k=0}^{n} (-\!1)^{k+1} \underbrace{\sum_{j=0}^{k} \binom{l}{j}\binom{n - l}{k-j}}_{\text{Vandermonde}}  \\
	& = 1 + 0 - 1 = 0 \quad ,
  \end{align*}
 \endgroup
writing the $n+1$ term separately to recover the base case.
The term in the second line vanishes due to the lower part of the binomial coefficient always being larger than the top part.
Going to the third line, an index shift in $k$ is performed, and then using \Cref{eq:mun-sum-id} with an additional reinserted $(-\!1)$ factor to recover the original form of $\xi$.
Separating the $k=0$ case and the additional introduced $-\!2$ term, the base case can be inserted again, resulting in zero after using again the Vandermonde identity and the alternating binomial sum.
In a similar fashion, also the induction step in $l$ can be shown.

\textit{Induction step $l \rightarrow l+1$:} With the base case for arbitrary $n$ and the corresponding induction step in $n$, one can, in the following, assume the base case to be true for arbitrary $n$, in particular for $n' = n - 1$. Performing the induction step in $l\rightarrow l+1$ and again using the recurrence relation, one gets
\begingroup
\allowdisplaybreaks
  \begin{align*}
	\sum_{k=1}^{n}& (-\!1)^{k+1} \sum_{j=0}^{l+1} \binom{l+1}{j} \binom{n - (l+1)}{k-j} \xi \\
				  & = \sum_{k=1}^{n} (-\!1)^{k+1} \sum_{j=0}^{l+1} \binom{l}{j-1} \binom{n'-l}{k-j} \xi + \underbrace{\sum_{k=1}^{n'} (-\!1)^{k+1} \sum_{j=0}^{l} \binom{l}{j} \binom{n'-l}{k-j} \xi}_{=0 \text{ (Base case for $n'$)}  } \\
	\\
				  & \hspace{1cm} + \cancel{\underbrace{\sum_{k=1}^{n} (-\!1)^{k+1} \binom{l}{l+1} \binom{n'-l}{k-l -1}}_{j=l+1 \text{ case}}} + \cancel{\underbrace{(-\!1)^{n+1} \sum_{j=0}^{l} \binom{l}{j} \binom{n'-l}{n-j} \xi}_{k=n \text{ case }}} \\
				  & = \sum_{k=1}^{n} (-\!1)^{k+1} \sum_{j=-1}^{l} \binom{l}{j} \binom{n'-l}{k-j - 1} \frac{(-\!1)^{j} + 1}{2} \\
				  & = \sum_{k=0}^{n'} (-\!1)^{k} \sum_{j=0}^{l} \binom{l}{j} \binom{n'-l}{k-j} \frac{(-\!1)^{j+1} + 1 }{2} \\
	& = 0 \quad ,
  \end{align*}
  \endgroup
  with the base case for $n'$ used in the second line and the additional terms vanishing because either the first or the second binomial coefficient is zero.
  In the following two lines, first, an index shift in $j$ is performed, and then, secondly, in $k$.
  The result is the same formula as in the previous calculation, just for $n'$ and therefore, also vanishes.

This concludes the induction step in $l$, concluding also the $n\neq l$ case of \Cref{eq:lemma-n-ueq-l} and therefore proving \Cref{lemma:multi-controlled-phase-gate}.

\end{proof}

\end{document}